\documentclass[a4paper,UKenglish,cleveref, autoref, thm-restate]{lipics-v2021}

\nolinenumbers

\usepackage[utf8]{inputenc}
\usepackage[T1]{fontenc} 

\usepackage{mathtools}               	

\usepackage{url}
\usepackage{setspace}
\usepackage{ dsfont }
\usepackage{color}
\usepackage{float}

\usepackage{amsmath} % algorithms
\usepackage[linesnumbered,ruled]{algorithm2e} % algorithms
\usepackage{amsthm} % proofs

\usepackage{apxproof}

\newtheoremrep{lemma}[theorem]{Lemma}
\newtheoremrep{proposition}[theorem]{Proposition}

\Copyright{Manuel Lafond and Francis Sarrazin}

\newcommand{\M}{\mathcal{M}}
\newcommand{\C}{\mathcal{C}}

\newcommand{\F}{\mathcal{F}}

\newcommand{\R}{\mathcal{R}}
\newcommand{\one}{\mathds{1}}
\renewcommand{\H}{\mathcal{H}}
\newcommand{\HH}{\mathcal{H}}
\newcommand{\HHF}{$\mathcal{H}$-\textit{Free}}
\newcommand{\HF}{\H\textit{-free}}
\newcommand{\HHFE}{\textsc{$\H$-free Editing($\sigma$)}}
\newcommand{\HHFD}{\textsc{$\H$-free Deletion}}
\newcommand{\HHFI}{\textsc{$\H$-free Completion}}

\newcommand{\op}{\Delta}

\newcommand{\E}{\mathcal{E}}

\title{An $O^*((2 + \epsilon)^k)$ Time Algorithm for Cograph Deletion Using Unavoidable Subgraphs in Large Prime Graphs}

\titlerunning{An $O^*((2 + \epsilon)^k)$ Time Algorithm for Cograph Deletion}

\author{Manuel Lafond}{Université de Sherbrooke}{manuel.lafond@USherbrooke.ca}{}{}
\author{Francis Sarrazin}{Université de Sherbrooke}{}{}{}

\authorrunning{M. Lafond and F. Sarrazin}


\Copyright{Manuel Lafond and Francis Sarrazin}


\begin{document}

\hideLIPIcs

\maketitle

\begin{abstract}
    We study the parameterized complexity of the Cograph Deletion problem, which asks whether at most $k$ edges can be deleted from a graph to make it $P_4$-free.  This is a well-known graph modification problem with applications in computational biology and social network analysis.  
    All current parameterized algorithms use a similar strategy: they identify an induced $P_4$ and branch on its local structure. The current best such algorithm achieves running time $O^*(2.303^k)$, but relies on a large number of intricate branching cases found by automated search, making further improvements along this line difficult.  

    We therefore introduce a new approach based on modular decompositions.  We solve each module and the quotient graph independently, with the latter being the core problem.  This reduces the problem to solving on a prime graph, in which all modules are trivial.  We then use a characterization of Chudnovsky et al. stating that any large enough prime graph has one of seven structures as an induced subgraph.  These all have linearly many $P_4$s, which we use in a  recursive branch tree algorithm to achieve running time $O^*((2 + \epsilon)^k)$ for any $\epsilon > 0$.  

 Although the hidden polynomial factors can be very large, our result shows that the exponential base of the complexity can be made arbitrarily close to 2, suggesting a new direction for improving graph modification algorithms. More broadly, this appears to be the first algorithmic application of the prime graph characterization. To support further future applications, we characterize the graph classes $\H$ for which the $\H$-free editing problem admits a similar reduction to prime graphs.
\end{abstract}

\section{Introduction}

Modifying a graph to satisfy a desired property is a fundamental problem in computer science.  Notable examples include clustering applications where edges are added/removed to obtain disjoint cliques~\cite{bocker2012golden,cohen2024combinatorial,lafond_et_al:LIPIcs.STACS.2025.64},  or computing the treewidth of a graph by adding edges to make it chordal~\cite{bouchitte2001treewidth}.  Famously, for any graph class characterized by a finite set of forbidden induced subgraphs, deciding whether it is enough to modify at most $k$ edges from a graph to make it belong to the class is fixed-parameter tractable (FPT)~\cite{cai1996fixed}.  This means that there is an algorithm whose running time has the form $O(f(k) poly(n)) = O^*(f(k))$\footnote{Unless stated otherwise, $n$ and $m$ denote the number of vertices and edges of the given graph, respectively; $poly(n)$ represents any polynomial function of $n$; $O^*$ hides polynomial factors.}.  In practice, we aim to make $f(k)$ as small as possible, and recent works have studied  this for a variety of hereditary graph classes (see~\cite{crespelle2023survey} for a recent survey).

In this work, we focus on the \textsc{Cograph Deletion} problem, which asks, given a graph $G$ and an integer $k$, whether it is possible to delete at most $k$ edges from $G$ to obtain a cograph.  
Cographs have multiple equivalent definitions, the simplest one being that they are exactly the graphs with no $P_4$, the induced path with four vertices, and have several applications.  They have been widely studied in computational biology, since every cograph can be represented as a so-called cotree which represents an evolutionary scenario~\cite{hellmuth2013orthology,lafond2014orthology,hellmuth2024resolving}, and errors in the data are corrected using edge modifications~\cite{hellmuth2015phylogenomics,lafond2016link}.   
Cographs also have applications in social network studies, as they can be used to represent overlapping communities~\cite{jia2015defining}.
The variant \textsc{Cograph Completion}, where only edge additions are allowed, is equivalent to the deletion variant since $P_4$ is self-complementary, and \textsc{Cograph Editing} refers to the variant where both deletions/insertions are allowed.

 There is a simple $O(3^k (n + m))$ time algorithm for \textsc{Cograph Deletion}: find a $P_4$, then recursively branch into the three ways to delete an edge to remove this $P_4$, with each branch reducing the parameter by $1$.  This gives a recursion tree of arity $3$ and depth at most $k$, and the complexity follows from the fact that one can find a $P_4$, if any, in linear time~\cite{habib2005simple}.
 This complexity was subsequently improved using more sophisticated branching rules.  Nastos and Gao~\cite{nastos2012bounded} showed that either a graph is $P_4$-sparse (every induced subgraph on five vertices has at most one $P_4$), or it has one of seven induced subgraphs on five vertices that can be branched on efficiently.  This gave them a time of $O^*(2.562^k)$.  The same technique was used by Liu et al.~\cite{liu2012complexity} to solve \textsc{Cograph Editing} in time $O^*(4.612^k)$.  By diving deeper into the structure of $P_4$-sparse graphs, Tsur~\cite{Tsur:[3]} improved this to $O^*(4.329^k)$ and, for the \textsc{Cograph Deletion} problem, obtained time $O^*(2.303^k)$.  
 
 The latter is the best upper bound at the moment.  Achieving this complexity is a non-trivial feat and requires several combinatorial ideas, but unfortunately the number of branching cases becomes unmanageable by hand.  Tsur's algorithm relies on a python script to infer a set of branching rules with optimal branching factor, making it difficult to verify and implement.  %Moreover, all the previous algorithms are based on $P_4$-sparsity and use a similar strategy: find a $P_4$ and look at the small induced subgraphs ``around'' this $P_4$, and use the local structure to find more clever branching rules.
This suggests an inherent limitation of purely local strategies, and that further improvements demand fundamentally new ideas.

We introduce an approach that departs from previous local branching techniques by exploiting global structure in large prime graphs, obtained from modular decompositions.  A module $M$ in a graph $G$ is a subset of vertices that all have the same neighbors outside of $M$, and it is well-known that graphs that are connected and complement-connected admit a vertex partition $\M^* = \{M_1, \ldots, M_p\}$ into maximal strong modules (see next section for definitions).  The quotient $G / \M^*$ is the result of contracting each module $M_i$ into a single vertex.  Previous results on the cograph editing problem, namely those of Guillemot et al.~\cite{guillemot2013non}
 and of~Fritz et al.~\cite{fritz2020cograph}, have established that 
 there is always an optimal deletion set that preserves the modules of the graph, which then suggests the following strategy: (1) compute a modular decomposition $\M^*$; (2) recursively solve each induced subgraph $G[M_i]$ independently; (3) solve the quotient $G / \M^*$ (with its vertices weighted by the size of the contracted modules); (4) expand the modifications made on $G / \M^*$ to the original graph $G$.

The hard part is step (3), solving the quotient $G / \M^*$, which lets us assume that the graph is \emph{prime}, that is, whose modules  have either one or all of the vertices.  Recently, Chudnovsky et al.~\cite{Chudnovsky:[1]} provided an unavoidable subgraph characterization of prime graphs (also re-proved a bit later in~\cite{malliaris2018unavoidable}).  They showed that if $G$ is large enough and prime, then it has a large subgraph that induces one of seven structures or their complements.  These are described later on, but one concrete example is that $G$ must contain either a subdivided $K_{1, c}$, or one of the six other structures (a subdivided $K_{1, c}$ is a star graph with $c$ leaves in which we subdivide each edge).  Here, $c$ depends on the size of the given prime graph and can be arbitrarily large.
 The $K_{1, c}$ has \emph{many} induced $P_4$s, and it is possible to devise branching rules that achieve a time complexity $O^*(2 + o(1))^k$, where $o(1)$ tends to $0$ as $c$ grows.  We show that all seven unavoidable structures allow the same type of efficient branching, without requiring automated branching analysis.  
Thus, to solve $G / \M^*$ in time $O^*((2 + \epsilon)^k)$, either $G / \M^*$ is small enough in terms of the desired $\epsilon$ and we use brute-force, or we find a large unavoidable subgraph to branch on.   
 To our knowledge, this is the first algorithmic application of the unavoidable subgraph characterization of Chudnovsky et al. for prime graphs.

This does come with an inherent limitation: the structural characterization we rely on only applies beyond a very large threshold size that depends on $\epsilon$, derived from Ramsey-type arguments.  This is not a problem in theory, since this threshold is just a constant, but in practice this means that ``smaller'' graphs of realistic size cannot benefit from our novel machinery. 
Nonetheless, our result establishes a new structural approach to graph modification problems. It shows that global properties of prime graphs, often studied mainly in structural graph theory, can be leveraged to control the exponential dependence on $k$. 
%This yields the first framework in which the base of the exponent can be driven arbitrarily close to $2$, and suggests that further progress may come from improved structural characterizations rather than more refined local branching.
Our result should thus be viewed primarily as a theoretical advance on the achievable complexity, rather than a practical algorithm.
Let us also note that future improvements on the upper bounds of the size of ``large enough'' prime graphs could immediately translate into improvements in the applicability of our algorithm.

%One drawback of this approach is that the value of the constant $c$ controls $\epsilon$ while being part of the polynomial factors involved.  That is, the algorithm needs to spend time  $n^{O(c)}$ to find one of the seven forbidden substructures used for branching, and there is thus an inherent tradeoff between the exponential and polynomial factors.  Moreover, the prime graphs need to be ``large enough'' to guarantee that the desired structure exists, where large enough means at least of a size $C$ that involves generalized Ramsey numbers with respect to $c$.  Although this galactic constant $C$ does not take any part in the complexity, graphs of size smaller than this $C$ cannot make use of our advanced machinery.  Nonetheless, we believe that the added theoretical value of our algorithm will open the door to new and possibly practical algorithmic avenues in the future, either by improving our algorithms or our knowledge of prime graphs.

 The astute reader may notice that the general strategy of solving each $G[M_i]$ and $G / \M^*$ independently could be applicable beyond \textsc{Cograph Deletion}.  This begs the question: for which graph property exactly does this general strategy work?  
 In other words, for which editing problems is it sufficient to solve on a prime graph? 
 To answer this, we present our algorithm in the context of the more general $\HF$ problem, where $\H$ is a (possibly infinite) 
 set of graphs and we must remove all induced copies of graphs in $\H$ using edge operations.  We formalize the aforementioned generic strategy and characterize exactly for which hereditary class the algorithm works: it is correct if, and \emph{only if}, every minimal  graph in $\H$ is prime (where minimal means that none of its proper induced subgraphs is in $\H$).  

 We say that a recursive branching algorithm has branching factor $b$ if it leads to a recursion tree with $O(b^k)$ nodes.  We summarize our findings as follows:
 \begin{enumerate}
     \item 
     Let $\H$ be a possibly infinite set of graphs in which every minimal graph is prime, such that one can check in polynomial time whether a graph $G$ is $\HF$.  Suppose that for the $\HF$ editing problem, where edge deletions and/or insertions are allowed, there is a polynomial-time branching algorithm that, \emph{on prime graphs}, achieves branching factor $b$.  Then the $\HF$-editing problem can be solved in time $O^*(b^k)$ on any graph.

     %If $\H$ is infinite but has some non-prime minimal graph, there exist arbitrarily large graphs that are not $\HF$, but for which the quotient graph of its modular decomposition and also each module itself, are $\HF$.  Therefore this strategy 

     \item 
     For the \textsc{Cograph Deletion} problem, for any fixed real number $\epsilon > 0$, there is a polynomial-time branching algorithm that, on prime graphs, achieves branching factor $b = 2 + \epsilon$.  Consequently, \textsc{Cograph Deletion} can be solved in time $O^*((2 + \epsilon)^k)$.
 \end{enumerate}

We also show that whenever some minimal graph of $\H$ is \emph{not} prime, we can find instances on which the modular decomposition strategy fails.
 This uses a novel graph-theoretical result stating that if $\H$ has some non-prime minimal graph, there exist arbitrarily large graphs that are not $\HF$, but for which the quotient graph of its modular decomposition and also each module itself, are $\HF$ (assuming $\HF$ is an infinite graph class). 
 This clarifies which $\HF$ modification problems can make use of the prime graph characterizations.  
 %One notable example on which the strategy is applicable is the \textsc{Cluster Editing} problem, which is equivalent to \textsc{$\{P_3\}$-free Editing}.  The unavoidable subgraphs of large prime graphs also contain a large number of $P_3$s, opening the door to improvements over the $O^*(1.618^k)$ time achieved in~\cite{bocker2012golden}, which has stood for more than a decade.
 In particular, our techniques could also be used to improve the complexity of the \textsc{Cograph Editing} problem where both insertions/deletions are allowed.  The current best algorithm from Tsur achieves $O^*(4.329^k)$, and perhaps our approach can lower this to $O^*((4 + \epsilon)^k)$, albeit refinements will be required to avoid an explosion in the number of cases to look at.  Another good candidate for future improvements is the \textsc{Bicluster Editing} problem, which further restricts the prime graphs to be bipartite~\cite{tsur2023faster,lafond2024improved}.
 In the longer term, novel branching ideas or novel prime graph characterizations may then lead to improvements of other $\HF$-editing problems when $\H$ contains prime graphs --- these graphs are not yet fully understood, and our techniques have the potential to leverage any future knowledge on prime graphs we may acquire.

 \medskip 

 \noindent 
 \textbf{Related work.}  
 Aside from the FPT algorithms mentioned above, it was shown by Guillemot et al.~\cite{guillemot2013non} that the cograph edge modification variants admit a kernel with $O(k^3)$ vertices.  The authors also use modular decompositions for this, but not the properties of prime graphs.
 This was improved recently  to a $O(k^2 \log k)$ vertex-kernel by Crespelle et al.~\cite{crespelle2024quasi}, and it remains open to find ways to remove the $\log k$ factor.  
 It is known that \textsc{Cograph Deletion} admits no subexponential time algorithm in $k$ unless the ETH fails~\cite{drange2015exploring}.  However, we are not aware of a SETH lower bound for the problem: it could be that $O^*(2^k)$ is best possible under the hypothesis, making our result almost tight, but current known reductions do not appear to imply this.  
  In terms of approximation, it is to our knowledge widely open whether any of the cograph modification variants admit a constant factor approximation algorithm (note that~\cite{bliznets2018hardness} exclude this for many $\HF$ editing problems, but not cographs).  The best results we are aware of are a $O(d)$-approximation for graphs of maximum degree $d$~\cite{natanzon2001complexity}, and a no-constant factor approximation for \emph{edge-weighted} cograph editing under the Unique Games Conjecture~\cite{dondi2017approximating}. 
 Let us finally mention that~\cite{liu2015edge} also use modular decompositions to devise branching algorithms for other edge-deletion problems, namely for chain graphs and trivially perfect graphs, but do not make use of prime graph characterizations.

% All graphs discussed here are either simple unless stated otherwise. Chudnovsky et al. defined in 2015 a set of subgraphs with a constant growth factor $c$ that are unavoidable in any prime graph of order larger than a constant $C=f(c)$, where $f$ is a function of \Ramsey{} \cite{Chudnovsky:[1]}. M. Malliaris and C. Terry confirmed and improved the bounds \cite{Malliaris:[5]}. We provide a generic recursive \HHFE{} algorithm and specific \branching{} and \solver{} functions that together solve \cographDeletion{} with at most $k$ deletions for any sufficiently large graph by editing unavoidable induced subgraphs with a growth factor of $c\geq 6$. Although this paper does not stray from the boundaries of the $P4$-\textit{Free} class and \cographDeletion{}, the genericity of our approach allows investigation of editing problems within all \HHF{} families, for any $\HH$ where all minimal graphs are prime.
 
%------------------------------------------------------------------------------
%	GENERAL DEFINITIONS
%------------------------------------------------------------------------------

\section{Preliminary notions}

The graphs in this work are simple, loopless and \emph{vertex-weighted}.  
Such a graph is denoted $G = (V, E, \omega)$, where $V = V(G)$ is the vertex set, $E = E(G)$ is the edge set, and  $\omega : V \rightarrow \mathbb{N}$ assigns a positive integer weight to each vertex. 
If a weighing of a graph $G$ is not specified, we assume that $\omega(u) = 1$ for all $u \in V$. 
We extend $\omega$ to a vertex pair $uv$ and we define $\omega(uv) = \omega(u) \cdot \omega(v)$.  Roughly speaking, vertex weights are used to represent the number of vertices in a module that got contracted into a single vertex when computing quotient graphs.

The complement of $G$, denoted $\overline G$, has $V(\overline{G}) = V(G)$ and two distinct vertices adjacent if and only if they are not adjacent in $G$ (with vertex weights unchanged). A graph $G$ is \emph{connected} if there exists a path between every pair of vertices. It is \emph{complement-connected} if $\overline{G}$ is connected. A \emph{component} of a graph is a connected subgraph of $G$ that is maximal, with respect to the number of vertices and edges. 
The subgraph of $G$ induced by $S \subseteq V(G)$, denoted $G[S]$, is the graph with vertex set $S$ in which $uv$ is an edge if and only if $u,v \in S$ and $uv \in E(G)$.

\paragraph*{Modules and modular decompositions.}
A \emph{module} $M$ of $G$ is a subset of $V(G)$ such that every vertex of $M$ has the same set of neighbors in $V(G) \setminus M$. Note, $V(G)$ is always a module.  A module $M$ is \emph{trivial} if $M=V(G)$ or $|M|=1$, and \emph{maximal} if $M \neq V(G)$ and $M$ is not a proper subset of any module other than $V(G)$. A module is \emph{strong} if it does not overlap with any other module, that is, for any other module $M'$ either $M \cap M' = \emptyset$ or one module is contained in the other.  
A graph is \emph{prime} if all of its modules are trivial\footnote{Note that some works require at least four vertices for a graph to be prime, for instance~\cite{hellmuth2024resolving}, which we do not require here.  In particular, an edgeless graph with two vertices is prime, but with three vertices it is not prime.}.

 A \emph{modular partition} $\M$ of $G$ is a partition of $V(G)$ in which each set $M \in \M$ is a module. 
 A subgraph $G[M]$ induced by a module $M \in \M$ is called a \emph{factor} of $\M$. 
 The \emph{quotient graph} $G/\M = (V_\M, E_\M, \omega_\M)$ is the graph whose set of vertices is $\M$ and in which two vertices $M_1, M_2 \in V(G/\M)$ are adjacent if and only if every vertex of $M_1$ is adjacent to every vertex of $M_2$ in $G$.   Moreover, we define the weights as $\omega_\M(M) = \sum_{v \in M} \omega(v)$ for each $M \in \M$. This is the usual notion of quotient, but when a module $M$ is ``compressed'' to a single vertex, $\omega_{\M}$ keeps track of the total weight of vertices represented by that vertex.  Observe that alternatively, one can obtain $G / \M$ by choosing exactly one vertex per module of $\M$, and taking the subgraph induced by the $|\M|$ chosen vertices (then adjusting weights accordingly).  This works because in a module, all vertices have the same adjacencies to other modules. 
 
  When $G$ is connected and complement-connected, then its maximal strong modules form a unique modular partition.
  We define the \emph{modular decomposition} $\M^*$ of $G$, which is a partition of $V(G)$, as follows:
  \begin{itemize}
      \item 
      If $G$ is connected and complement-connected, then $\M^*$ is the set of  maximal strong modules of $G$.

      \item 
      If $G$ is disconnected, then we define $\M^* = \{V(G_1), \ldots, V(G_k)\}$,  where $G_1, \ldots, G_k$ are the components of $G$.
      
      \item 
      If $\overline{G}$ is disconnected, then we define $\M^* = \{V(G_1), \ldots, V(G_k)\}$,  where $G_1, \ldots, G_k$ are the components of $\overline{G}$.
  \end{itemize}
  
  In the first case where $G$ is connected and complement-connected, the quotient graph $G / \M^*$ is known to be prime \cite{Gallai:[4]}. 
  Otherwise, $G / \M^*$ is either an independent set or a clique.  
  We note that $\M^*$ can be found in linear time~\cite{Tedder:[2]}.

\paragraph*{Graph classes.}
A \emph{graph class} $\mathcal{C}$ is a set of graphs, possibly infinite.  
A graph class $\mathcal{C}$ is \emph{hereditary} if, for any $G \in \mathcal{C}$, every induced subgraph of $G$ is in $\C$.  If $\H$ is a set of 
graphs, the class $\HF$ is the set of all graphs that do not contain an induced subgraph isomorphic to a graph in $\H$.  
It is well-known that a class $\C$ is hereditary if and only if $\C = \HF$ for some set of graphs $\H$ (possibly infinite, see e.g.~\cite[Chapter 2]{Kitaev2015}). 
We may say that a graph $G$ is $\HF$ instead of saying that $G$ is ``in'' $\HF$. 
An hereditary graph class $\HF$ is 
\emph{infinite} if there is an infinity of graphs in $\HF$.
A graph $H \in \H$ is \emph{minimal} if no proper induced subgraph of $H$ is in $\H$.

The graph class \emph{cograph} is the set of graphs with no induced $P_4$, the chordless path on four vertices.  It therefore coincides with the class $\{P_4\}$-free, which we may just write as $P_4$-free.

\paragraph*{Edge modification problems.}
For a graph $G$, we write ${ V(G) \choose 2}$ for the set of unordered pairs of $V(G)$ and call a subset of ${V(G) \choose 2}$ an \emph{editing set}. For $\E \subseteq {V(G) \choose 2}$, we denote by $G \op \E$ the graph with vertex set $V(G \op \E) = V(G)$ and $E(G \op \E) = E(G) \Delta \E$, where $\Delta$ is the symmetric difference.  The pairs in $\E \cap E(G)$ are called \emph{edge deletions}, and those of $\E \setminus E(G)$ are \emph{edge insertions}.  If $G \op \E$ is $\HF$, then $\E$ is an \emph{$\HF$ editing set}.  If $\E$ contains only insertions or only deletions, we may call it an $\HF$ insertion set and $\HF$ deletion set, respectively.  We write $cost(\E)$ for the sum of its edge weights, i.e., $cost(\E) = \sum_{uv \in \E} \omega(u) \omega(v)$.

We use the symbol $\oplus$ to denote the operation of edge insertions and $\ominus$ for edge deletions.  
For $\sigma \subseteq \{\oplus, \ominus\}$, we say that an editing set $\E$ \emph{respects} $\sigma$ if either $\sigma = \{\oplus, \ominus\}$; or $\sigma = \{\oplus\}$ and $\E$ contains only edge insertions; or $\sigma = \{\ominus\}$ and $\E$ contains only edge deletions.  This allows us to define the edge deletion/insertion/both variants of the editing problem.

\medskip
\noindent 
The \HHFE{} problem \\
\noindent 
\textbf{Input.}  A graph $G$ and an integer $k$. \\
\textbf{Question.}  Does there exist an $\HF$ editing set $\E \subseteq {V(G) \choose 2}$ that respects $\sigma$ with $cost(\E) \leq k$?

\medskip 

The variants $\sigma = \{\oplus\}$ and $\sigma = \{\ominus\}$ are called the \HHFI{} and \HHFD{} problems, respectively.
If $\E$ is an $\HF$ editing set of minimum cost that respects $\sigma$, then $\E$ is called an \emph{optimal $\HF$ editing set}.  In this case, we denote $opt(G) = cost(\E)$.

\paragraph*{Safe editing sets and branching factors.}
Given an instance $G$ of \HHFE, a collection $S = \{\E_1, \ldots, \E_p\}$ of non-empty editing sets is \emph{safe} if each of them respects $\sigma$, and there exists $\E_i \in S$ and an optimal $\HH$-free editing set $\E^*$ of $G$ such that $\E_i \subseteq \E^*$.  Roughly speaking, this means that if we recursively branch into all the editing sets of $S$, one of the resulting graphs can lead to an optimal solution.
The \emph{branching vector} of $S$ is the vector $\beta(S) = (cost(\E_1), \ldots, cost(\E_r))$, and the \emph{branching factor} $b(S)$ of $S$ is the largest real root of the polynomial 
$x^k - \sum_{\E_i \in S} x^{k-cost(\E_i)}$.  It is well-known that a recursive algorithm that branches on each editing set $\E_1, \ldots, \E_p$, where each branch reduces the parameter $k$ by $cost(\E_i)$, has $O^*( b(S)^k)$ nodes in the recursion tree.

Given two branching vectors $\beta(S) = (b_1, b_2, \ldots, b_x)$ and $\beta(S') = (b'_1, b'_2, \ldots, b'_y)$ we say that $\beta(S)$ \emph{dominates} $\beta(S')$ if there is an injective map $f : [y] \rightarrow [x]$ such that $b'_i \geq b_{f(i)}$ for each $i \in [y]$ (we use the notation $[n] = \{1, 2, \ldots, n\}$).  The intuition is that each entry of $\beta(S')$ can be paired with an entry of $\beta(S)$ such that the latter reduces the parameter by \emph{less}, and thus $\beta(S)$ has a larger recursion tree.  For example, $(1, 1, 2, 3, 4)$ dominates $(1, 2, 3, 5)$.  Observe that when this is the case, then $b(S') \leq b(S)$, so in our worst-case analysis we usually only consider the dominant branching vectors.  

%\cographDeletion{} is the set of \textit{$P_4$-Free} graphs achieved through edges deletion only.

%------------------------------------------------------------------------------
%	ALGORITHM
%------------------------------------------------------------------------------

\section{The generic \HHFE{} Algorithm}

The idea of our \HHFE{} algorithm is as follows:
\begin{enumerate}
    \item 
    If $G$ has at most $C$ vertices, where $C$ is a constant to be specified later, then solve $G$ optimally using any brute-force strategy.

    \item 
    Otherwise, compute the modular decomposition $\M^*$ of $G$.  Solve the quotient $G / \M^*$ optimally and make the corresponding modifications in $G$.  Then, for each $M_i \in \M^*$, solve $G[M_i]$ independently.
\end{enumerate}
 
We note that although we focus on cograph deletion problems in this paper, there is nothing specific to cographs in this general strategy, and it could be applicable to other graph classes, beyond cographs as done in~\cite{guillemot2013non,fritz2020cograph}.  
To describe the algorithm more precisely, let $\M$ be a modular partition of $G$, and let $\E$ be an editing set of $G/\M$.  Observe that vertices of $G / \M$ are modules, so its edges have the form $M_i M_j$, where $M_i, M_j \in \M$.  
We define
\[
    ext(G,\M,\E) = \{xy \mid M_1 M_2 \in \E, x \in M_1, y \in M_2\},
\]
which is the editing set of $G$ obtained by \emph{extending} the modifications on the quotient graph $G/\M$ to $G$.
Let us note that, as argued later, our weights of quotient graphs ensure that $cost(\E) = cost(ext(G, \M, \E))$.

Our \HHFE{} algorithm uses two black-box functions:
\begin{itemize}
    \item 
    \textbf{solve} takes as input a vertex-weighted graph $G$ and returns an optimal editing set of $G$. 
    
    \item 
    \textbf{branch} takes as input a vertex-weighted graph $G$ and returns a safe set $S = \{\E_1, \ldots, \E_p\}$ of editing sets of $G$ that each respect $\sigma$.  The \emph{branching vector of $S$} is $(cost(\E_1), \ldots, cost(\E_p))$.  The \textbf{branch} function may return sets with  different branching vectors depending on the input graph.  
    The \emph{branching factor} of \textbf{branch} is the maximum branching factor of all branching vectors possibly returned by the function.
    
\end{itemize}
We also assume that we have access to a black-box algorithm to test whether a graph $G$ is $\H$-free or not (which may not be trivial if $\H$ is infinite), and the \textbf{branch} function may need to be able to find induced subgraphs that belong to $\H$.  

The intent is that when solving the \HHFE{} problem for a specific graph class $\H$, one only needs to define these two routines and feed them to the general purpose algorithm.  
The \textbf{solve} routine is used in the first step above for small enough graphs, and the \textbf{branch} routine is used in the second step when solving $G /\M^*$ or a $G[M_i]$ factor.  In fact, it then only suffices to describe the \textbf{branch} routine on a prime graph.  Algorithm~\ref{alg:alg} describes the procedure, which we simply call $f$, in more detail.
Note, due to how the algorithm works, in particular Line~\ref{line:else-solve-modules}, it is not enough to return yes or no to the decision problem --- it is important to return $opt(G)$ for the parent calls.  We use $INF$ to denote infinity and assume that adding $INF$ to any number yields $INF$.

%Our algorithm uses modular \branching{} and \solver{} functions that support specific \HHFE{} problems. The support of a \branching{} function is defined by its ability to correctly and efficiently return $\E$ a non-empty set of editing sets, each targeting the same two to many subgraphs that are also graphs of $\HH$ in a bounded search tree fashion over nodes that are prime, so that for any optimal editing set $X$, there is always an editing set $E \in \E$ such that $E \subseteq X$.

% Let $G'\subseteq G$, then we note that at least one of the resulting editing sets in $G'$ is a subset of the optimal editing set in $G$. The support of a \solver{} function is defined by its ability to solve $G$, where $V(G)<C$, trivially and in constant time relative to $C$ by returning the optimal solution with lowest weight to the selected \HHFE{} problem.

% Let a simple graph $G$ extended to its vertex-weighted setting $G=(V,E,\omega)$, a parameter $k$, and a \branching{} function that supports the targeted graph family. We now show the \HHFE{} \textit{algorithm}~[\ref{alg:alg}].

\vspace{3mm}
\begin{algorithm}[H]
\setstretch{1.35}
\SetAlgoNoLine
\SetKwProg{Fn}{function}{}{}
\Fn{$f$($G=(V,E,\omega),$ $k$)}{
    \lIf{$k \geq 0$ and $G$ is $\HF$}{\Return $0$} \label{line:is-h-free}
    \lIf{$k \leq 0$}{\Return $INF$}  \label{line:k-zero}
    \lIf{$|V|<C$}{\Return \textbf{solve}$(G,k)$} \label{line:small-graph}
    Let $e=0$ and let $\M^*$ be the modular decomposition of $G$\label{line:get-modular}\;
    \lIf{$|V(G/\M^*)|\geq C$ and $G / \M^*$ is not $\HH$-free}{$e=\min\limits_{\E\in \mathbf{branch}(G/\M^*)} f(G \Delta ext(G,\M^*,\E),k-cost(ext(G,\M^*,\E))) + cost(ext(G,\M^*,\E))$} \label{line:if-solve-gm}
    \lElse{$e= \mathbf{solve}(G/\M^*, k) +\sum_{M\in \M^*} f(G[M],k)$}  \label{line:else-solve-modules}
   \lIf{$e\leq k$}{\Return $e$}  \label{line:good-k}
    \Return $INF$ \label{line:bad-k}\;
 }
\caption{FPT \HHFE{} on a vertex-weighted graph $G$.}
\label{alg:alg}
\end{algorithm}
\vspace{3mm}

We now state the main theorem of this section. 
Here, we say that the algorithm is \emph{correct} if, given $G$ and $k$, either $opt(G) \leq k$ and the algorithm returns $opt(G)$, or $opt(G) > k$ and the algorithm returns either $opt(G)$ or $INF$.
The variable $e$ is used to store the best solution found and ensure it stays below the parameter $k$.

\begin{theorem}\label{thm:algocorrect}
    Suppose that $\HF$ is an infinite graph class.  If all minimal graphs in $\H$ are prime, then Algorithm~\ref{alg:alg} is correct for any constant $C$.  Conversely, if some minimal graph in $\H$ is not prime, then for any constant $C$ there exist instances on which Algorithm~\ref{alg:alg} fails, even if \textbf{solve} is correct and \textbf{branch} returns a safe editing set.

    Moreover, suppose that one can verify whether a given graph $G$ with $n$ vertices is $\HF$ in time $O(g(n))$, that \textbf{branch} runs in time $O(h(n))$, and that \textbf{branch} has branching factor $b$.  
Then Algorithm~\ref{alg:alg} runs in time $O(n \cdot b^k \cdot (g(n) + h(n) + n + m))$. 
\end{theorem}

The rest of the section builds the necessary tools to prove the theorem, which has two non-obvious components.  First, that the algorithm does return an optimal editing set if the minimal graphs of $\H$ are prime.  Second, that if some minimal graph of $\H$ is not prime, then the algorithm \emph{fails}, no matter how large the constant $C$ is.  This requires showing that there exist arbitrarily large graphs that are not $\H$-free, but for which the quotient graph of its modular decomposition and also each module itself are $\H$-free.

As another remark, note that the constant $C$ is not involved directly in the complexity.  Instead, $C$ determines the point at which we can rely on an improved \textbf{branch} procedure or not.  In the latter case where $C$ is not large enough, we can still make use of other FPT algorithms, for instance Tsur's algorithm. 

Let us begin with the easier properties, by emphasizing that most of the following intermediate results apply to any modular partition $\M$, not necessarily the modular decomposition $\M^*$.  The latter is used to guarantee that $G / \M^*$ is a prime graph later on.
 
\begin{lemma}\label{lem:module-h-free}
    Suppose that all minimal graphs of $\HH$ are prime.  
    Let $\M$ be a modular partition of graph $G$.
    If $G / \M$ is $\HH$-free and $G[M_i]$ is also $\HH$-free for every $M_i \in \M$, then $G$ is $\HH$-free.
\end{lemma}

\begin{proof}
    Assume for contradiction that $G$ is not $\H$-free, i.e., there is $X \subseteq V(G)$ such that $G[X]$ is (isomorphic to) a graph in $\HH$.  Choose such an $X$ of minimum size.
    Thus $G[X]$ is a minimal graph in $\HH$, since otherwise, some proper subset of $X$ would induce a smaller graph that is in $\HH$.  So by assumption, $G[X]$ is prime.  Observe that $X$ cannot be entirely contained in a module $M_i$ of $\M$, since $G[M_i]$ is $\HH$-free.  
    Likewise, $X$ cannot contain one vertex or less per module of $\M$: if that was the case, then $G[X]$ would also be an induced subgraph of $G / \M$, which is assumed to be $\H$-free.
Thus we may assume that $X$ contains at least two vertices of some module $M_i \in \M$, plus some vertices from modules other than $M_i$. 
But then, all the vertices in $X \cap M_i$ have the same neighbors in $X \setminus M_i \neq \emptyset$.  So, $X \cap M_i$ is a non-trivial module of $G[X]$, and so $G[X]$ is not prime, contradicting our assumption.  It follows that $G$ is $\HF$.
\end{proof}

We next show that we may assume that modules of the chosen modular partition $\M$ remain modules after editing.  A similar statement was shown in~\cite[Lemma 2]{guillemot2013non} for cographs, here we generalize to a larger number of $\HF$ classes, and to vertex-weighted graphs.  Since a similar fact already appears in~\cite{guillemot2013non}, we defer the proof to the appendix and provide a sketch of the proof.

\begin{lemmarep}\label{lem:keep-modules}
Suppose that all minimal graphs of $\HH$ are prime. Then for any vertex-weighted graph $G = (V, E, \omega)$ and any modular partition $\M$ of $G$, there exists an optimal \HHF{} editing set $\E$ such that between any two distinct modules $M_1, M_2 \in \M$, either $G \op \E$ contains all edges between $M_1$ and $M_2$, or none.
\end{lemmarep}

\begin{proofsketch}
    We define a representative set as a set $R \subseteq V(G)$ that contains exactly one vertex per module $M_i \in \M$.  Let $\E$ be any optimal $\HF$ editing set of $G$.  Consider the representative set $R$ that minimizes the edit cost to transform $G[R]$ into $(G \Delta \E)[R]$.  These representatives have essentially found the optimal way to make $G[R]$ $\HF$, so every vertex might as well do the same as the representative from its module.  That is, consider an alternate editing set, where we add/delete every edge between modules $M_i, M_j$ if and only if the edge between the vertex of $R \cap M_i$ and of $R \cap M_j$ is added/deleted.  We can argue that the result is $\HF$ by Lemma~\ref{lem:module-h-free} and  also optimal, the main difficulty being in the calculations that take weights into account.
\end{proofsketch}

\begin{proof}
Denote $\M = \{M_1, \ldots, M_p\}$.  For a module $M_i \in \M$ we write $\omega(M_i) = \sum_{u \in M_i} \omega(u)$. 
For the duration of the proof, a 
\emph{representative} of a module $M_i$ is a pair $\langle u, t \rangle$ such that $u \in M_i$ and $t \in [\omega(u)]$. 
The intuition is that we choose one vertex $u$ to represent $M_i$, but give ``more ways'' to choose $u$ according to $\omega(u)$.
A \emph{representative set} is a set $R = \{ \langle u_1, t_1 \rangle, \ldots, \langle u_p, t_p \rangle \}$, where for each $i \in [p]$, $\langle u_i, t_i \rangle$ is a representative of module $M_i$.  For $i \in [p]$, we denote by $R(i)$ the vertex $u_i$ that belongs to the unique representative $\langle u_i, t_i \rangle$ of $M_i$.  
Let $\R$ denote the set of all possible representative sets, and note that $|\R| = \prod_{i \in [p]}\omega(M_i)$.

Let $\E$ be an optimal $\HF$-Editing set of $G$, and denote $\hat{G} = G \op \E$. 
For a representative set $R \in \R$, 
obtain a graph $\hat{G}|_R$ by starting from $\hat{G}$ and applying the following:\\
- for each $M_i, M_j \in \M$ such that $R(i) R(j) \in E(\hat{G})$, add all edges between $M_i$ and $M_j$; \\
- for each $M_i, M_j \in \M$ such that $R(i) R(j) \notin E(\hat{G})$, remove all edges between $M_i$ and $M_j$.

In other words, just use the representative vertices to dictate the presence/absence of all or none of the edges between modules.
Observe that for any $R \in \R$, any $M_i$ of $\M$ is a module of $\hat{G}|_R$, since by construction all vertices of $M_i$ are made to have the same neighbors in $V(G) \setminus M_i$.  Therefore, the graph $\hat{G}|_R / \M$ is well-defined.  Moreover, note that $\hat{G}|_R / \M$ is isomorphic to $\hat{G}[\{R(1), \ldots, R(p)\}]$, and since $\hat{G}[\{R(1), \ldots, R(p)\}]$ is $\HH$-free by heredity, we have that $\hat{G}|_R / \M$ is also $\HH$-free.  Also, for any $M_i \in \M$, $\hat{G}|_R[M_i]$ is $\HH$-free because our construction of $\hat{G}|_R$ does not modify edges of $\hat{G}$ within the modules.  
We therefore satisfy all the conditions of Lemma~\ref{lem:module-h-free} and we deduce that $\hat{G}|_R$ is $\HH$-free.

We argue that there is some $R \in \R$ such that the cost of edge modifications needed to transform $G$ into $\hat{G}|_R$ is at most $cost(\E)$.  
From now on, denote $\gamma = \sum_{i \in [p]} cost(E(G[M_i]) \Delta E(\hat{G}[M_i]))$, which is the cost of edge modifications from $G$ to $\hat{G}$ whose two ends are inside a module of $\M$.  Notice that $\gamma$ is equal to
$\sum_{i \in [p]} cost(E(G[M_i]) \Delta E(\hat{G}|_R[M_i]))$ for any $R \in \R$, since again $\hat{G}|_R$ does not alter the edges and non-edges of $\hat{G}$ within the $M_i$ modules.  

For $u, v \in V(G)$ that belong to two distinct modules of $\M$, denote 
\begin{align*}
\one_{u,v} = \begin{cases}
	1 & \mbox{if $uv \in E(G) \Delta E(\hat{G})$} \\
	0 & \mbox{otherwise}
\end{cases}
\end{align*}
which is just an indicator function for whether $uv$ was modified from $G$ to $\hat{G}$.

Notice that for $R \in \R$ and two modules $M_i, M_j \in \M$, if $\one_{R(i) R(j)} = 1$, then the edge $R(i) R(j)$ between $M_i$ and $M_j$ was  inserted/deleted to obtain $\hat{G}$.  There are only edges or only non-edge between $M_i$ and $M_j$ in $G$, so in $E(G) \Delta E(\hat{G}|_R)$, that same modification is done between all pairs of vertices in $M_i, M_j$.
Thus when $\one_{R(i) R(j)} = 1$, the cost of modifications between $M_i$ and $M_j$ in $\hat{G}|_R$ is $\omega(M_i) \omega(M_j)$. 
If $\one_{R(i), R(j)} = 0$, then by the same reasoning there is nothing to edit between $M_i$ and $M_j$ to transform $G$ into $\hat{G}|_R$.

For $R \in \R$, denote the quantity 
\[
cost(R) = \sum_{1 \leq i < j \leq p} \omega(M_i)\omega(M_j) \cdot \one_{R(i),R(j)}
\]
which, by the explanations from the previous paragraph, is the cost of edges to modify between modules of $\M$ to transform $G$ into $\hat{G}|_R$.
Let $R^* \in \R$ be a set of representatives that minimizes $cost(R^*)$.  
Note that transforming $G$ into $\hat{G}|_{R^*}$ requires a cost of $\gamma + cost(R^*)$ edge modifications.  

Now consider the cost of edge modifications needed to transform $G$ into $\hat{G}$, which we count by summing over all sets of representatives.  That is, we can express 
\begin{align*}
	cost(\E) = cost(E(G) \Delta E(\hat{G})) = \gamma + \sum_{R \in \R} \sum_{1 \leq i < j \leq p} \frac{\omega(R(i)) \omega(R(j)) \cdot \one_{R(i), R(j)}}{\alpha(R(i), R(j))}
\end{align*}
where here, $\alpha(R(i), R(j))$ is the number of times that the value $\omega(R(i)) \omega(R(j) \cdot \one_{R(i), R(j)}$ appears in the double summation, to compensate for over-counting of the modified edges (that is, each unmodified edge adds $0$ to the double-summation, whereas each modified edge is added $\alpha(R(i), R(j))$ times).  
This value $\alpha(R(i), R(j))$ is equal to the number of elements of $\R$ that contain pairs of the form $\langle R(i), t_i \rangle$ and $\langle R(j), t_j \rangle$, where $t_i \in [\omega(R(i))]$ and $t_j \in [\omega(R(j))]$.  There are $\omega(R(i)) \omega(R(j))$ ways to fix those two, and counting the number of ways to choose representatives in other modules, we infer that 
\[
\alpha(R(i), R(j)) = \omega(R(i)) \omega(R(j)) \cdot \prod_{h \in [p] \setminus \{i,j\}} \omega(M_h) = \omega(R(i)) \omega(R(j)) \cdot \frac{\prod_{h \in [p]} \omega(M_h)}{\omega(M_i) \omega(M_j)}.
\]  
It follows that 
\begin{align*}
	cost(E(G) \Delta E(\hat{G})) &= \gamma + \sum_{R \in \R} \sum_{1 \leq i < j \leq p} \omega(M_i) \omega(M_j) \frac{ \omega(R(i)) \omega(R(j)) \cdot  \one_{R(i), R(j)}}{ \omega(R(i)) \omega(R(j)) \cdot \prod_{h \in [p]} \omega(M_h)} \\
					 &= \gamma + \frac{1}{\prod_{h \in [p]} \omega(M_h)} \sum_{R \in \R} \sum_{1 \leq i < j \leq p} \omega(M_i) \omega(M_j) \one_{R(i), R(j)} \\
					 &\geq \gamma + \frac{1}{\prod_{h \in [p]} \omega(M_h)} \sum_{R \in \R} cost(R^*) \\
                     &= \gamma + \frac{1}{\prod_{h \in [p]} \omega(M_h)} \prod_{h \in [p]} \omega(M_h) \cdot cost(R^*) \\
					 &= \gamma + cost(R^*)
\end{align*}
where, for the third line, we used the fact that $R^*$ minimizes the quantity $\sum_{1 \leq i < j \leq p} \omega(M_i)\omega(M_j) \one_{R(i), R(j)}$, and on the fourth line we used the fact that the number of elements in $\R$ is $\prod_{h \in [p]} \omega(M_h)$. 

This shows that transforming $G$ into $\hat{G}|_{R^*}$ requires no more modifications than $\hat{G}$, and $E(G) \Delta E(\hat{G}|_{R^*})$ is also an optimal editing set.  This concludes the proof, since $\hat{G}|_{R^*}$ has the form required by our statement.
\end{proof}

We can now argue that if all minimal graphs of $\H$ are prime, then we may indeed decompose the problem into optimizing $G /\M^*$ and each $G[M_i]$ independently.  Note that it is necessary to have vertex weights to be able to use $opt(G / \M)$ in the statement below.

\begin{lemma}\label{lem:opt-g}
    Suppose that all minimal graphs of $\HH$ are prime.  Let $G$ be a graph and let $\M$ be any modular partition of $G$.  
    Then $opt(G) = opt(G/\M) + \sum_{M_i \in \M} opt(G[M_i])$.

    Moreover, for any optimal editing set $\E$ of $G/\M$, there exists an optimal editing set $\E^*$ of $G$ such that $ext(G, \M, \E) \subseteq \E^*$.
\end{lemma}

\begin{proof}
    First let $\E^*$ be an optimal editing set of $G$ such that for distinct $M_i, M_j \in \M$, either all edges or none are present in $G \op \E^*$. Such an $\E^*$ exists by Lemma~\ref{lem:keep-modules}.  We argue that $cost(\E^*) \geq opt(G / \M) + \sum_{M_i \in \M} opt(G[M_i])$.   
    Since the $M_i$ modules are vertex-disjoint, $\E^*$ must have a cost of at least $opt(G[M_i])$ edge modifications whose endpoints are both in $M_i$, for each $M_i \in \M$.  
    These modifications cannot remove any copy of an $\HH$ graph that contains one vertex per module, and for that we claim that at least $opt(G / \M)$ additional edge modifications are needed.  
    Indeed, it suffices to observe that because either all or none of the edges are present between modules, for any two $M_i, M_j \in \M$ the cost of edge modifications of $\E^*$ between the two modules is either $0$ or $\sum_{v \in M_i} \omega(v) \cdot \sum_{w \in M_j} \omega(w)$.  
    Because the weight of an $M_i$ in $G / \M$ is the sum of its vertex weights, modifiying $M_i M_j$ in $G / \M$ also has that same cost, and it follows that $(G \Delta \E^*) / \M$ (which is well-defined because $\M$ is still a modular partition of $G \Delta \E^*$) yields a graph that can be obtained from $G / \M$ with a cost that is exactly the sum of costs between modules from $G$ to $G \Delta \E^*$.  Moreover, since $(G \Delta \E^*) / \M$ is isomorphic to any induced subgraph of $G \Delta \E^*$ obtained by taking one vertex per module, it is $\HF$.  Therefore, the edges of $\E^*$ whose ends are in separate modules correspond to an editing set of $G / \M$ of the same cost, and so those edges incur a cost of at least $opt(G / \M)$, which proves our claim.  Therefore, $opt(G) \geq opt(G / \M) + \sum_{i \in [p]} opt(G[M_i])$.  

    For the converse bound, 
    consider an editing set $\E'$ of $G$ obtained as follows: take an optimal editing set $\E$ of $G/\M$, then for each module $M_i \in \M$ take an optimal editing set $\E_i$ of $G[M_i]$.  Define the editing set $\E' = ext(G, \M,\E) \cup \E_1 \cup \ldots \cup \E_{|\M|}$, and define $G' = G \op \E'$.  

    Note that $\M$ is a modular partition of $G'$.  Moreover, $G'/\M$ is $\HH$-free and every $G'[M_i]$ is $\HH$-free for $M_i \in \M$, so by Lemma~\ref{lem:module-h-free} $G'$ is $\HH$-free.  Thus $\E'$ is an $\HH$-free editing set of $G$ and the cost of modified edges is the expression given in the lemma statement.  This shows that 
    $opt(G) \leq opt(G/\M) + \sum_{M_i \in \M} opt(G[M_i])$.
    Our complementary bounds therefore yield the desired equality on $opt(G)$.
    In particular, the editing set $\E'$ we just described achieves the lower bound, and so it is optimal.  Since the optimal editing set $\E$ of $G / \M$ used to construct $\E'$ is arbitrary, this shows that any such $\E$ is contained in some optimal editing set, which proves the ``moreover'' part of the statement.
\end{proof}

As a consequence, extending the modifications on $G /\M^*$ to $G$ itself, as performed by Algorithm~\ref{alg:alg}, can be shown to be safe.

\begin{lemma}\label{lem:safe-transfer}
    Suppose that all minimal graphs of $\HH$ are prime.  
    Let $G$ be a graph and let $\M$ be any modular partition of $G$.  Suppose that a collection of editing sets $S$ is safe for the quotient graph $G /\M$.  

    Then the collection
    \[
    S' = \{ ext(G, \M, \E) \mid \E \in S \}
    \]
    is safe for $G$.
\end{lemma}
 
\begin{proof}
    Suppose that $S$ is safe for $G/\M$.  Then by definition, there exist an optimal editing set $\E^*$ of $G /\M$ and $\E \in S$ such that $\E \subseteq \E^*$.  Notice that $ext(G, \M, \E) \subseteq ext(G, \M, \E^*)$.  
    By Lemma~\ref{lem:opt-g}, there exists an optimal editing set $\E_G$ of $G$ such that $ext(G, \M, \E^*) \subseteq \E_G$.
    We infer that 
    \[
    ext(G, \M, \E) \subseteq ext(G, \M, \E^*) \subseteq \E_G
    \]
    implying that $S'$ as defined in the statement is safe. 
\end{proof}

The facts accumulated so far are sufficient to prove the correctness of the algorithm.  Before moving on to the proof of the theorem, we look at the cases in which the algorithm is incorrect, which are arise from the following result.  Since the proof is relatively long and does not involve cograph deletion, we defer it to the Appendix.

\begin{propositionrep}\label{prop:alg-incorrect}
    Suppose that not all minimal graphs of $\H$ are prime and that $\HF$ is infinite.  Then for any constant $C$, there exists a graph $G$ with $|V(G)| \geq C$ that is not $\H$-free, but such that $G / \M^*$ and $G[M_1], \ldots, G[M_p]$ are all $\H$-free, where $\M^* = \{M_1, \ldots, M_p\}$ is the modular decomposition of $G$.
\end{propositionrep}

\begin{proof}
We need a few cases depending on the type of non-prime graph found in $\H$.

\medskip 

\noindent 
\emph{Case 1.}  $\H$ contains a minimal graph $H$ that is not prime and edgeless.    In that case, note that $|V(H)| \geq 3$ since an edgeless graph with two vertices is prime (all modules are trivial).  
Moreover, since $\HF$ is infinite, one can see that all complete graphs are $\HF$.  Indeed, if some $K_l$ is not $\HF$, then all graphs in $\H$ are $\{H, K_l\}$-free.  But $H$ is edgeless and $K_l$ is complete, so by a Ramsey argument, every large enough graph has an induced $H$ or $K_l$ and the class could not be infinite.

So, consider the graph $F$ obtained as follows: take a clique of size $C$, then add two vertices $v, v'$ adjacent to all vertices of the clique ($v, v'$ are not adjacent).  
Let $\M_F$ be the modular decomposition of $F$.  Note that $\M_F$ has two modules: the clique of size $C$ (which is $\H$-free as we just argued) and $\{v, v'\}$ (which is $\H$-free since $|V(H)| \geq 3$ and $H$ is minimal).  Moreover, $F / \M_F$ is a clique of size two and thus $\H$-free. 
Therefore, if $F$ is \emph{not} $\H$-free, we are done.  Assume henceforth that $\F$ is $\H$-free.

Then, obtain another graph $F'$ from $F$ by adding $|V(H)| - 2$ singleton vertices to $F$, i.e., new vertices that are incident to no edge.  
First note the modular decomposition $\M^*$ of $F'$ consists of its $|V(H)| - 1$ connected components.  Hence $F' /\M^*$ is an edgeless graph with $V(H) - 1$ vertices, so it is a proper induced subgraph of $H$ an is thus $\HF$, by the minimality of $H$.  The singleton vertices are also $\HF$, and so is $F$ by assumption.  
On the other hand, $v, v'$ and the singleton vertices induce an edgeless graph with $|V(H)|$ vertices and thus $F'$ is not $\HF$.

\medskip 

\noindent 
\emph{Case 2.}  $\H$ contains a minimal graph $H$ that is not prime and that is a complete graph.  This case is symmetric to the previous one: one can see that since $\HF$ is infinite, all edgeless graphs are $\HF$.  Then, it suffices to take the complement of the graph $F'$ described above, and the same arguments apply.

\medskip 

\noindent 
\emph{Case 3.} No minimal and non-prime graph in $\H$ is an edgeless graph or a complete graph.  
In that case, let $H$ be a minimal graph of $\H$ which is not prime.  First note that $\H$ has no edgeless graph: if it did, that edgeless graph would be prime by assumption and would have only two  vertices.  By minimality, $H$ could not have a non-edge and it would be a clique, which we assume is not the case.

Let $\M^* = \{M_1, \ldots, M_p\}$ be the modular decomposition of $H$, with  $p > 1$.  By the minimality of $H$, each $H[M_i]$ is $\HH$-free, for $M_i \in \M^*$.
Also note that some module $M_i \in \M^*$ has at least two vertices: if $H$ is disconnected, some component is not a singleton because $H$ is not an edgeless graph; if $\overline{H}$ is disconnected, the same applies; if $H$ and $\overline{H}$ are connected,  some module is non-trivial.  This means that $H / \M^*$ is isomorphic to some induced subgraph of $H$, and we have that $H / \M^*$ is also $\HH$-free.   
Therefore, $H$ itself satisfies the properties of our statement, except that it may have less than than $C$ vertices.
So, we next to modify $H$ so that it has more than $C$ vertices.

Let us suppose, without loss of generality, that module $M_1 \in \M^*$ of $H$ is of minimum size, among all modules of $\M^*$.
Consider the graph $H'$ obtained from $H$ by taking an arbitrary vertex $v \in M_1$, and replacing it by some graph $F$ that is $\HF$ and has at least $C$ vertices (which exists by infinity).  To be precise, take the disjoint union of $H - v$ and $F$, then for every $x \in V(F)$ and $y \in V(H - v)$, add $xy$ to $E(H')$ if and only if $vy \in E(H)$.  Note that $H'$ is not $\HH$-free, since any vertex of $F$ can play the role of $v$.
Denote $M'_1 = V(F) \cup (M_1 \setminus \{v\})$.  One can verify that $\M' = \{M_1', M_2, \ldots, M_p\}$ is the modular decomposition of $H'$ (these are modules by construction, and no proper subset of all the modules can be merged by taking their union to obtain a larger module, as otherwise we could do the same in $\M^*$, and so they are also maximal).

Suppose that $H'[M_1']$ is $\HH$-free.  In that case, we have that $H' / \M'$ is isomorphic to $H / \M^*$, and it is thus $\HH$-free as argued before.  Moreover, $H'[M_1']$ is $\HH$-free by assumption, and each $H'[M_i]$ is $\HH$-free because they are the same graph as $H[M_i]$, for $M_i \in \M'$.  We therefore have our desired graph.

So assume from now on that $H'[M_1']$ is not $\HH$-free.  Note that $|M_1| > 1$, because if $|M_1| = 1$ by replacing the single vertex $v$ of $M_1$ with $F$, we get that $M'_1 = V(F)$, which induces an $\HH$-free graph.  So we assume that $|M_1| \geq 2$.
We will find the graph we need among induced subgraphs of $H'[M_1']$ of size more than $C$.
Let $X \subseteq M_1'$ be a subset of $M'_1$ of minimum size such that $V(F) \subset X$ and $H'[X]$ is not $\HH$-free.  Note that $X$ exists since $M_1'$ is one possible value for $X$.  
Also note that for any non-empty subset $Y \subseteq X \setminus V(F)$, $H'[X] - Y$ is $\HH$-free.

Let $\M^X$ be the modular decomposition of $H'[X]$.
Note that $H'[X]$ is not prime, because it contains $V(F)$ as a non-trivial module, in particular.  
We must have one more subdivision into subcases.

\begin{itemize}
    \item 
    $H'[X]$ is disconnected.  
    Then $H'[X] / \M^X$ is an independent set and, since we may assume that $\HH$ does not contain an edgeless graph, it is $\HH$-free.  
    Now let $Y$ be a connected component of $H'[X]$.  If $Y$ does not intersect with $V(F)$, then $Y$ is a proper subset of $X$ and thus of $M_1$, and since $H[M_1]$ is $\HF$ we have that $H'[Y] = H[Y]$ is $\HH$-free.
    If $Y$ contains $V(F)$, then $Y$ is $\HH$-free as otherwise we would take $X$ to be $Y$ instead.
    So assume that $Y$ intersects with $V(F)$ but does not contain $V(F)$.  If $v$, the vertex we replaced with $F$, has a neighbor in $X \setminus V(F)$, then every vertex of $F$ has that neighbor and $V(F)$ would form a connected subgraph of $H'[X]$, in which case $Y$ would need to contain $V(F)$.  So we may assume that $v$ and thus all of $V(F)$ have no neighbor in $X \setminus V(F)$.  Therefore, $Y$ must be a proper subset of $V(F)$, and $H'[Y]$ is therefore $\HH$-free.  

    We have thus argued that $H'[X]$ is not $\HH$-free, but $H'[X] / \M^X$ is $\HH$-free and all of its connected components are $\HH$-free, which completes this case.

    \item 
    $\overline{H'[X]}$ is disconnected.  
    The arguments are similar to the previous case.  
    Then $H'[X] / \M^X$ is a complete graph and because $\H$ has no complete graph it is $\HH$-free.  
    Then, a connected component $Y$ of
    $\overline{H'[X]}$ can either be a subset of $X \setminus V(F)$, contain $V(F)$, or be a proper subset of $V(F)$ (which happens only if, in $H$, $v$ has no non-neighbor in $X$).  In all cases, $H'[Y]$ is $\HH$-free, and we are done.

    \item 
    Both $H'[X]$ and $\overline{H'[X]}$ are connected.  In that case, 
    we have that $V(F)$ is a non-trivial module of $H'[X]$, and it is thus contained in some module of $\M^X$.  So $H'[X] / \M^X$ is (isomorphic to) an induced subgraph of $H[M_1]$, because we blew up $v \in M_1$ into $V(F)$, and when taking the quotient there is no difference between $v$ and $V(F)$.  
Because $H[M_1]$ and its induced subgraphs are $\HH$-free, $H'[X] / \M^X$ is also $\HH$-free.  
Now consider a module $M \in \M^X$.  If $M$ is the module that contains $V(F)$, then $H'[M]$ is $\HH$-free because $M$ is a proper subset of $X$ that contains $V(F)$.  Otherwise, $M$ does not intersect with $V(F)$ and thus $M \subset M_1$, in which case it is again $\HH$-free.  Hence $H'[X]$ is again the desired graph.
\end{itemize}

This completes the proof since every case leads to a graph of arbitrary size on which the algorithm fails.  
\end{proof}

\medskip 

As for the proof of Theorem~\ref{thm:algocorrect}, it follows from the previous, but the details can be found in the appendix.  We provide the main ideas here.

\medskip 

\noindent 
\emph{Sketch of the proof of Theorem~\ref{thm:algocorrect}.}
We have all the necessary facts need for  Theorem~\ref{thm:algocorrect}.  Since the proof uses standard arguments, the details are in the appendix.
In short, the fact that it is correct when all minimal graphs of $\H$ are prime is proved by induction over the recursion tree created by the algorithm.  The base cases are easy to check, and the induction uses Lemma~\ref{lem:opt-g} and Lemma~\ref{lem:safe-transfer}.  The fact that the algorithm fails if not all minimal graphs of $\H$ are prime follows from the instances arising from Proposition~\ref{prop:alg-incorrect}.  
As for the complexity, we can again use induction on the recursion tree to argue that it has at most $nb^k$ leaves.  In one case, we recurse on Line~\ref{line:if-solve-gm}, making calls that each reduce the parameter $k$ in correspondence with the branching vector of \textbf{branch}, whose branching factor is at most $b$ by assumption, in which case the bound can be shown to hold.  If we recurse on Line~\ref{line:else-solve-modules}, the calls do not reduce $k$, but they are made on a partition of $V(G)$ and a simple inductive argument shows the bound.  The final complexity is obtained by multiplying the number of nodes in the recursion tree with the time spent per node.

\begin{toappendix}
    
\subsection*{Proof of Theorem~\ref{thm:algocorrect}.}

The statement contains three facts, and we prove each of them separately.  

% \medskip 

% \noindent
% \emph{Proof of Theorem~\ref{thm:algocorrect}}.

\medskip 

\noindent
\textbf{Fact 1.}  If all minimal graphs in $\H$ are prime, then Algorithm~\ref{alg:alg} is correct for any constant $C$.

The proof is by induction on the height of the recursion tree created by the algorithm.
As a base case, consider the leaves of the recursion tree.  If the algorithm returns on Line~\ref{line:is-h-free}, then $k \geq 0$ and $G$ is \HHF{}, in which case $0 \leq k$ editions are required and therefore $0$ is correctly returned. If we do not return on line~\ref{line:is-h-free}, then $G$ is not \HHF{} or $k < 0$. If $k\leq 0$, we correctly return $INF$ on Line~\ref{line:k-zero} since $G$ is not \HHF{} and no more editions are allowed. If the algorithm returns on Line~\ref{line:small-graph} with $G$ being of constant size, then by assumption \textbf{solve} returns the correct value. One can easily check that any other situation makes a recursive call to $f$, so this covers every base case. 

We now assume by induction that any child call of a recursion tree node returns a correct value.
We first assume that the algorithm makes recursive calls on Line~\ref{line:if-solve-gm} on each editing set returned by \textbf{branch}. For this to happen, we have $k>0$, $G$ is not $\HH$-free, and the size of the quotient graph $G/\M^*$ is larger than or equal to the constant $C$.
Suppose that there is an optimal editing set for $G$ with $opt(G) \leq k$ modifications. 
By assumption, \textbf{branch} returns a safe collection $S$ of editing sets of $G / \M^*$.  By Lemma~\ref{lem:safe-transfer}, $\{ext(G, \M^*, \E) \mid \E \in S\}$ is safe for $G$.  So there is some optimal editing set $\E^*$ of $G$ with $cost(\E^*) = opt(G) \leq k$, and $\E \in S$ such that $ext(G, \M^*, \E) \subseteq \E^*$.  
This means that $G \Delta ext(G, \M^*, \E)$ can be made $\H$-free optimally with $cost(\E^*) - cost(ext(G, \M^*, \E))$ modifications.  The algorithm recurses on $G \Delta ext(G, \M^*, \E)$ at some point, and by induction we may assume that it correctly returns the optimal cost $cost(\E^*) - cost(ext(G, \M^*, \E))$ of an editing set on that graph, and therefore the variable $e$ gets set to a value of at most $cost(\E^*)$.  One can also see that, in the minimization, $e$ is never set to a quantity smaller than $cost(\E^*)$. 
This is because every recursive call on some $\E$ returns either $INF$, or the cost $e_\E$ of an $\HF$ editing set of $G \Delta ext(G, \M^*, \E)$.  The value $cost(ext(G, \M^*, \E)) + e_\E$ used in the minimization is therefore the cost of some $\HF$ editing set.  In other words, all values considered in the minimization correspond to costs of editing sets or $INF$, so none of these can be smaller than the optimal $cost(\E^*)$.  It follows that $e = cost(\E^*)$ and that we return the correct value.

Next suppose that $G$ cannot be made $\HH$-free with a cost of at most $k$ modifications.  This implies that for any $\E$ in the \textbf{branch} set,
there is no $\HF$ editing set with cost at most $k-cost(ext(\E, G, \M^*))$ for $G \Delta ext(\E, G, \M^*)$.
By induction, every call on $f$ returns $INF$ and therefore $e$ is set to $INF$ and we return the correct value.

We then assume that the algorithm makes recursive calls on Line~\ref{line:else-solve-modules}. 
The correctness follows from Lemma~\ref{lem:opt-g}, since  $opt(G / \M^*) + \sum_{M_i \in \M^*} opt(G[M_i])$ is the cost of an optimal editing set and this is what we compute on that line.  
More specifically, if there exists an editing set of $G$ of size at most $k$, then the same is true for $G / \M^*$ and each $G[M_i]$ subgraph.  In that case, $opt(G / \M^*)$ is computed correctly using \textbf{solve} by assumption, and each $opt(G[M_i])$ is correct by induction (in particular, none of these is $INF$).   
Thus $e$ is set to the correct value.  
If no editing set of size at most $k$ exists, then if \textbf{solve} or a recursive call returns $INF$, we correctly return $INF$.  If \textbf{solve} and all recursive calls return some value, then again $e$ is set to the minimum size of an editing set --- but since that size is strictly greater than $k$, we will correctly return $INF$.  Noting that our arguments apply regardless of the constant $C$, this completes the proof of correctness.

% Once again assuming that there is an optimal editing set $X$ for $G$ with at most $|X| \leq k$ modifications. We now additionally know that the constant size of $G/\M^*$ allows for it to be solved trivially, and \textbf{solve} returns the optimal solution set. If less or equal to $k$, we return the size of this set in addition to the summation of all calls of $f$ on each modules, which we assume return the optimal edition set in each cases. By the properties of the quotient graph and the maximal modular partition, no two vertices of a minimal graph of $\HH$, which must be prime, may under this condition share the same module unless all do. Therefore, either all edges of a prime minimal graph of $\HH$ must be in a module, or they all must be crossing modules, meaning contained in the quotient graph. By solving the quotient graph so that no edge between any two module may be within a minimal graph of $\HH$, then by additionally solving each module separately, we cover the entire graph. Finally, if there are no solution within $k$, this means either \textbf{solve} or any call on $f$ returned $INF$ on line~\ref{line:7}, or that the sum of all editions is superior to $k$. In such case, as we know the graph has been entirely covered, we finally return $INF$ on line~\ref{line:9}.

\medskip 

\noindent
\textbf{Fact 2.}  If some minimal graph in $\H$ is not prime, then for any constant $C$ there exist instances on which Algorithm~\ref{alg:alg} fails.  To see this, take a large enough graph $G$ obtained from Proposition~\ref{prop:alg-incorrect} and any $k \geq 1$.  
Since $V(G / \M^*)$ is $\HF$, the algorithm will call \textbf{solve} on $G / \M^*$ and return $0$ since it is $\H$-free, and the same holds for every $G[M]$.  Thus the algorithm returns $0$, which is incorrect since $G$ is not $\HF$.

\medskip 

\noindent
\textbf{Fact 3.}  Algorithm~\ref{alg:alg} runs in time $O(n b^k (g(n) + h(n) + n + m))$, with $g(n)$ the time to verify $G \in \HF$, $h(n)$ is the running time of \textbf{branch}, and $b$ is the branching factor of \textbf{branch}.

Consider the time spent in one call to the algorithm, without counting recursive calls.  First checking if $G \in \HF$ takes time $O(g(n))$.  
If we call \textbf{solve}, then $|V| < C$ for a constant $C$ and thus this takes constant time.  Otherwise, we need to compute $\M^*$ in time $O(n + m)$, and computing $G / \M^*$ takes the same complexity.  Checking for $G / \M^* \in \HF$ again takes time $O(g(n))$, and in case it is not $\HF$, calling \textbf{branch} adds a time of $h(n)$. 
We then enumerate each editing set returned by \textbf{branch}.  We assume there is a constant number of such sets, and for each we compute $G - ext(G, \M^*, \E)$ in time $O(n + m)$, which does not add to the complexity.  
If instead we call \textbf{solve} and recurse on the modules, the former takes constant time and computing the induced subgraphs in time $O(n + m)$.  Overall, one call takes time $O(g(n) + h(n) + n + m)$.

It remains to bound the number of recursive calls made.  Let $t(n, k)$ be the number of leaves of the recursion tree when given as input a graph with $n$ vertices and parameter $k$.  We prove by induction on the depth of the recursion tree that $t(n, k) \leq n \cdot b^k$, where we recall that $b$ is the branching factor of \textbf{branch}.
As a base case, a terminal call has one leaf in its recursion tree and the statement holds.  
So assume that recursive calls are made and that the statement is true in each such call.  
We note that as Line~\ref{line:if-solve-gm} and Line~\ref{line:else-solve-modules} never occur simultaneously, we will prove them individually.
If the algorithm enters Line~\ref{line:if-solve-gm}, it calls the \textbf{branch} function which returns editing sets of sizes $a_1, \ldots, a_p$, and then makes a recursive call on each of them, where the passed parameters on each call are $k - a_1, \ldots, k - a_p$.  Using induction, we thus get:
\begin{alignat*}{2}
    t(n,k) &\leq t(n,k-a_1) + t(n,k-a_2) + \dots + t(n,k-a_p)\\
    &\leq n\cdot b^{k-a_1} + n\cdot b^{k-a_2} + \dots + n\cdot b^{k-a_p}\\
    &= n\cdot \sum_{i=1}^p b^{k-a_i}.
\end{alignat*}
Now, recall that $b$ is the branching factor of the worst branching vector $(c_1, \ldots, c_q)$ of \textbf{branch}.  That is, $b$ is the largest real that is a solution to 
$b^k - \sum_{i=1}^q b^{k - c_i} = 0$.  So, $b^k = \sum_{i=1}^q b^{k - c_i}$.  Since $(c_1, \ldots, c_q)$ is the branching vector possibly returned by \textbf{branch} that results in the largest possible real root, this means that $(a_1, \ldots, a_p)$ is not worse, i.e., the largest real root of $x^k - \sum_{i=1}^p x^{k - a_i}$ is at most $b$.  Thus when $x = b$, that expression is positive and it follows that $b^k \geq \sum_{i=1}^p b^{k - a_i}$.
Thus we obtain $t(n, k) \leq n \cdot b^k$ as desired.

If instead the algorithm enters Line~\ref{line:else-solve-modules}, it makes a recursive call on each module of $\M^* = \{M_1, \ldots, M_p\}$, with $k$ unchanged.  Denote by $n_i = |M_i|$ for $i \in \{1, \ldots, p\}$.  Using induction and the fact that modules are disjoint, we get:
\begin{alignat*}{2}
    t(n,k) &= t(n_1,k) + t(n_2,k) + \dots + t(n_p, k)\\
    &\leq n_1\cdot b^k + n_2\cdot b^k + \dots + n_p\cdot b^k\\
    &= b^k\cdot \sum_{i=1}^p n_i\\
    &= n\cdot b^k\\
\end{alignat*}
The number of nodes in the recursion tree is at most twice $nb^k$, and multiplying the time per call gives the complexity given in the statement.
\qed
\end{toappendix}

\section{Using unavoidable induced subgraphs for cograph deletion}

We now shift from our generic view to the specific \textsc{Cograph Deletion} problem. We apply Algorithm~\ref{alg:alg} to obtain time $O^*((2 + \epsilon)^k)$.  Recall that cographs are exactly the $P_4$-free graphs, so we are solving the $\H$-free Deletion problem with $\H = \{P_4\}$.    Every graph in $\H$ is prime, so by Theorem~\ref{thm:algocorrect} the algorithm is correct.  
Note that Algorithm~\ref{alg:alg} uses a certain constant $C$, which we leave unspecified for now, but should be thought of as ``as large as needed'' (without ever depending on the size of the input graph).
According to Theorem~\ref{thm:algocorrect}, the complexity of the algorithm depends on the time taken to decide if a given graph is a cograph, 
which can be done in linear time~\cite{habib2005simple}, and on the running time and branching factor of a black box \textbf{branch} function.  We now unravel this black box for \textsc{Cograph Deletion}.

Importantly, the algorithm only calls \textbf{branch} on a quotient graph $G / \M^*$.
If $G$ or its complement is disconnected, then $G / \M^*$ is an independent set or a clique and it is a cograph.  We assume that this is not the case, and therefore $G / \M^*$ is a prime graph, a property that we will exploit.
From here on, we therefore let $G$ be the \emph{prime} graph given as input to the \textbf{branch} function.

Some remarks are needed before proceeding.  
Our goal is to return a safe set of deletion sets
with the smallest possible branching factor.  In most cases, we will provide a list of deletion sets $S$ such that \emph{any} optimal solution must perform the deletions that are in one element of $S$. When describing an element $\E$ of $S$, we may say that we \emph{delete} an edge, meaning that we put it in $\E$, or that we \emph{conserve} some edge.  The latter means that we assume that an optimal deletion set does not delete that edge, which often allows further deductions. 
As an example, if $G$ contains a $P_4$ $w - x - y - z$, then one can easily see that \emph{any} optimal deletion set must do one of the following: either (1) delete $wx$; (2) conserve $wx$ and delete $xy$; (3) conserve both $wx, xy$ and delete $yz$.  This describes a set $S$ with three deletion sets, each one containing one edge to delete, and thus with corresponding branching vector $(1, 1, 1)$.
Occasionally, we will argue that there \emph{exists} some optimal deletion set that has the same deletions as some set of $S$.

Recall that $G$ is assumed to be vertex-weighted according to a function $\omega$, and that deleting an edge $uv$ costs $\omega(u) \omega(v)$.  In a branching vector $B = (b_1, \ldots, b_p)$, it is always advantageous to replace an entry $b_i$ with a larger number $b'_i > b_i$, since $B$ will dominate the resulting vector.  Therefore, deleting an edge of weight larger than $1$ is better than deleting an edge of weight $1$, and for that reason we shall assume that all edges have weight $1$, unless weights need to be taken into account. 

% A bounded search tree in fixed-parameter tractability theory is an algorithmic technique that evaluates the complexity of editing a graph by branching on all possible prospects of an optimal solution while avoiding repeating states. This is hereby achieved by declaring a previously edited edge either permanent or forbidden in later branches. Consequential complexity gain is obtained by reducing the branching vectors length while increasing their branching factors, resulting in favorable, sufficiently low, branching numbers. We will show bounded search trees for every unavoidable induced subgraphs, as well as related branching vectors and numbers. 

\subsection{Unavoidable subgraphs and ever-growing branching vectors}

We use the unavoidable induced subgraphs in large prime graphs developed by Chudnovsky et al.~\cite{Chudnovsky:[1]}, and then later by~\cite{malliaris2018unavoidable}. 
Beforehand, we need the following definition.

\begin{definition}
    A \emph{chain} is a set of vertices $\{v_1, v_2, \ldots, v_c\}$ in which, for $i \in \{2, \ldots, c\}$, the vertex $v_i$ is of one of two types: \\
(type 0) $v_i$ is adjacent to every vertex in $v_1, \ldots, v_{i-2}$  and is not adjacent to $v_{i-1}$; or \\
(type 1) $v_i$ is adjacent to $v_{i-1}$ and not adjacent to any of $v_1, \ldots, v_{i-2}$.
\end{definition} 

Note that $v_1$ is of both types, so we assume that its type is assigned arbitrarily unless otherwise specified.
A chain can be described by a string over alphabet $\{0, 1\}$, where the $i$-th character indicates the type of the $i$-th vertex in the chain.  For example, we may write $[111111]$ to denote the chain consisting of $c = 6$ vertices of type $1$, which is in fact a $P_6$.  Chordless paths are thus special cases of chains.
Let us note for later that if $\{v_1, \ldots, v_c\}$ forms a chain, then for any $1 \leq i \leq j \leq c$, $\{v_i, v_{i+1}, \ldots, v_j\}$ also forms a chain.  Moreover, if $B$ is the binary representation of the first chain, then the representation of the smaller chain is the substring of $B$ from position $i$ to $j$.
%To simplify notations, we describe chain vertices using $0$, $1$ and $x$ characters within square brackets, where $0$ is adjacent to all but its predecessor, $1$ is only adjacent to its predecessor and $x$ could be any of the previouses. In our work, the chain is divided in complementary sub-cases.

We can now state the main theorem of~\cite{Chudnovsky:[1]}, and we refer to the figures below for the definition of the induced subgraphs stated therein.

\begin{theorem}[\cite{Chudnovsky:[1]}]\label{thm:chudnov}
    For every integer $c\geq 3$, there exists a constant $C$ such that every prime graph with at least $C$ vertices contains one of the following graphs or their complements as an induced subgraph.

    \begin{enumerate}
      \itemsep0em
      \item The 1-subdivision of $K_{1,c}$ (denoted by $K^{(1)}_{1,c}$) (Figures \ref{fig:us1} and \ref{fig:us2}).  
      \item The line graph of $K_{2,c}$, where $K_{2,c}$ is a complete bipartite graph with two vertices on one side and $c$ on the other (Figures \ref{fig:us3} and \ref{fig:us4}).
      \item The thin spider with $c$ legs, which has an independent set $\{v_1, \ldots, v_c\}$, a clique $\{w_1, \ldots, w_c\}$, and edges $v_i w_i$ for $i \in [c]$ (Figures \ref{fig:us5} and \ref{fig:us6}).  The complement of a thin spider is called a thick spider.
      \item The half-graph of height $c$, denoted $H_c$, which has two independent sets $\{v_1, \ldots, v_c\}, \{w_1, \ldots, w_c\}$ and edge $v_i w_i$ iff $i \leq j$ (Figures \ref{fig:us7} and \ref{fig:us8}).
      \item The graph $H'_{c,I}$, obtained from $H_c$ by making the $w_i$'s a clique and adding a vertex $u$ adjacent to all the $v_i$'s (Figure \ref{fig:us9}).  This graph is self-complementary.
      \item The graph $H^*_c$, obtained from $H'_{c,I}$ by making $u$ only adjacent to $v_1$ (Figures \ref{fig:us10} and \ref{fig:us11}).
      \item A chain with $c + 1$ vertices that induces a prime graph.
    \end{enumerate}
\end{theorem}

\begin{figure}[!htb]
   \begin{minipage}{0.33\textwidth}
     \centering
     \includegraphics[height=2.9cm]{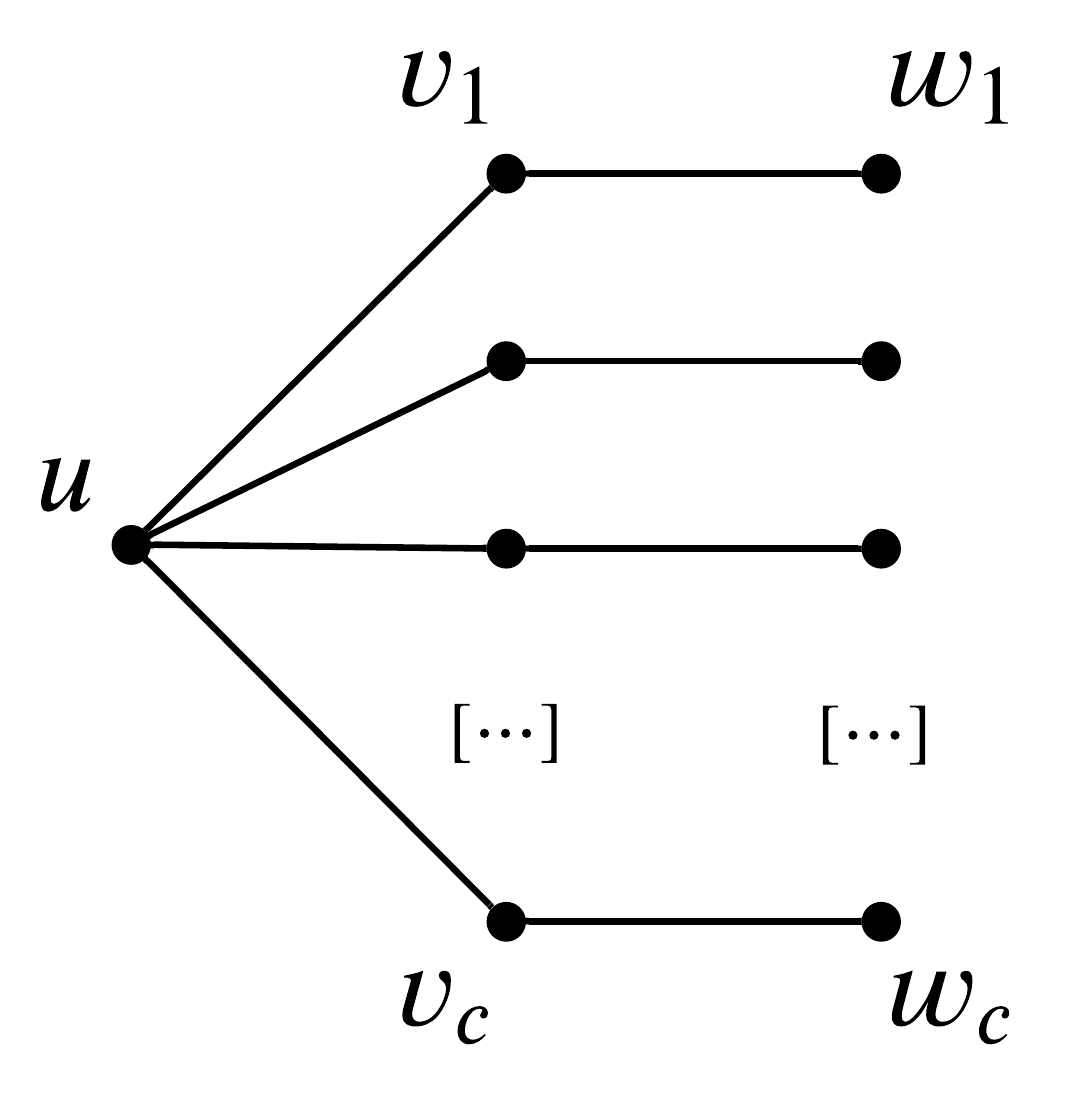}
     \caption{$K_{1,c}$}\label{fig:us1}
   \end{minipage}\hfill
   \begin{minipage}{0.33\textwidth}
     \centering
     \includegraphics[height=3cm]{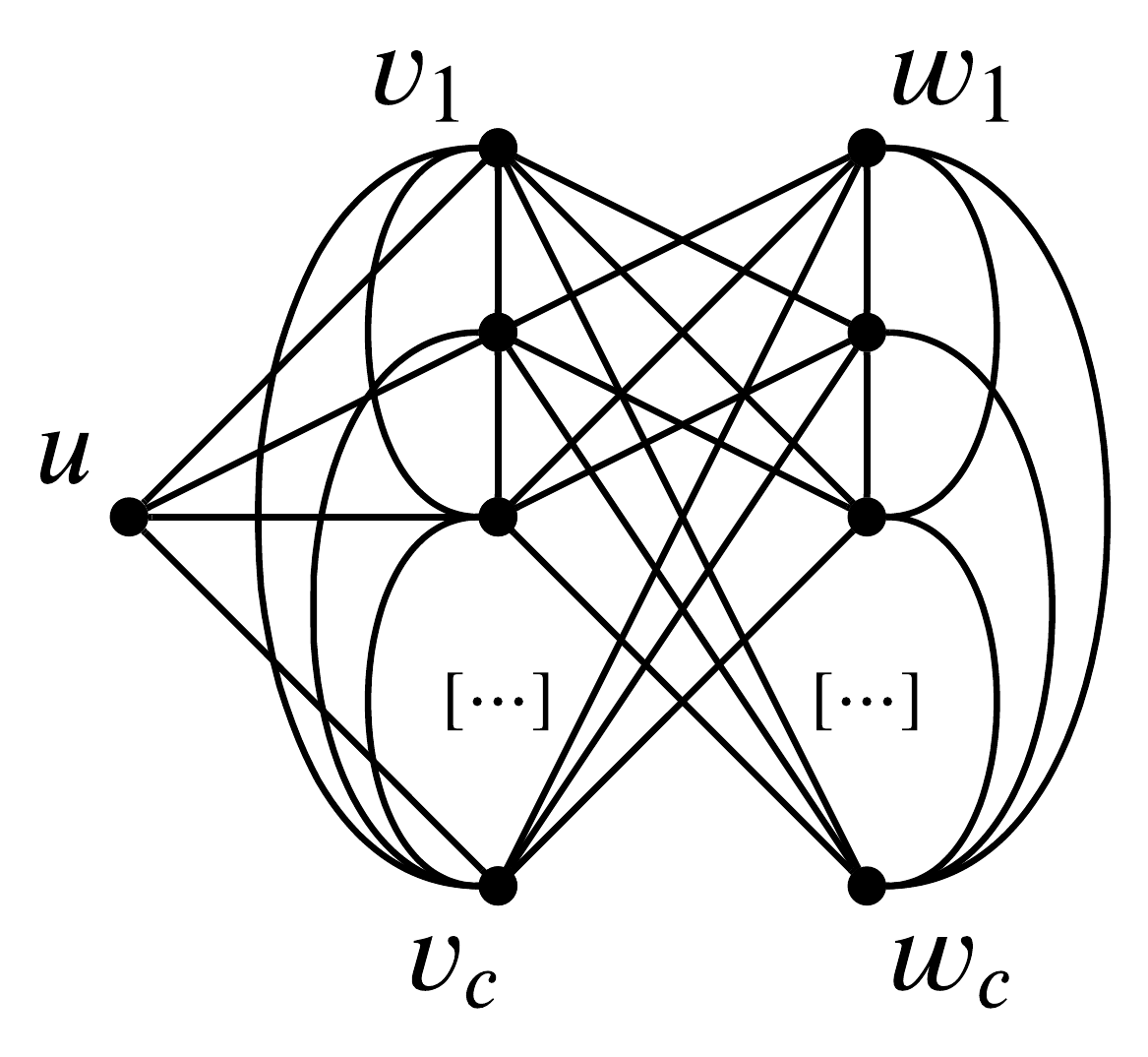}
     \caption{$\overline{K_{1,c}}$ (note, the vertex labels are not in correspondence with Figure~\ref{fig:us1} since the $v_i, w_i$'s were swapped)}\label{fig:us2}
   \end{minipage}\hfill
   \begin{minipage}{0.33\textwidth}
     \centering
     \includegraphics[height=3cm]{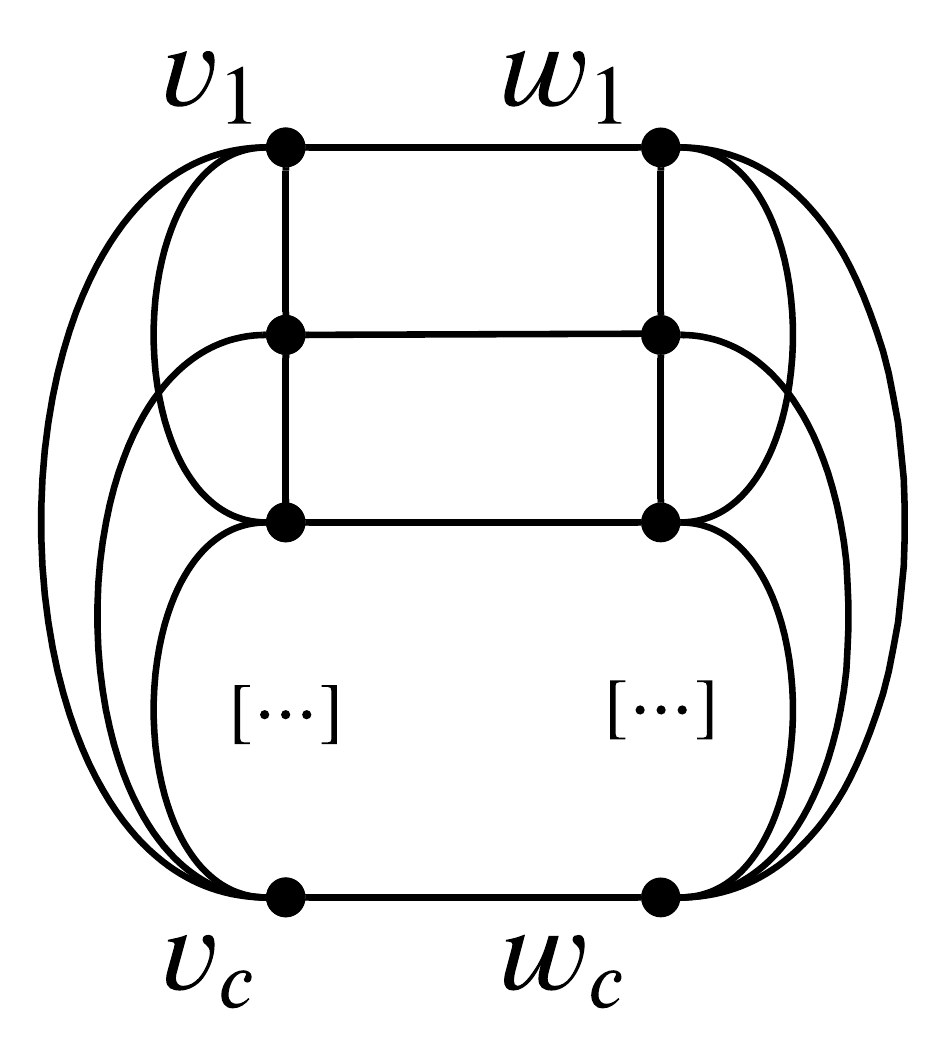}
     \caption{$L(K_{2,c})$}\label{fig:us3}
   \end{minipage}\hfill
\end{figure}
\begin{figure}[!htb]
   \begin{minipage}{0.33\textwidth}
     \centering
     \includegraphics[height=3cm]{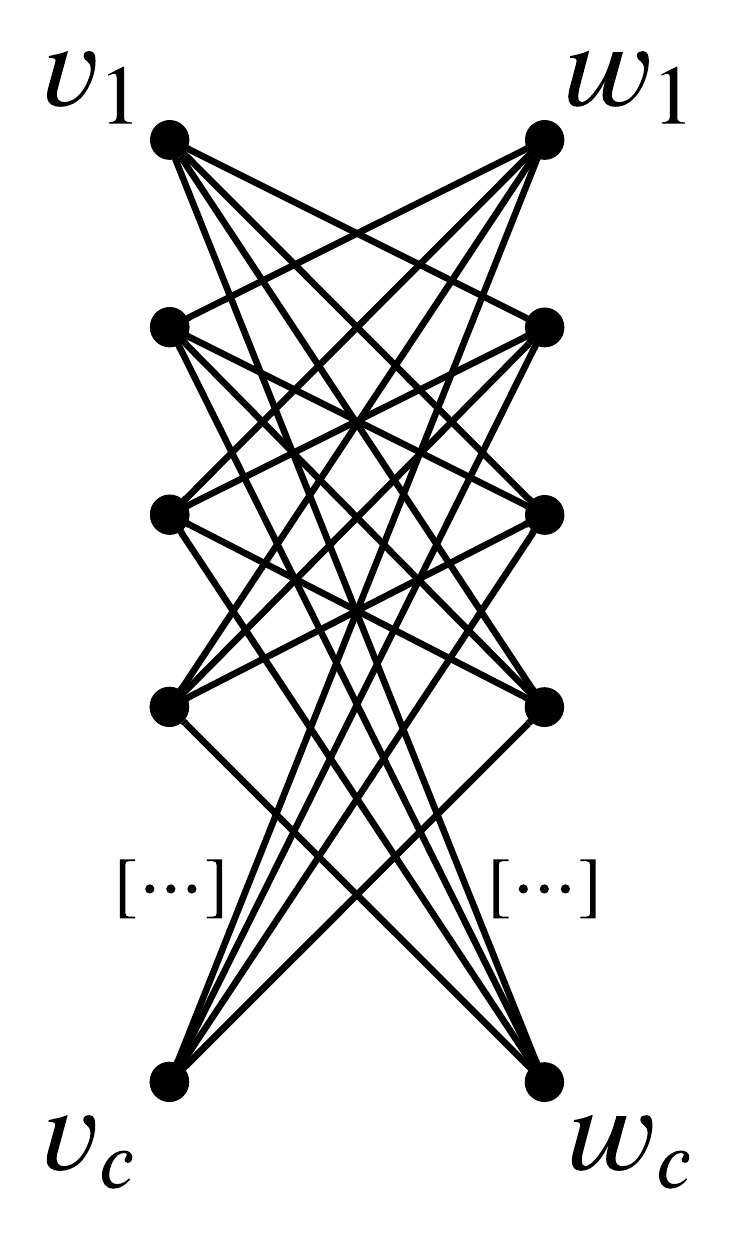}
     \caption{$\overline{L(K_{2,c})}$}\label{fig:us4}
   \end{minipage}\hfill
   \begin{minipage}{0.33\textwidth}
     \centering
     \includegraphics[height=3cm]{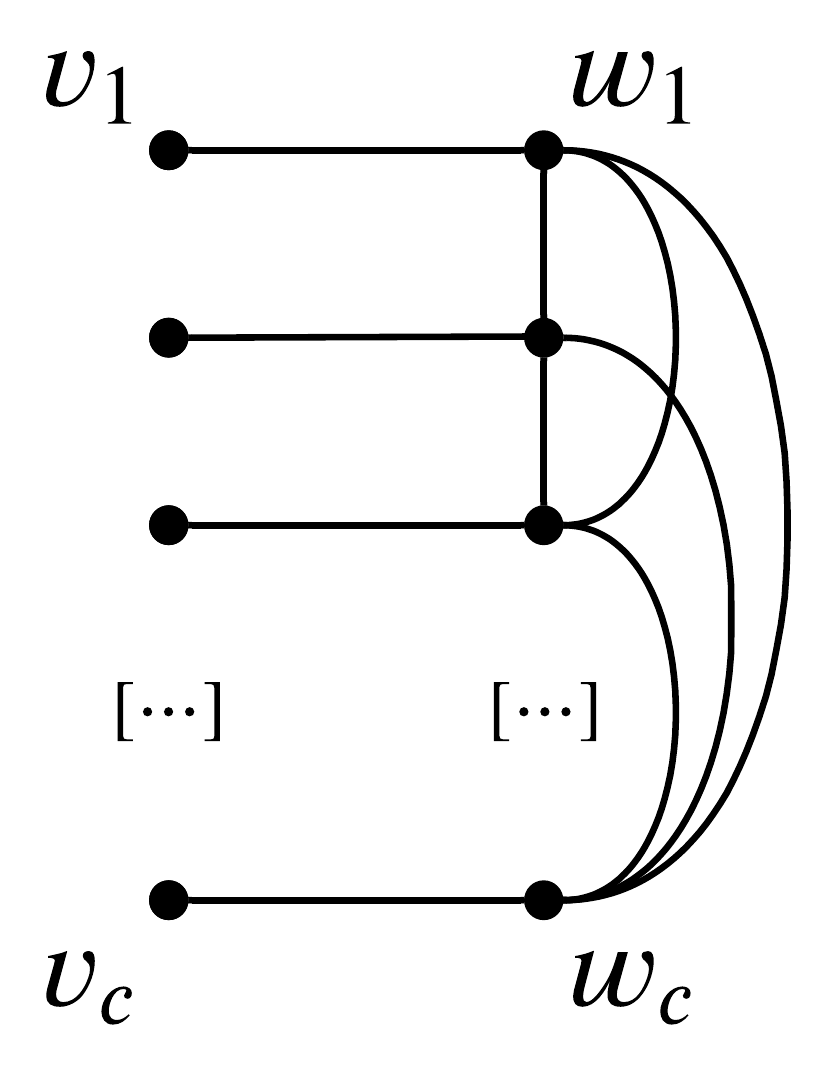}
     \caption{$Thin$ $Spider$}\label{fig:us5}
   \end{minipage}\hfill
   \begin{minipage}{0.33\textwidth}
     \centering
     \includegraphics[height=3cm]{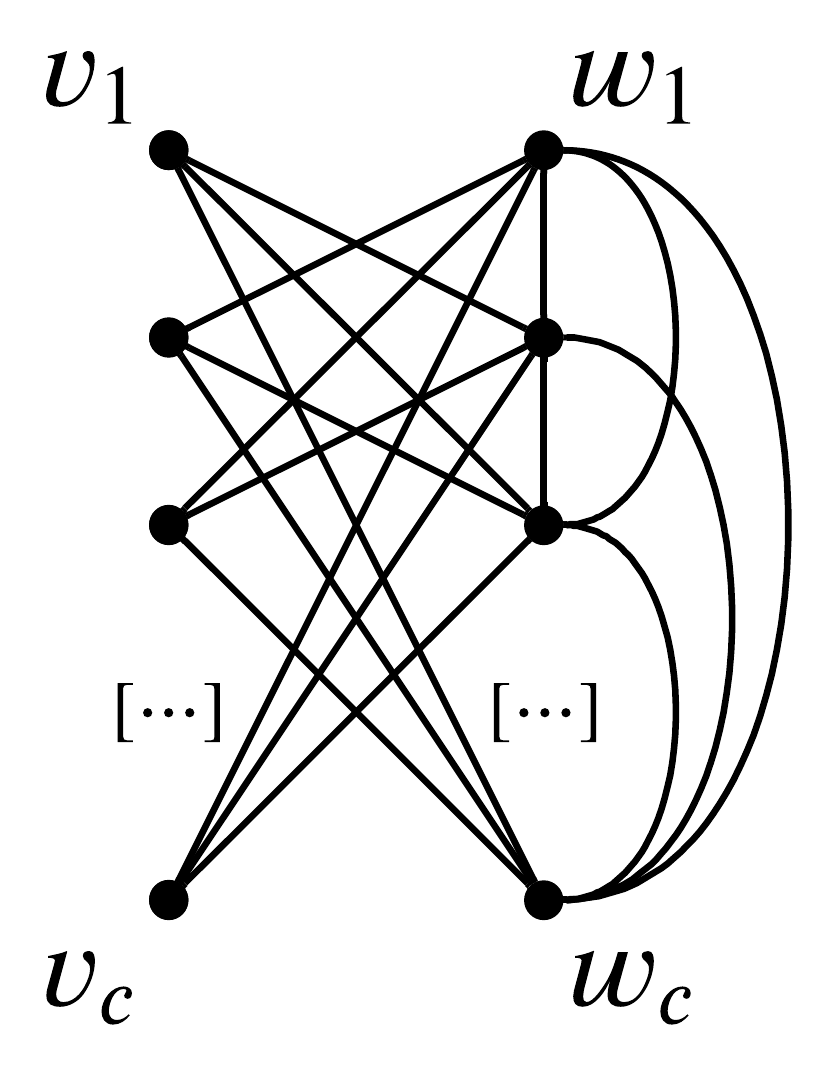}
     \caption{$Thick$ $Spider$}\label{fig:us6}
   \end{minipage}\hfill
\end{figure}
\begin{figure}[!htb]
   \begin{minipage}{0.33\textwidth}
     \centering
     \includegraphics[height=3cm]{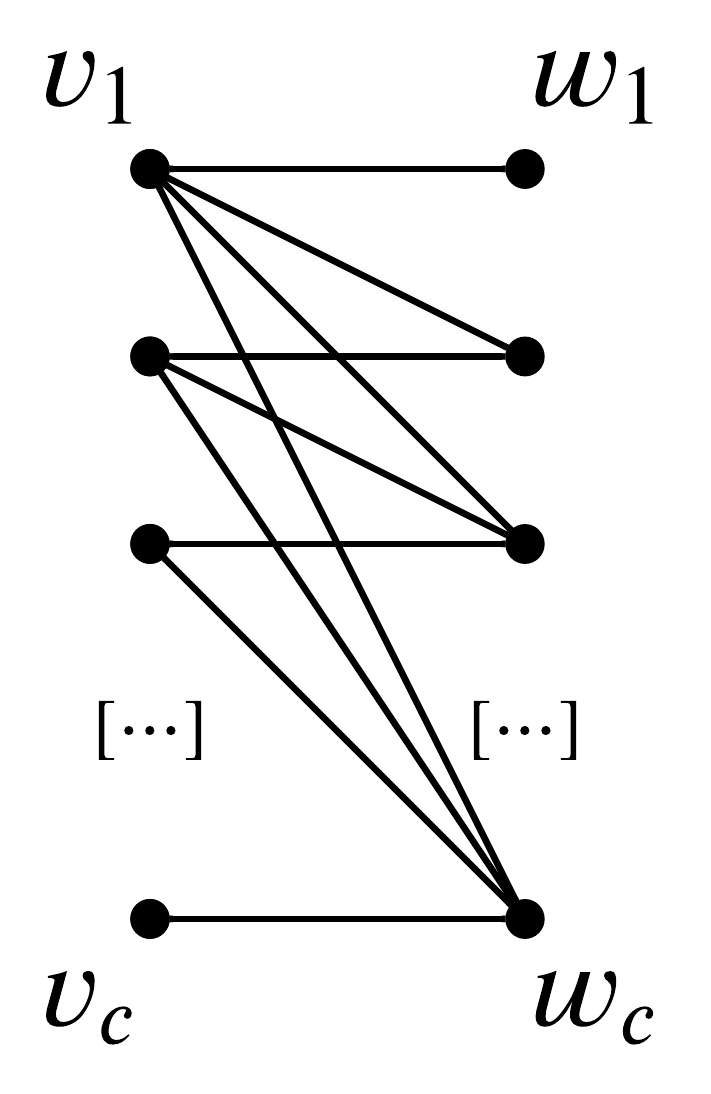}
     \caption{$H_c$}\label{fig:us7}
   \end{minipage}\hfill
   \begin{minipage}{0.33\textwidth}
     \centering
     \includegraphics[height=3cm]{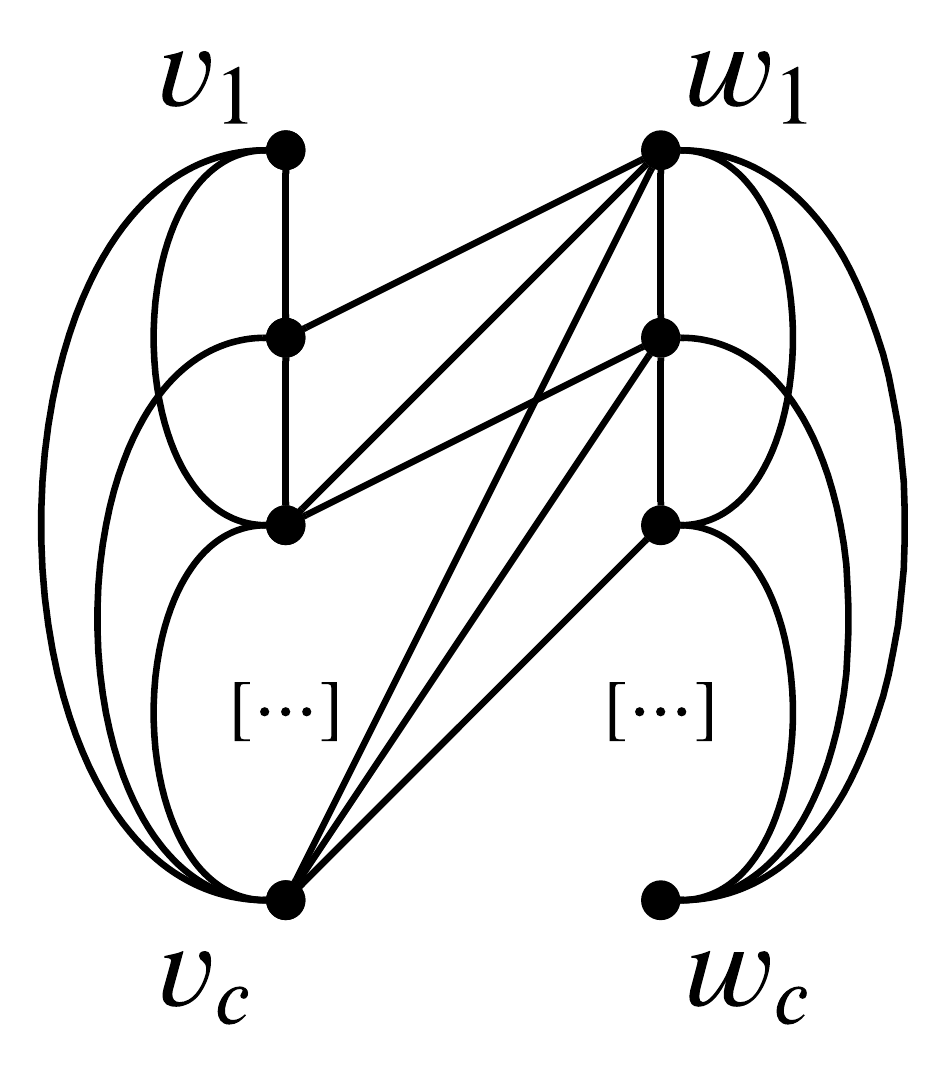}
     \caption{$\overline{H_c}$}\label{fig:us8}
   \end{minipage}\hfill
   \begin{minipage}{0.33\textwidth}
     \centering
     \includegraphics[height=3cm]{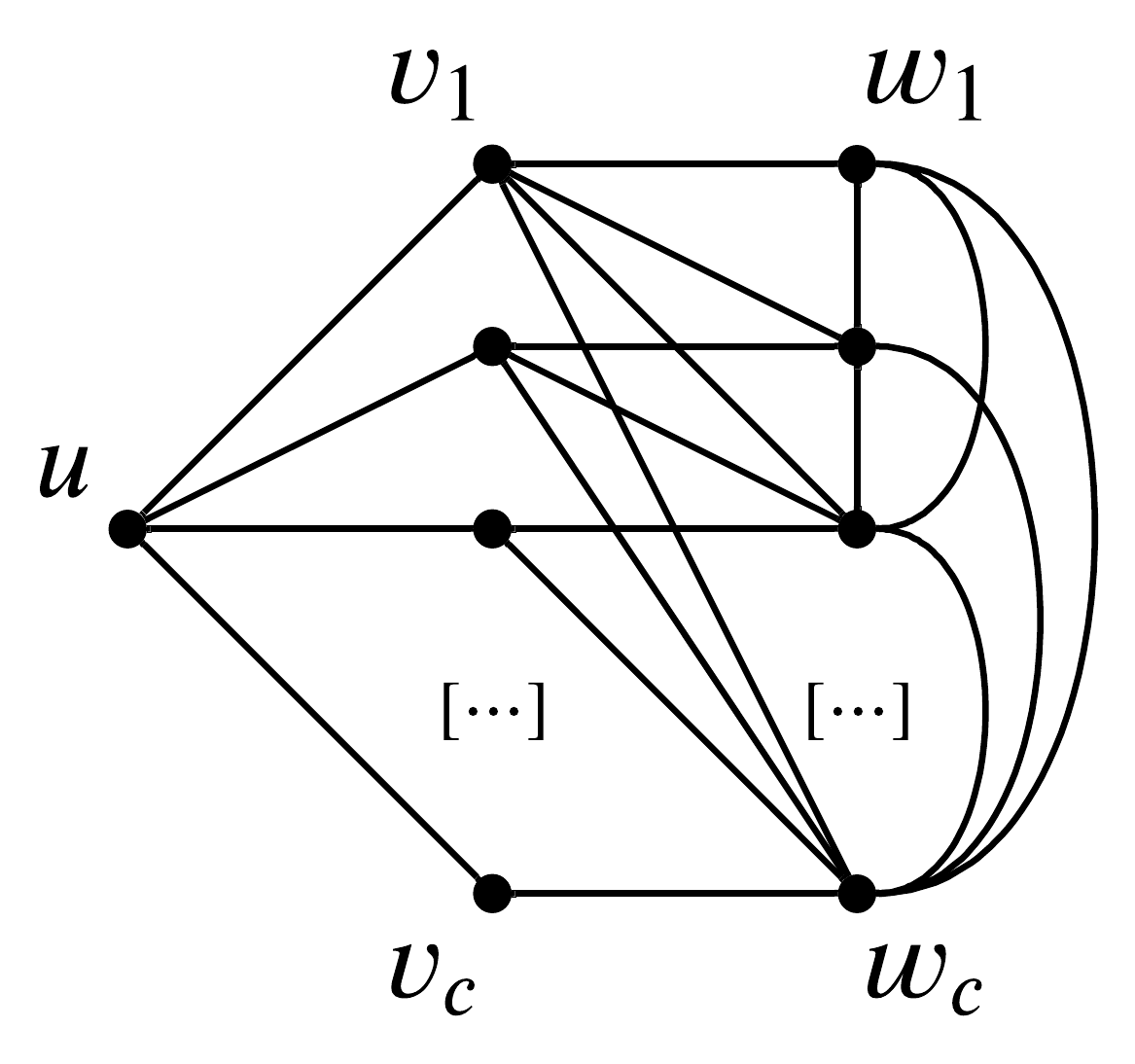}
     \caption{$H'_{c,I}$}\label{fig:us9}
   \end{minipage}\hfill
\end{figure}
\begin{figure}[!htb]
   \begin{minipage}{0.48\textwidth}
     \centering
     \includegraphics[height=3cm]{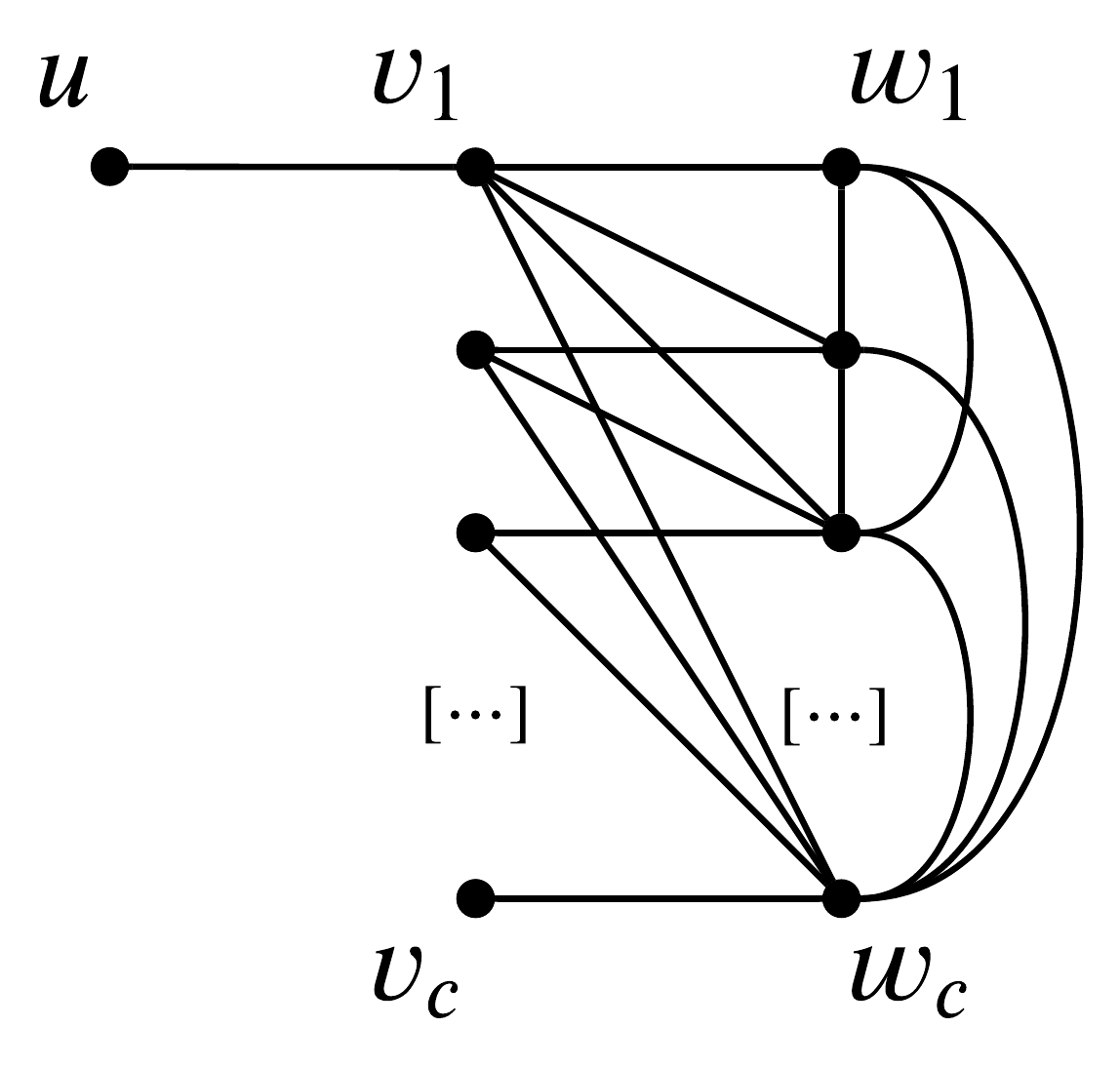}
     \caption{$H^*$}\label{fig:us10}
   \end{minipage}\hfill
   \begin{minipage}{0.48\textwidth}
     \centering
     \includegraphics[height=3cm]{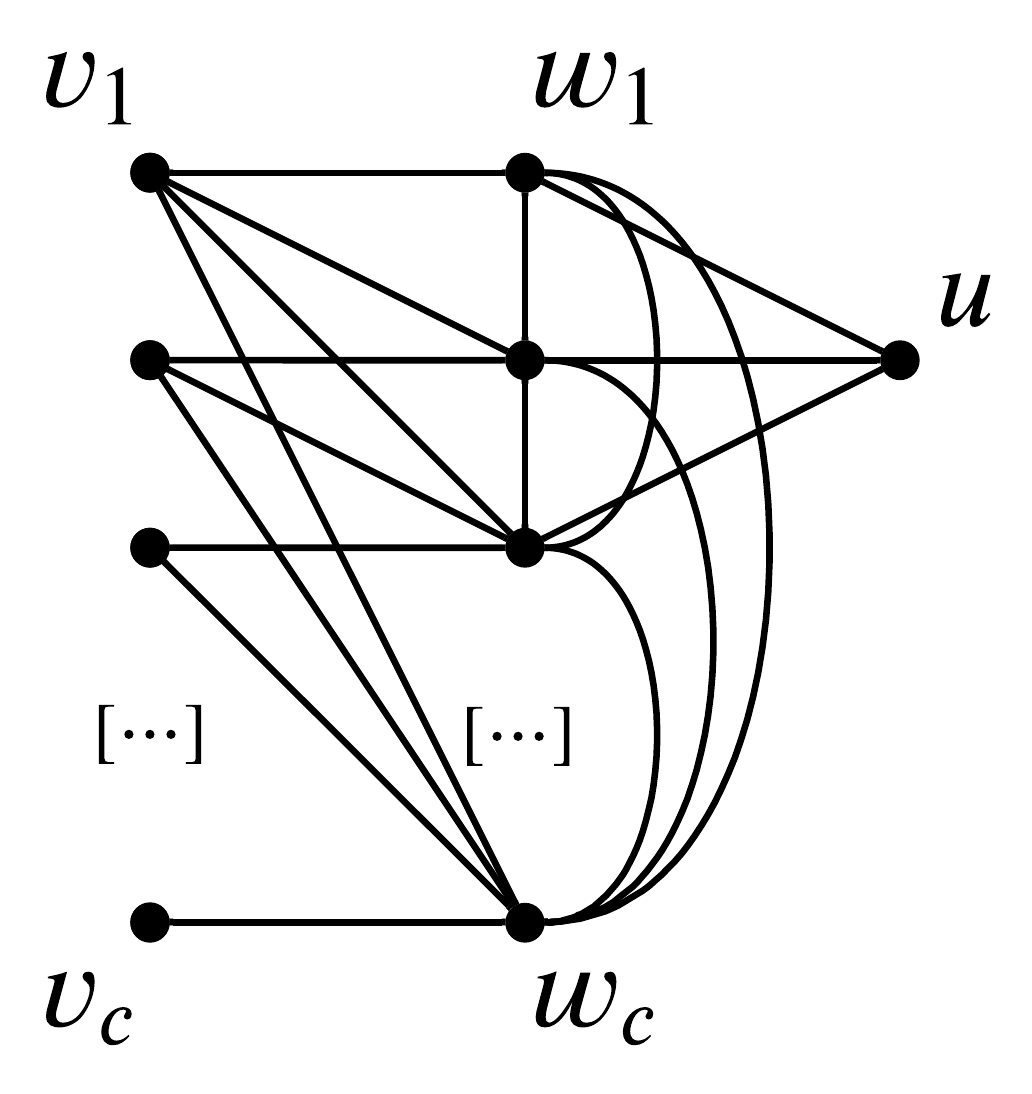}
     \caption{$\overline{H^*}$. 
 Note, to ease future analysis the labels from $H^*$ do not correspond.  The $w_c$ vertex from $\overline{H}^*$ is $u$ from $H^*$, $u$ is $v_1$ from $H^*$, and $w_1, \ldots, w_{c_1}$ is $v_2, \ldots, v_c$.}\label{fig:us11}
   \end{minipage}\hfill
\end{figure}

The theorem allows us to make $c$ as large as desired, at the cost of increasing the constant $C$, which will be used in Algorithm~\ref{alg:alg}.  In~\cite{Chudnovsky:[1]}, the value of $C$ is inferred using a Ramsey argument and is at least exponential in $c$, but it is nonetheless a constant if we assume that $c$ is.
The number of entries in our branching vectors will depend on $c$ and we shall choose specific values later on.  To ease notation, we may describe our branching vectors by concatenating repeated entries.  That is, for an integer $i$, we write $(i)^r$ for the vector that contains the entry $i$ exactly $r$ times, and given  two vectors $\mathbf{a}$ and $\mathbf{b}$, we write $\mathbf{a} + \mathbf{b}$ to denote the vector obtained by concatenating the two vectors.
%Our figures show the structures for very small $c$ values, but in most cases, the complexities are evaluated for an arbitrary large constant growth factor $c$, which tends to infinity. Branching vectors notation, as per this example, is simplified from 
For example, $(1, 1, 1, \underbrace{c-1, c-1, \dots, c-1}_\textrm{$r$ times})$ can be written as $(1)^3+(c-1)^{r}$. 

We now turn to finding efficient safe deletion sets in the unavoidable subgraphs.  When we say that ``one can achieve branching vector $B = (b_1, \ldots, b_p)$'', we mean that it is possible to return, in polynomial time, a set of safe deletion sets whose branching vector is $B$.  
%We use the analogous definition for ``achieving branching factor $b$''.
Importantly, note that our notion of polynomial time assumes that $c$ is a constant, so $n^{O(c)}$ is allowed.

We will encounter achievable vectors of the form $(1, 1) + (c)^{2^c}$, and of the form $(1, 2, \ldots, c) + (\alpha c + \beta)^{\gamma}$, where $c$ can be as large as desired.  We argue that these can yield branching factor $2 + \epsilon$ for any desired $\epsilon$.

\begin{proposition}\label{prop:branch-factors}
    Let $\epsilon > 0$ be any positive real number.  Then the following hold:
    \begin{enumerate}
        \item 
        there exists an integer $c$ such that, for any integer $d \geq c$, the branching vector $(1, 1) + (d)^{2^{d}}$ has branching factor at most $2 + \epsilon$.

        \item 
        for any real $\alpha$ with $0 < \alpha < 1$, any real $\beta$, and any integer $\gamma$, there exists an integer $c$ such that the branching vector $(1, 2, 3, \ldots, c) + (\alpha c + \beta)^{\gamma}$ has branching factor at most $2 + \epsilon$.
    \end{enumerate}
\end{proposition}

\begin{proof}
Focus on the first statement.  Consider the branching vector $(1, 1) + (d)^{2^d}$, whose branching factor is the largest real root of the polynomial $a^k - 2a^{k-1} - 2^d a^{k-d}$.  To simplify, the roots are the same as the polynomial
$f(a) = a^d - 2a^{d-1} - 2^d$.  Observe that for $a > 2$ and any $d > 3$, this function is non-decreasing with respect to $a$ (this can be verified by taking the derivative and noting that it is non-negative for $a \geq 2$).  
Moreover, since $f(2) < 0$, the largest real root is somewhere between $2$ and any $a > 2$ such that $f(a)$ is positive.  

For any fixed $\epsilon > 0$, define $\delta = \epsilon/2$, so that $2 + \epsilon = 2 ( 1 + \delta)$.
Consider the value $a = 2 ( 1 + \delta)$.  Then $f(2(1 + \delta))$ simplifies to $2^d ((1 + \delta)^d - (1 + \delta)^{d - 1} - 1) = 2^d ((1 + \delta)^{d-1} \cdot \delta - 1)$.  Since $\delta > 0$ is fixed and $(1 + \delta)^{d-1} \geq (1 + \delta)^{c-1}$ grows to infinity with $d$ and thus $c$, there is a $c$ large enough such that $(1 + \delta)^{d-1} \cdot \delta - 1 > 0$.  Therefore, for large enough $c$, $f(2(1 + \delta)) = f(2 + \epsilon) > 0$, meaning that the largest real root of $f$ is no more than $2 + \epsilon$, as desired.

Now consider the second statement.  We use a similar approach.  The branching factor of $(1, 2, 3, \ldots, c) + (\alpha c + \beta)^{\gamma}$ is the largest real root of $f(a) = a^c - \sum_{i=0}^{c-1} a^i - \gamma a^{\alpha c + \beta}$.  Here $f(a)$ simplifies to 
\[
a^c \left(1 - \frac{1}{a-1} - \frac{\gamma}{a^{c(1 - \alpha) - \beta}} \right) + \frac{1}{a - 1}.
\]
Note that for any fixed $\epsilon > 0$ and $a = 2 + \epsilon$, the quantity $1 - \frac{1}{a - 1} = 1 - \frac{1}{1 + \epsilon}$ is strictly greater than $0$.  Moreover, with large enough $c$ the expression $\frac{\gamma}{(2 + \epsilon)^{c (1 - \alpha) - \beta}}$ can be made arbitrarily small since $\alpha < 1$ and $\beta, \gamma$ are fixed.  We may thus choose a $c$ such that $f(2 + \epsilon)$ is positive.  By observing that, with this particular $c$, increasing $a$ beyond $2 + \epsilon$ also increases $f(a)$, we deduce that the largest real root is at most $2 + \epsilon$.
\end{proof}

%Where $1$ is a fixed integer, $c$ is the constant growth value, $i$ is a variable over a given interval and $u$ is a distinct vertex with no need of an index, $P4$s sets are provided using example form $v_1-w_c-x_i-u$, which cardinality is simply $i$. We note \cographDeletion{} requires only deletions, and resulting permanent edges are represented by doubled lines. In one case, insertion on the complement graph is used instead, and dotted lines point to previously inserted edges that may no longer be used in further branching.

\subsection{Specific unavoidable subgraphs}

We can now describe our branching algorithm.  We will refer to the graphs numbered 1-6 in Theorem~\ref{thm:chudnov} as \emph{specific graphs}, which have a fixed structure unlike chains.  We consider all possible induced subgraphs that our prime graph $G$ could contain, starting with the specific graphs and then dealing with the more complex case of chains.
In the upcoming figures of graphs, unless stated otherwise, 
we use dotted lines to indicate edges that are deleted, solid doubled lines for edges that we assume are conserved, and for clarity some edges may be grayed out when they are not useful for our explanations.

\begin{lemma}\label{lem:fixed-graphs}
Suppose that a prime graph $G$ contains, as an induced subgraph, either $K_{1,c}$, $L(K_{2,c})$, a Thin Spider, $H_c$, $H'_{c,1}$, $H^*$, or their complements.
Then it is possible to achieve one of the following branching vectors: $(1, 1, c-1)$, $(1)+(c-1)^{2^{c-1}}$, $(1, 1, c-2)$, $(1)+(c-2)^{2^{c-2}}$ and $(1, 1)+(c-1)^{2^{c-2}}$.

Consequently, for any constant $\epsilon > 0$, there is a large enough constant $c$ such that if $G$ contains one of these induced subgraphs, one can achieve branching factor at most $2 + \epsilon$.
\end{lemma}

\begin{proof}
First note that all the graphs mentioned have $O(c)$ vertices.  Thus, if $G$ contains one of the induced subgraphs mentioned, we can find one in time $n^{O(c)}$ by simply enumerating every induced subgraph of appropriate size (perhaps this could be improved to $f(c) n^{O(1)}$, but for now we are not concerned with optimizing polynomial factors).  
%Since $c$ is constant, for each subgraph enumerated we can check if it is isomorphic to one of the subgraphs mentioned in constant time.  
Once a subgraph is found, we return a constant number of deletion sets on that subgraph, so computing those deletion sets does not add to the complexity.  The figures below detail the deletion sets for each induced subgraph.

% In most cases, we either delete an edge $uv$, or decide to conserve it.  That conserved edge is often part of multiple $P_4$'s, for instance $u - v - x_i - y_i$ for some range of $i$ values, whose only common edge is $xy$ (see e.g., when $G$ has an induced $\overline{K_{1,c}}$).  In that case, any solution that conserves $uv$ must delete one of the two other edges $v x_i$ or $x_i y_i$ for each possible $i$.  In such cases, we branch on every combination of choice, that is, two choices per value of $i$, which explains the exponential size of our branching vectors with respect to $c$.

\begin{figure}[!htb]
   \begin{minipage}{0.48\textwidth}
     \centering
     \includegraphics[height=2.9cm]{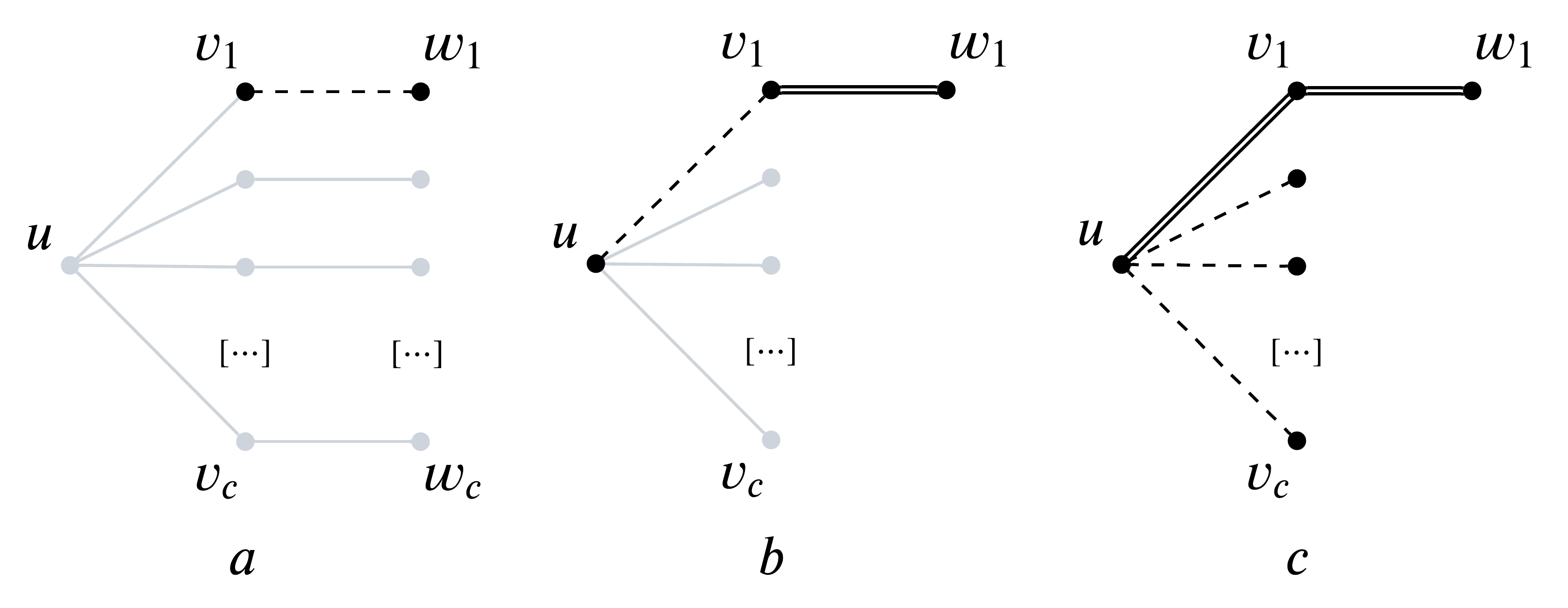}
     \caption{If $G$ has a $K_{1,c}$ as shown, then $w_1-v_1-u-v_i$ is a $P_4$ for each $2\leq i\leq c$.  Either we: (a) delete $w_1v_1$; (b) conserve it and delete $v_1 u$; (c) conserve both $w_1 v_1$ and $u v_1$, enforcing the deletion of $u v_i$ for each $2 \leq i \leq c$. This  results in branching vector $(1, 1, c-1)$.}\label{fig:fig1}
   \end{minipage}\hfill
   \begin{minipage}{0.48\textwidth}
     \centering
     \includegraphics[height=3cm]{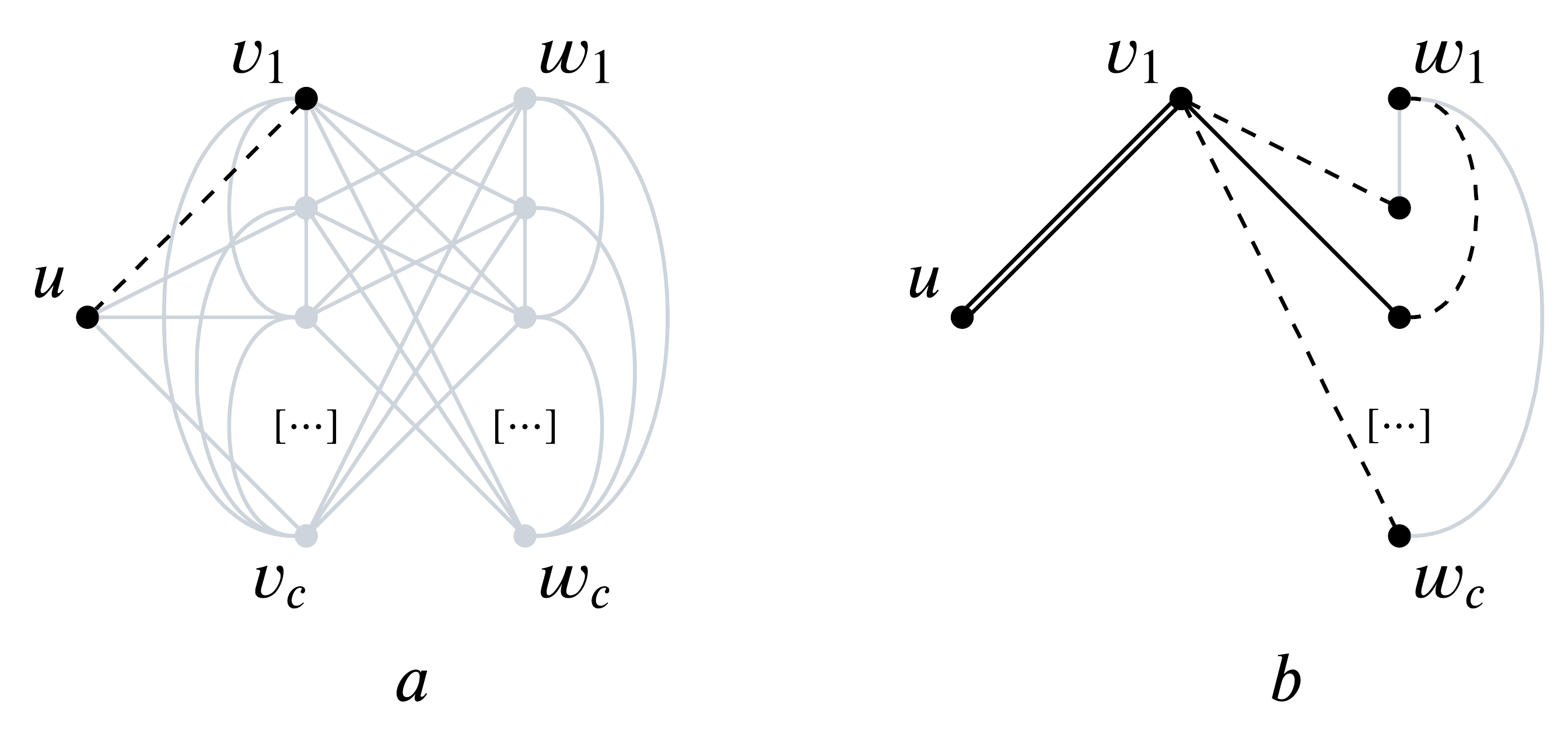}
     \caption{If $G$ has a $\overline{K_{1,c}}$, then $u-v_1-w_i-w_1$ is a $P_4$ for $2\leq i\leq c$.  Either we: (a) delete $u v_1$; (b) conserve $u v_1$.  In that case, for each $i$ with $2\leq i\leq c$, we must delete one of $v_1 w_i$ or $w_i w_1$, so two choices per $i$.  There are $2^{c-1}$ combinations of choices that each delete $c - 1$ edges, and we put each of them in our deletion sets, resulting in branching vector $(1)+(c-1)^{2^{c-1}}$.}\label{fig:fig2}
   \end{minipage}\hfill
\end{figure}
\begin{figure}[!htb]
   \begin{minipage}{0.48\textwidth}
     \centering
     \includegraphics[height=3cm]{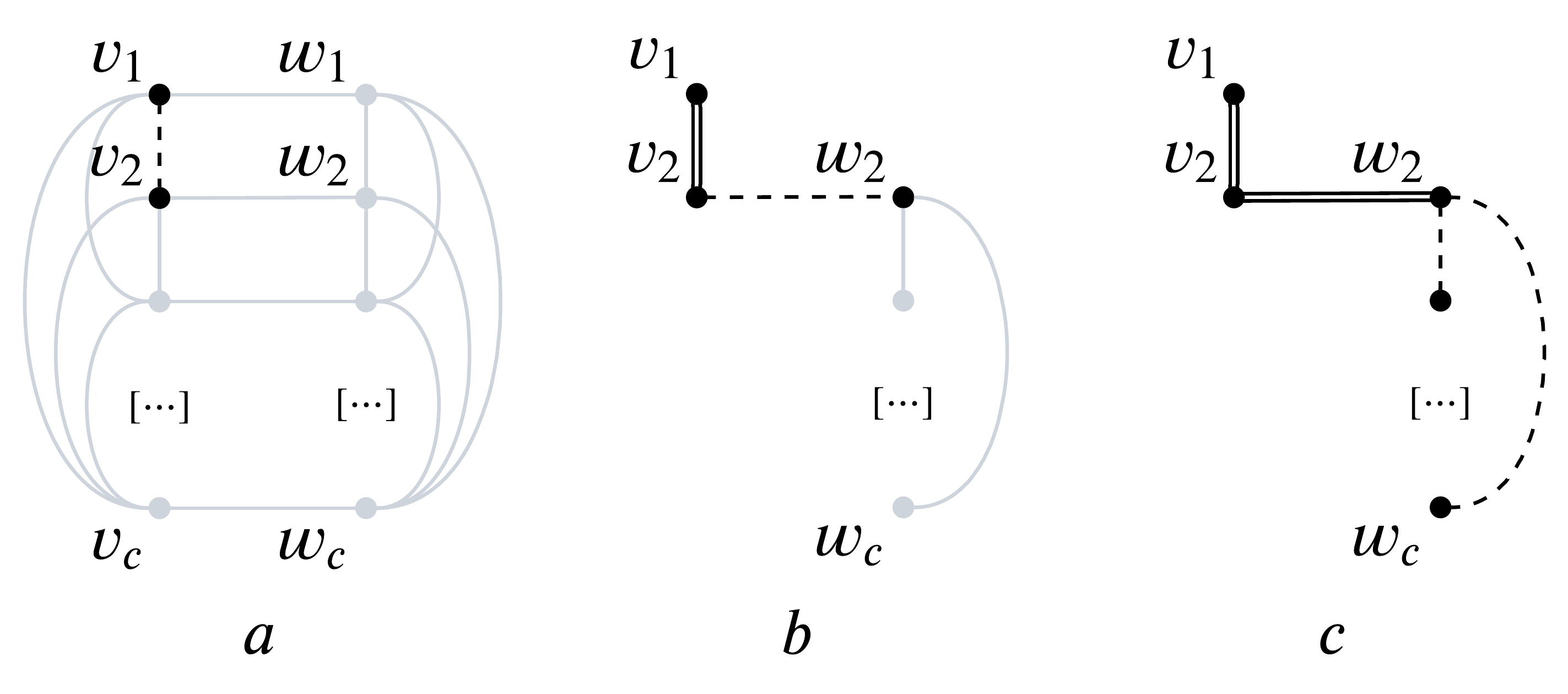}
     \caption{If $G$ has a $L(K_{2,c})$, then $v_1-v_2-w_2-w_i$ is a $P_4$ for $3\leq i\leq c$.  We either: (a) delete $v_1 v_2$; (b) conserve it and delete $v_2w_2$; (c) conserve both $v_1 v_2, v_2 w_2$, enforcing the deletion of $w_2w_i$ for each $3 \leq i \leq c$, resulting in branching vector $(1, 1, c-2)$.}\label{fig:fig3}
   \end{minipage}\hfill
   \begin{minipage}{0.48\textwidth}
     \centering
     \includegraphics[height=3cm]{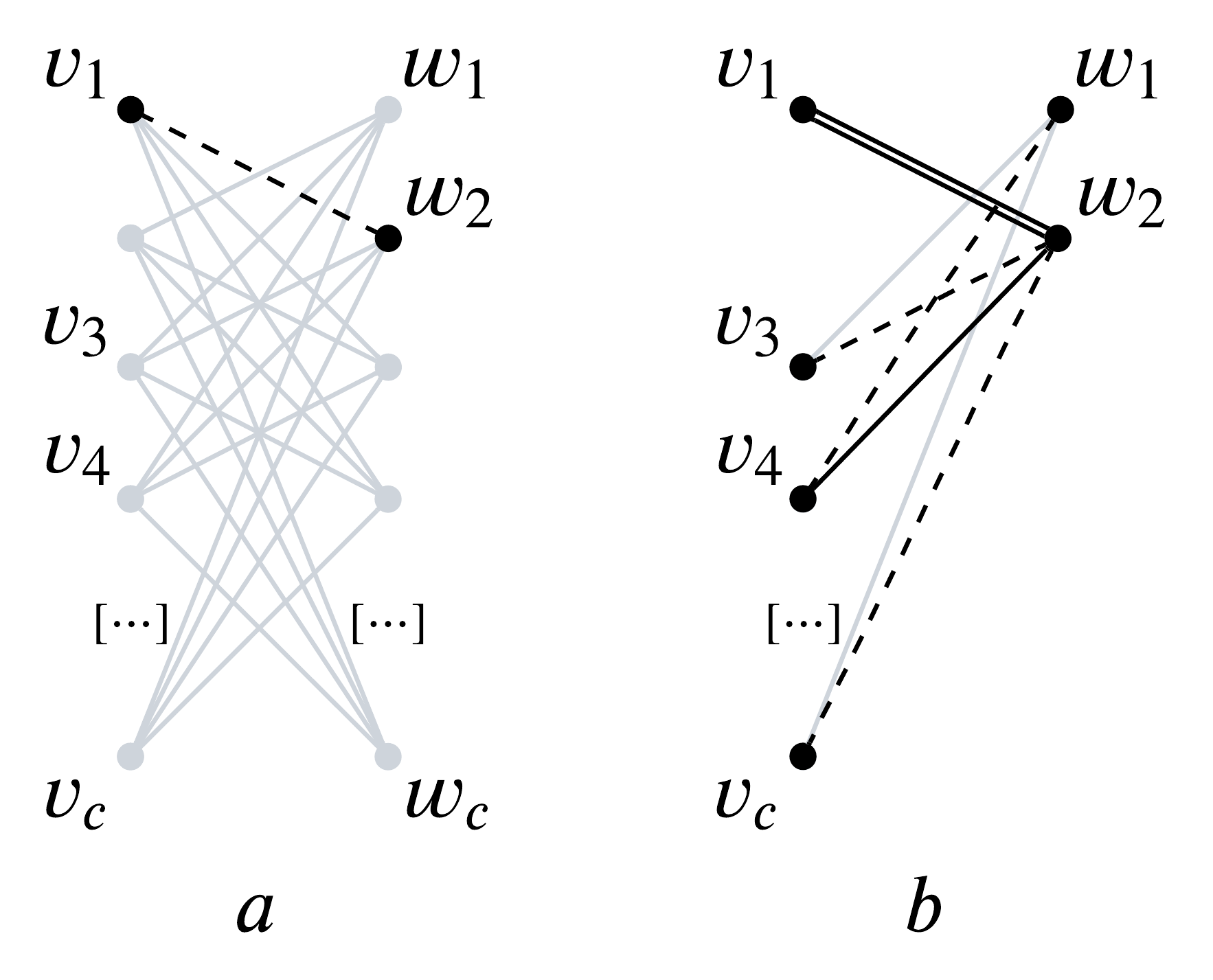}
     \caption{If $G$ has a $\overline{L(K_{2,c})}$, then $v_1-w_2-v_i-w_1$ is a $P_4$ for $3\leq i\leq c$.  We either: (a) delete $v_1 w_2$; (b) conserve it.  Then for each $i$ with $3 \leq i \leq c$, we must delete one of $w_2v_i$ or $v_i w_1$.  We try all combinations, resulting in branching vector $(1)+(c-2)^{2^{c-2}}$.}\label{fig:fig4}
   \end{minipage}\hfill
\end{figure}
\begin{figure}[!htb]
   \begin{minipage}{0.48\textwidth}
     \centering
     \includegraphics[height=3cm]{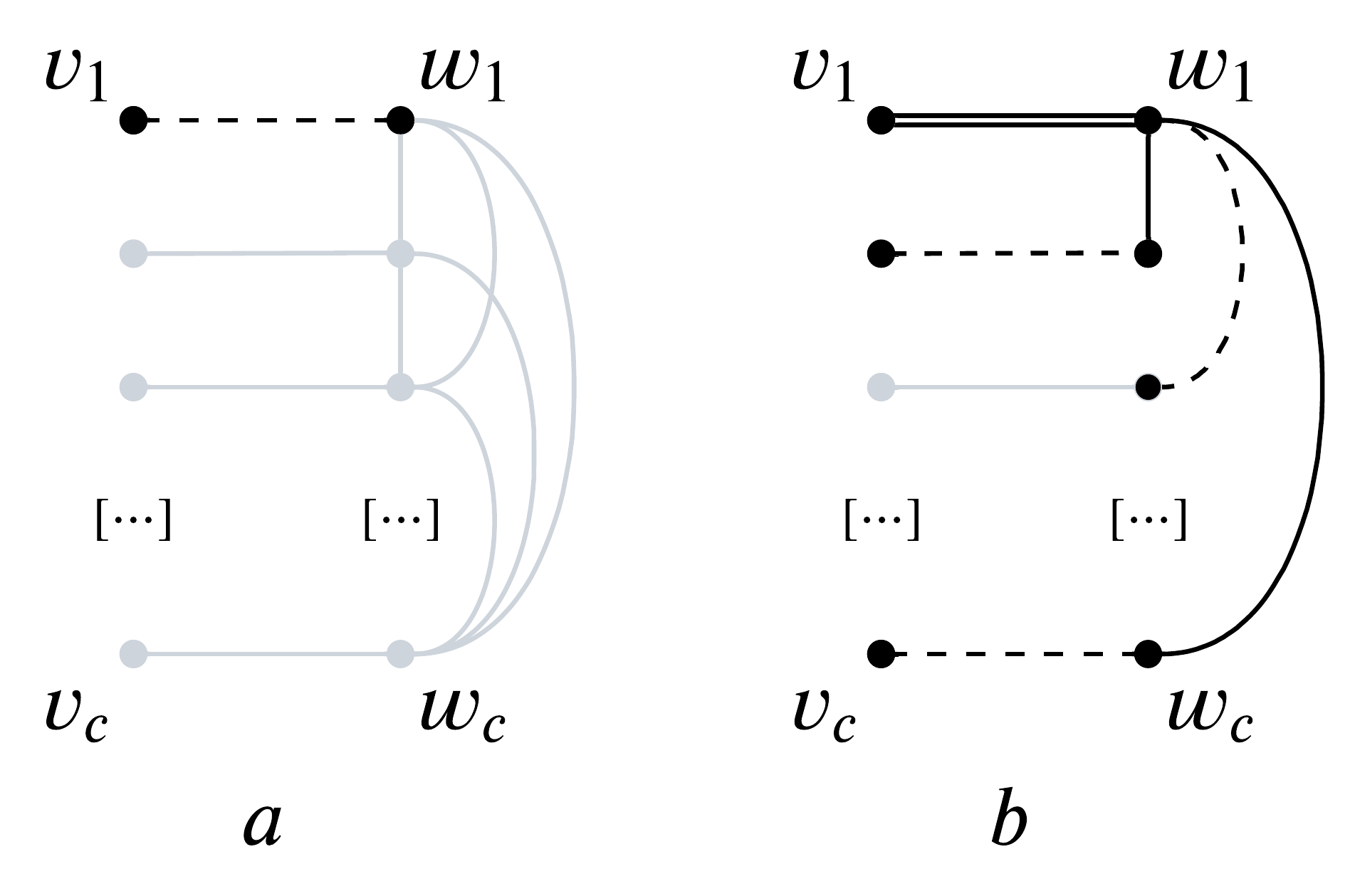}
     \caption{If $G$ has a $Thin$ $Spider$, then $v_1-w_1-w_i-v_i$ is a $P_4$ for $2\leq i \leq c$.  We either: (a) delete $v_1 w_1$; (b) conserve it.  Then for $i$ with $2 \leq i \leq c$, we must delete $w_1 w_i$ or $w_i v_i$, resulting in branching vector $(1)+(c-1)^{2^{c-1}}$.}\label{fig:fig5}
   \end{minipage}\hfill
   \begin{minipage}{0.48\textwidth}
     \centering
     \includegraphics[height=3cm]{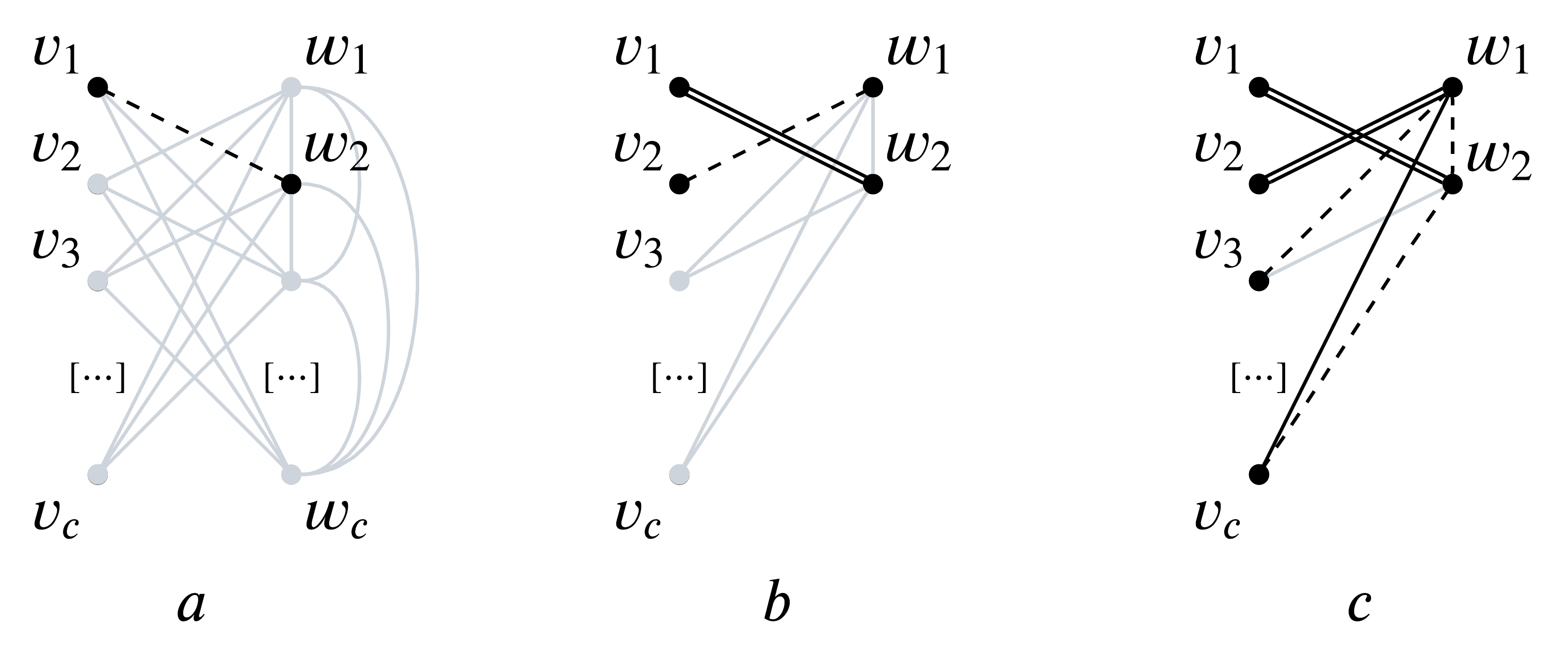}
     \caption{If $G$ has a $Thick$ $Spider$, this case is slightly special.  We have that $v_1-w_2-w_1-v_2$ is a $P_4$.  We either: (a) delete $v_1 w_2$; (b) delete $w_1 v_2$; (c) conserve both.  In that case we are then forced to delete $w_1 w_2$.  In turn, $v_2 - w_1 - v_i - w_2$ is a $P_4$ for $3 \leq i \leq c$.  Since $v_2 w_1$ is conserved, we must delete one of $w_1 v_i$ or $v_i w_2$, resulting in branching vector $(1, 1)+(c-1)^{2^{c-2}}$.}\label{fig:fig6}
   \end{minipage}\hfill
\end{figure}

\newpage

\begin{figure}[!htb]
   \begin{minipage}{0.48\textwidth}
     \centering
     \includegraphics[height=3cm]{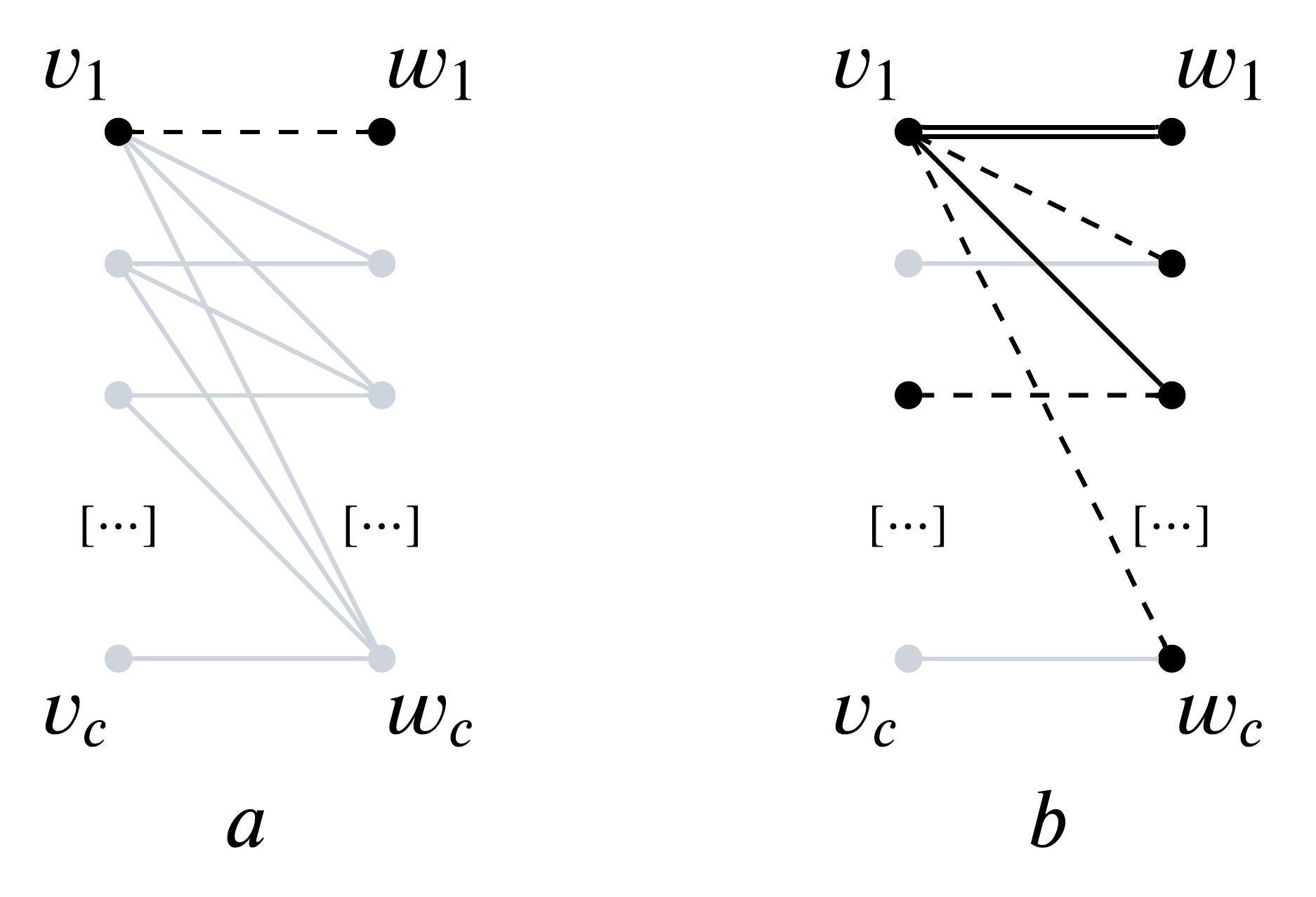}
     \caption{If $G$ has a $H_c$, then $w_1-v_1-w_i-v_i$ is a $P_4$ for  $2\leq i\leq c$.  We either: (a) delete $v_1 w_1$; (b) conserve it.  Then for each $i$ with $2 \leq i \leq c$, we must delete $v_1 w_i$ or $v_i w_i$, resulting in branching vector $(1)+(c-1)^{2^{c-1}}$.}\label{fig:fig7}
   \end{minipage}\hfill
   \begin{minipage}{0.48\textwidth}
     \centering
     \includegraphics[height=3cm]{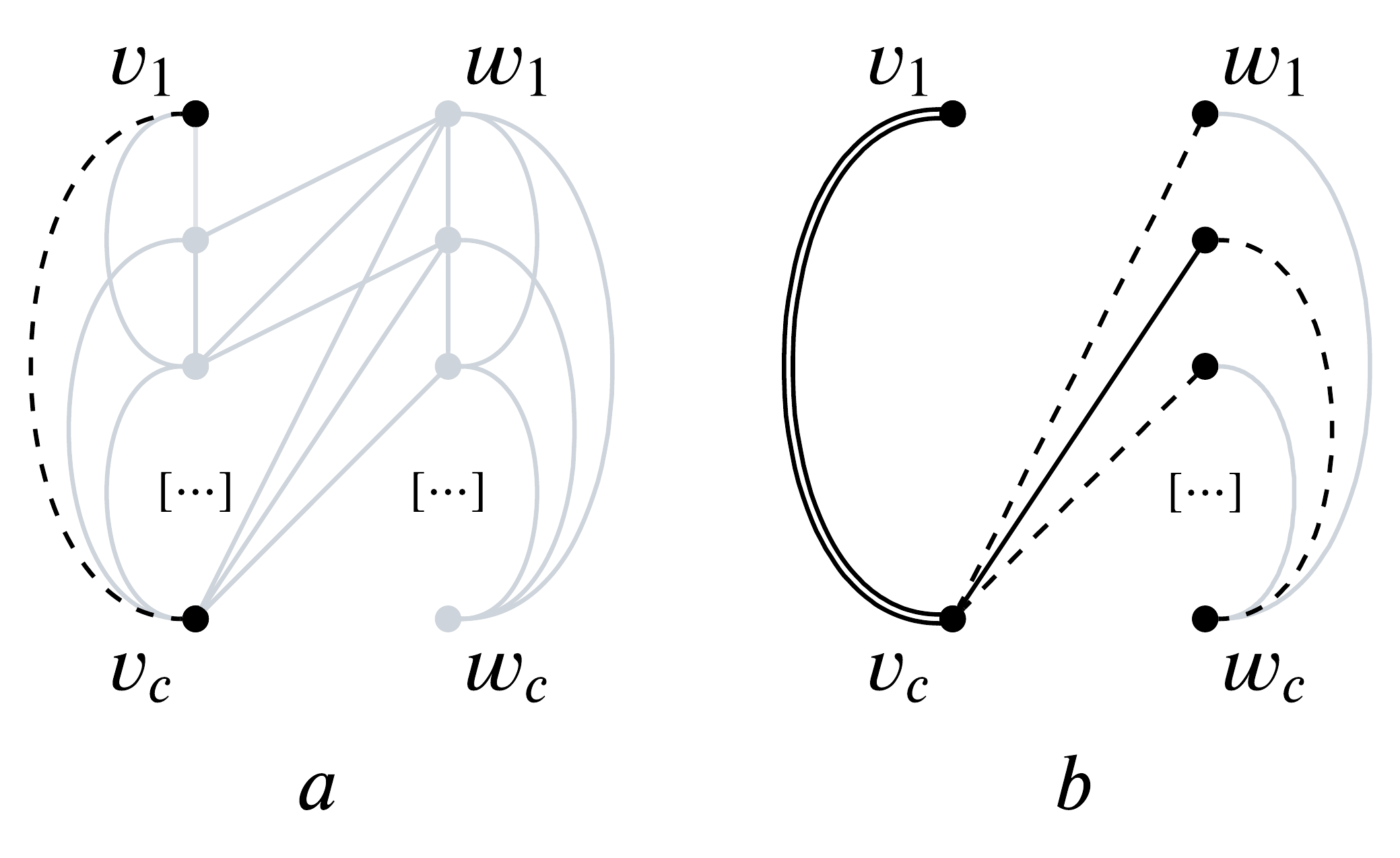}
     \caption{If $G$ has a $\overline{H_c}$, then $v_1-v_c-w_i-w_c$ is a $P_4$ for $1\leq i\leq c-1$.  We either: (a) delete $v_1 v_c$; (b) conserve it.  Then for each $i$ with $1 \leq i \leq c-1$, we must delete $v_c w_i$ or $w_i w_c$, resulting in branching vector $(1)+(c-1)^{2^{c-1}}$.}\label{fig:fig8}
   \end{minipage}\hfill
\end{figure}

%\newpage

\begin{figure}[!htb]
   \begin{minipage}{1\textwidth}
     \centering
     \includegraphics[height=3cm]{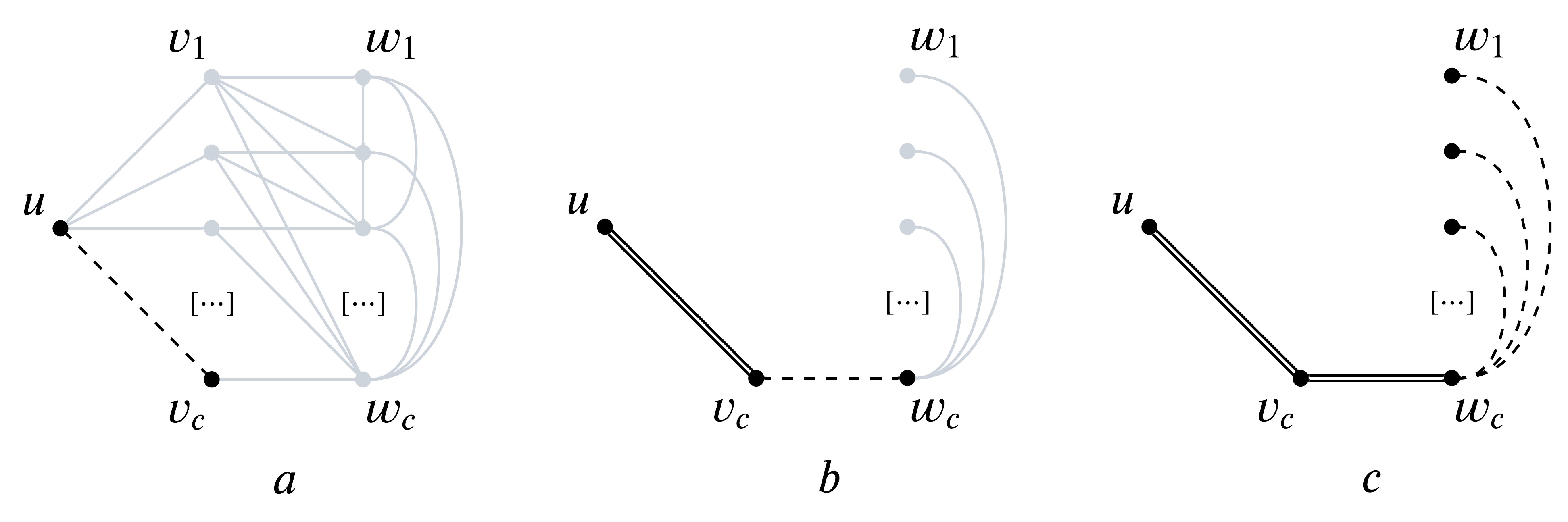}
     \caption{If $G$ has a $H'_{c,I}$, then $u-v_c-w_c-w_i$ is a $P_4$ for $1\leq i\leq c-1$.  We either: (a) delete $u v_c$; (b) delete $v_c w_c$; (c) conserve both, enforcing the deletion of $w_c w_i$ for $1 \leq i \leq c- 1$, resulting in branching vector $(1, 1)+(c-1)$. 
 Note that $H'_{c, I}$ is self-complementary, so the same analysis can be applied if $G$ has a $\overline{H'_{c,I}}$.}\label{fig:fig9}
   \end{minipage}\hfill
\end{figure}
\begin{figure}[H]%[!htb]
   \begin{minipage}{0.48\textwidth}
     \centering
     \includegraphics[height=3cm]{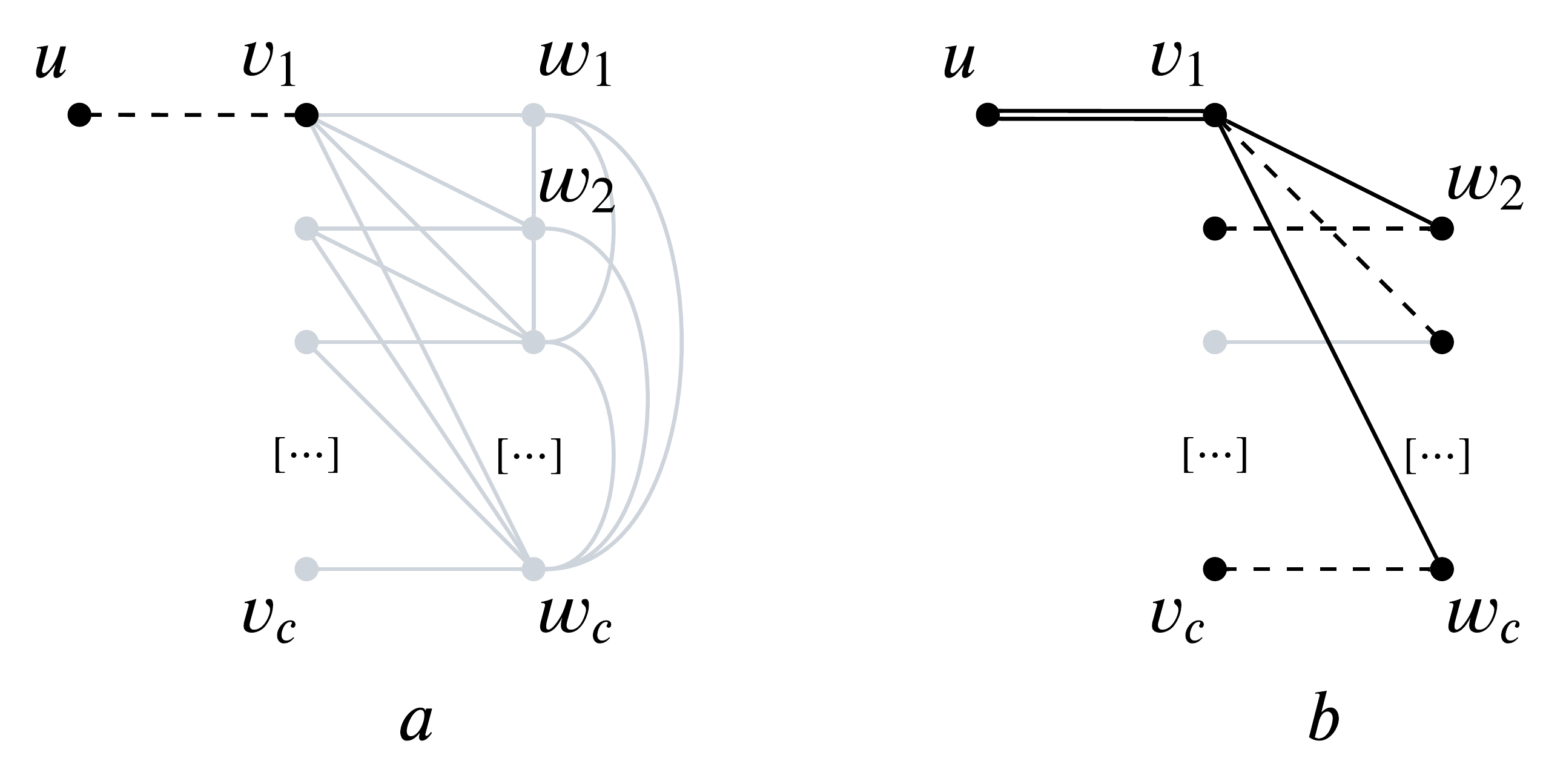}
     \caption{If $G$ has a $H^*$, then $u-v_1-w_i-v_i$ is a $P_4$ for $2\leq i\leq c$.  We either: (a) delete $u v_1$; (b) conserve it.  Then for each $i$ with $2 \leq i \leq c$, we must delete $v_1 w_i$ or $v_i w_i$, resulting in branching vector $(1)+(c-1)^{2^{c-1}}$.}\label{fig:fig10}
   \end{minipage}\hfill
   \begin{minipage}{0.48\textwidth}
     \centering
     \includegraphics[height=3cm]{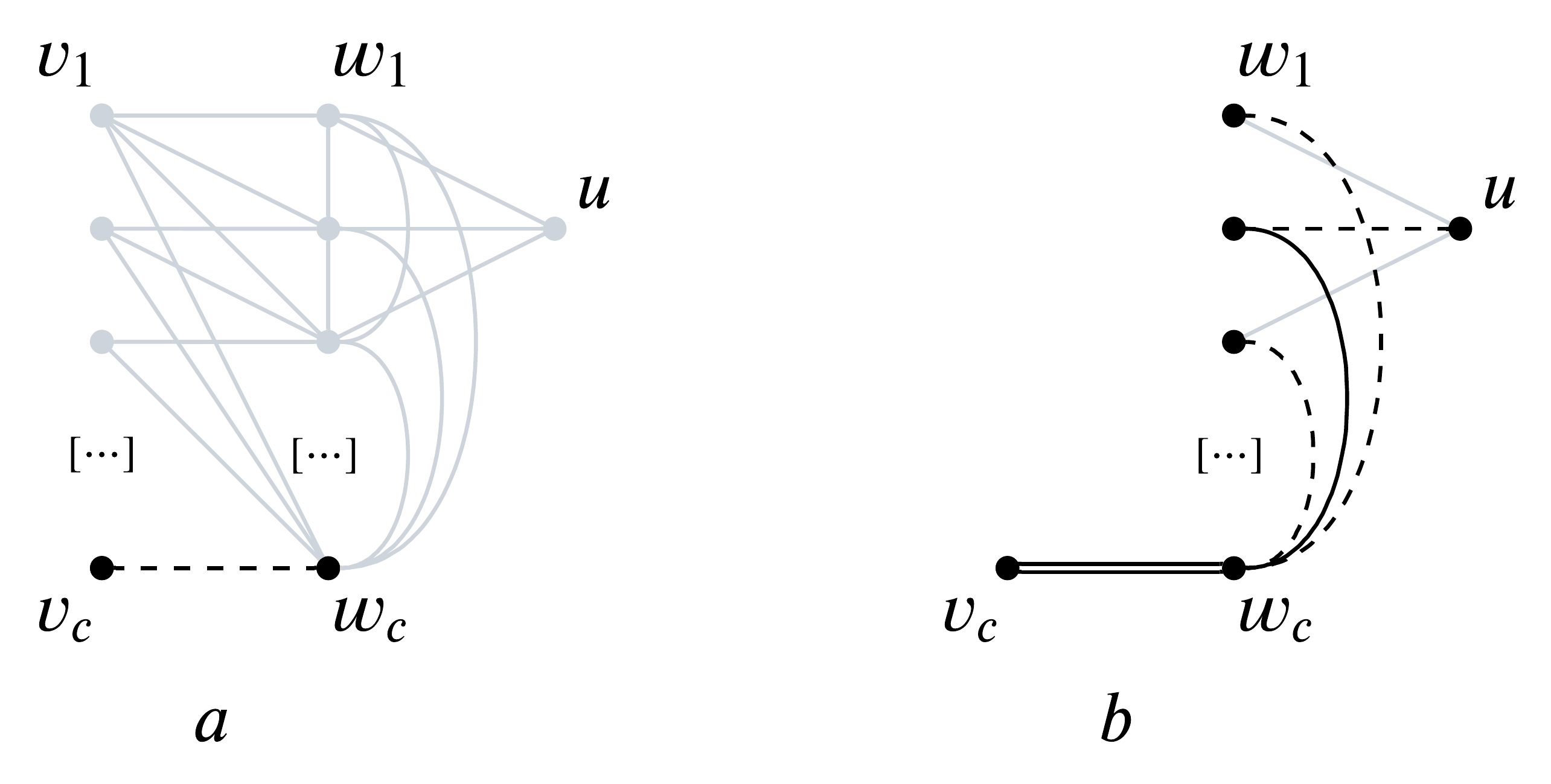}
     \caption{If $G$ has a $\overline{H^*}$, then $v_c-w_c-w_i-u$ is a $P_4$ for $1\leq i\leq (c-1)$.  We either: (a) delete $v_c w_c$; (b) conserve it.  Then for each $i$ with $1 \leq i \leq c - 1$, we must delete $w_c w_i$ or $w_i u$, resulting in branching vector $(1)+(c-1)^{2^{c-1}}$.}\label{fig:fig11}
   \end{minipage}\hfill
\end{figure}
This shows the first part of the lemma.

To see that this allows branching factor $2 + \epsilon$, we apply the first statement of Proposition~\ref{prop:branch-factors} but use $c - 2$ instead of $c$.  In this manner we can choose $c$ large enough so that for any $d \geq c - 2$, the branching vector $(1, 1) + (d)^{2^d}$ has branching factor at most $2 + \epsilon$.  
In particular, the branching vectors $B_{c-1} = (1, 1) + (c-1)^{2^{c-1}}$ and $B_{c-2} = (1, 1) + (c-2)^{2^{c-2}}$ achieve this.  
It then suffices to observe the following: $(1, 1, c - 1)$ and $(1) + (c-1)^{2^{c-1}}$ are both dominated by $B_{c-1}$; $(1, 1, c - 2)$ and $(1) + (c-2)^{2^{c-2}}$ are both dominated by $B_{c-2}$; and $(1, 1) + (c - 1)^{2^{c-2}}$ is dominated by $B_{c-1}$.
\end{proof}

\subsection{Easy chains}

We now assume that $G$ does not contain any of the specific induced subgraphs, so that by Theorem~\ref{thm:chudnov} it contains a long induced chain.  It turns out that the most complicated chains are long induced path graphs or their complement.  We start with the other, easier chains.
Recall that we use a binary representation of our chains, and it will be useful to find a forced pattern in those.  In the proofs that follow, for a string $B$, we write $B[i..j]$ to denote the substring of $B$ from position $i$ to $j$, inclusively.

\begin{lemma}\label{lem:binary-pattern}
Let $B$ be a binary string of length $d$ for some integer $d \geq 6$ that is a multiple of $3$.
Then $B$ contains either a substring with $d/3$ consecutive $0$s, a substring with $d/3$ consecutive $1$s, or one of the substrings $[0101], [001], [011]$.
% At least one of the chain subgraphs $[x10101]$, $[xxx001]$, $[xxx011]$, $[x11111]$ and $[x00000]$ may be found within any chain graph.
\end{lemma}

% \begin{proof}
% Split the string $B$ into three equal parts $B_1 = B[1..d/3], B_2 = B[d/3+1..2d/3], B_3 = B[2d/3+1..d]$.
% If either of these substrings contains only $0$s or only $1$s, our statement holds, so assume otherwise.
% Let $j$ be the first position of $B$ at which a $0$ occurs, noting that $j \leq d/3$.  If position $j + 1$ has a $0$, 
% we know that there is a $1$ in a later position (otherwise $B_3$ would only have $0$s), and in this case the substring $[001]$ must occur.  So position $j + 1$ has a $1$.  Consider position $j + 2$, which exists and is in either $B_1$ or $B_2$ since $j \leq d/3$ and $d \geq 6$.  If it is a $1$, then the substring $011$ occurs, so we assume that it has a $0$.  
% Then position $j + 3$ has a $0$ or else we have $[0101]$.  
% So $B[j + 2 .. j + 3] = [00]$. 
%  But we know that some $1$ must appear after position $j + 3$, as otherwise $B_3$ has only $0$s, and we have a $[001]$.
% % Per definition of a chain, vertices noted $0$ or $1$ as defined at the beginning of this section can be in any given order and count. 
% % Thus, a chain may consist of
% % infinitely repeating $1$s or $0$s, which are both covered by subgraphs $[x00000]$ and $[x11111]$. 
% % The $[x10101]$ subgraph covers the case where both alternates after each occurrence. If they do not, then $[xxx011]$ occurs as soon as $1$ repeats, or $[xxx001]$ occurs as soon as $0$ repeats. All possible chains of infinite length will therefore contain at least one of the stated chain subgraphs.
% \end{proof}

\begin{proof}
We may assume that $B$ has a $0$ within the first $2d/3 - 1$ positions, as otherwise $B$ starts with at least $d/3$ consecutive $1$s and we are done.
Let $j \leq 2d/3-1$ be the position of the first $0$ within $B$. 
%If $B[1..2d/3-1]$ or $B[j..d]$ (we note that $d-j \geq d/3+1$) 
If $B[j..d]$ contains only $0$s or only $1$s, then our initial statement holds, so we assume this is not the case. %Assuming otherwise, $B[1..2d/3-1]$ must contain at least a $0$ and there exists a value of $j$ such as $j$ is the position of the first $0$ within $B$.
If $B[j..d]$ begins with a series of two $0$s or more, then since $B[j..d]$ contain at least a $1$, the substring $[001]$ occurs. Otherwise, position $j+1 \leq 2d/3$ is a $1$ and we either find, starting from position $j$, $[011]$ or $[010]$. We are done if the first case occurs, and continue with the second, thus assuming that position $j + 2$ is $0$. 
Next note that position $j + 3$ exists since $j \leq 2d/3 - 1$ and $d \geq 6$.
If position $j+3$ is a $1$ then $[0101]$ occurs.
If instead position $j+3$ is a $0$, then either $B[j+2..d]$ only contains $0$s and has length $d - (j + 2) + 1 \geq d - (2d/3 - 1) - 1 \geq d/3$ and the statement still holds, or $B[j+2..d]$ has a $1$ and $[001]$ occurs. 
\end{proof}

We now look for unavoidable patterns in the chains.  Note that for simplicity, we require an initial chain with $4c$ vertices, because we will focus on a subchain with $c$ vertices.  This is for convenience, as we can ensure that $4c$ is large enough by increasing $C$ appropriately.

\begin{lemma}\label{lem:forced-subchain}
    Suppose that $G$ has an induced chain $P$ with $4c \geq 8$ vertices.  Then $G$ has a chain with $c$ vertices whose binary representation either has only $0$s or only $1$s, or ends with one of $[001], [011], [0101]$. 
\end{lemma}

\begin{proof}
    It suffices to argue that if $B$ is the binary representation of $P$, then $B$ contains a substring of length $c$ with only $0$s, only $1$s, or ends with one of $[001], [011], [0101]$.

    To that end, let $B' = B[c + 1..4c]$ be the substring of $B$ with the last $3c$ characters and apply Lemma~\ref{lem:binary-pattern} on $B'$.  
    If $B'$ has a substring of $3c/3 = c$ consecutive $0$s or consecutive $1$s, we have found the chain required by the lemma.  
    Otherwise, there is an occurrence of $[001], [011]$, or $[0101]$ in $B'$.  Let $j \geq c + 1$ be the position of $B$ at which this substring occurs in $B$, with $j$ the last position of the substring.  Then the substring $B[j - c + 1 .. j]$ has length $c$ and ends with one of $[001], [011], [0101]$.
\end{proof}

We next handle easy chains.

\begin{lemma}\label{lem:easy-chains}
Suppose that $G$ contains an induced chain with $c$ vertices whose binary representation ends with $[0101], [001]$, or $[011]$.  
Then one can achieve a branching vector among 
$(1, 1)+(c-3)^{2^{c-3}}$, $(1, 1)+(c-4)^{2^{c-4}}$ and $(1, 1)+(c-4)$.

Consequently, for any constant $\epsilon > 0$, there is a large enough constant $c$ such that if $G$ contains one of these induced chains, one can achieve branching factor at most $2 + \epsilon$.

% The chain subgraphs $[x10101]$, $[xxx001]$ and $[xxx011]$ branching vectors are all among the following : $(1)^2+(c-3)^{2^{c-3}}$, $(1)+(c-4)^{2^{c-4}}$ and $(1)^2+(c-4)$.
\end{lemma}

\begin{proof}
First note that if $G$ contains an induced chain with $c$ vertices, we can find it by brute-force enumeration of the $n^{O(c)}$ induced subgraphs of $G$.  For each subgraph, checking if it is a chain and computing its binary representation takes constant time, so the desired chain can be found in polynomial time.  We then return a constant number of deletion sets as follows.

Let $P$ be the induced chain with $c$ vertices whose binary representation ends with $[0101], [001]$, or $[011]$.   We provide Figures~\ref{fig:fig12} through~\ref{fig:fig14} as proof.
Recall that a vertex of type 1 is only adjacent to its immediate predecessor, and a vertex of type 0 to all predecessors except the immediate one. 

\begin{figure}[H]
   \begin{minipage}{1\textwidth}
     \centering
     \includegraphics[width=\linewidth]{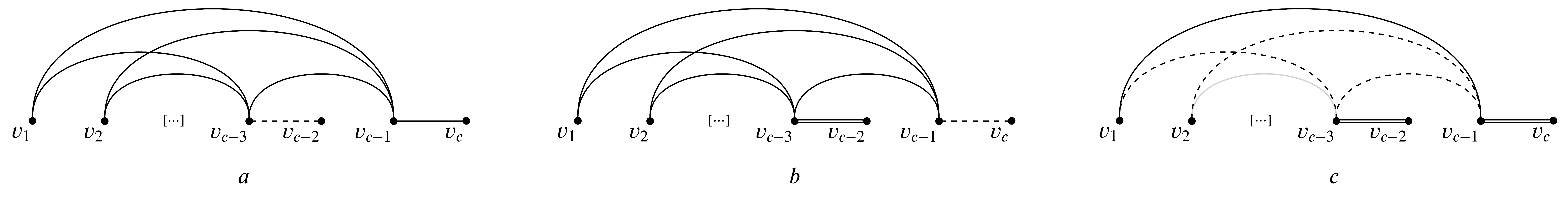}
     \caption{
     If $P$ ends with $[0101]$, corresponding to vertices $v_{c-3}, v_{c-2}, v_{c-1}, v_c$, then 
     $v_c - v_{c-1} - v_{c-3} - v_{c - 2}$ is a $P_4$.  We either (a) delete $v_{c-3} v_{c-2}$; (b) delete $v_{c-1} v_c$; (c) conserve both edges, in which case we must delete $v_{c-3} v_{c-1}$.  After doing so, $v_c - v_{c-1} - v_i - v_{c-3}$ is a $P_4$ for every $1 \leq i \leq c - 4$.  For each such $P_4$, we must delete either $v_i v_{c-3}$ or $v_i v_{c-1}$.  We add to our deletion sets all $2^{c-4}$ combinations, which each delete $c -4 + 1 = c - 3$ edges.  This results in branching vector $(1, 1) + (c - 3)^{2^{c-4}}$.
     %Let $G$ be the $Chain$ $[x10101]$, then $v_c-v_{c-1}-v_{c-i}-v_{c-3}-v_{c-2}$, where $1\leq i\leq (c-5)$, branches on $2+2^{c-5}$ prospects, resulting in branching vector $(1)^2+(c-5)^{2^{c-4}}$.
     }\label{fig:fig12}
   \end{minipage}\hfill
\end{figure}
\begin{figure}[H]
   \begin{minipage}{1\textwidth}
     \centering
     \includegraphics[height=2cm]{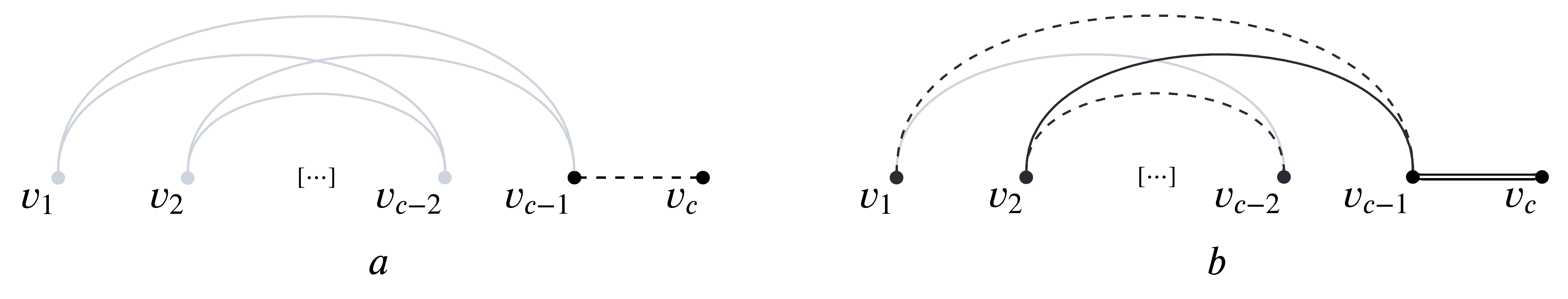}
     \caption{If $P$ ends with $[001]$, then $v_c-v_{c-1}-v_i-v_{c-2}$ is a $P_4$, where $1\leq i\leq c-4$.  We either: (a) delete $v_c v_{c-1}$; (b) conserve it.  Then for each $i$ with $1 \leq i \leq c- 4$, we must delete one of $v_i v_{c-2}$ or $v_i v_{c-1}$.  There are $2^{c - 4}$ possible combinations that each delete $c - 4$ edges, resulting in branching factor $(1)+(c-4)^{2^{c-4}}$.}\label{fig:fig13}
   \end{minipage}\hfill
\end{figure}
\begin{figure}[H]
   \begin{minipage}{1\textwidth}
     \centering
     \includegraphics[height=1.68cm]{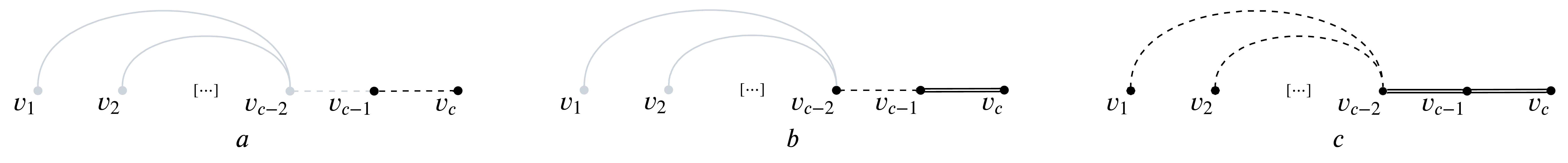}
     \caption{If $P$ ends with $[011]$, then $v_c-v_{c-1}-v_{c-2}-v_{i}$ is a $P_4$, where $1\leq i\leq c - 4$.  We either: (a) delete $v_{c-1} v_c$; (b) delete $v_{c-2} v_{c-1}$; (c) conserve both, in which case we must delete $v_i v_{c-2}$ for each $1 \leq i \leq c - 4$, resulting in branching vector $(1, 1, c - 4)$.}\label{fig:fig14}
   \end{minipage}\hfill
\end{figure}
To see that these achieve branching factor at most $2 + \epsilon$, we apply Proposition~\ref{prop:branch-factors} and take a large enough $c - 4$ that guarantees it (in particular, $(1, 1) + (c - 4)$ is dominated by $(1, 1) + (c - 4)^{2^{c-4}}$).
\end{proof}

% \begin{lemma}
% Branching vectors $(1)^2+(c-1)$, $(1)+(c-1)^{2^{c-1}}$, $(1)^2+(c-2)$, $(1)+(c-2)^{2^{c-2}}$, $(1)^2+(c-1)^{2^{c-2}}$, $(1)^2+(c-3)^{2^{c-3}}$, $(1)+(c-4)^{2^{c-4}}$ and $(1)^2+(c-4)$ all have a branching number of $2+\epsilon$.
% \end{lemma}

% \begin{proof}
% Since $c$ is a constant value which tends to infinity, to simplify, we will ignore any operation on $c$ if the result is constant. Once simplified, the listed branching vectors are limited to $(1)^2+(c)$, $(1)+(c)^{2^c}$ and $(1)^2+(c)^{2^c}$. In all three cases, when $c$ tends to infinity, $a$, the branching number or factor tends from infinity to $2$, which makes it $2+\epsilon$.
% \begin{figure}[!htb]
%    \begin{minipage}{1\textwidth}
%      \centering
%      \begin{equation*}
%         \begin{split}
%             a^k=2a^{k-1}+a^{k-c}\implies a=a^{1-c}+2&\\
%             a^k=a^{k-1}+2^c\cdot a^{k-c}\implies a=2^c\cdot a^{1-c}+1&\qquad c\rightarrow\infty\implies a\rightarrow2\\
%             a^k=2a^{k-1}+2^c\cdot a^{k-c}\implies a=2^c\cdot a^{1-c}+2&
%         \end{split}
%      \end{equation*}
%    \end{minipage}\hfill
% \end{figure}
% \end{proof}

\subsection{Chains that are complements of path graphs}

We now focus on the case where $G$ has an induced chain with $c$ vertices whose binary representation has only $0$s. 

\begin{lemma}\label{lem:chain-complement}
    Suppose that $G$ contains an induced chain with $c \geq 6$ vertices whose binary representation has only $0$s.  Then one can achieve branching vector $(c - 3, c - 3) + (1, 2, 3, \ldots, 2c - 6)$, whose branching factor is at most $2 + \epsilon$ for large enough $c$.
\end{lemma}

\begin{proof}
    As done previously, finding the desired chain takes time $n^{O(c)}$ and we shall return a constant number of deletion sets.  Let $P$ be the set of vertices on the induced chain. 
    Note that performing edge deletions on a graph to make it $P_4$-free is equivalent to performing edge additions on the complement to make it $P_4$-free, because $P_4$ is self-complementary.  We take the simpler latter view, since $\overline{G}$ is easier to work with.  
    Indeed, note that in the subgraph of $G$ induced by the vertices of the chain $P$, the only non-neighbor of a vertex of $P$ is its immediate predecessor on the chain.  Therefore, in $\overline{G}$, the chain $P$ is a path graph, that is, a chain with only vertices of type $1$.  We therefore return a safe set of edge additions on the path graph induced by $P$ on the complement of $G$.

    Denote the path induced by $P$ in $\overline{G}$ by $x_1 - x_2 - \ldots - x_c$.

Note that any optimal addition set of $\overline{G}$ either inserts no edge incident to $x_1$ whose other end is in $P$, or it does. In the latter case, we look at the minimum $i \in \{3, 4, \ldots, c\}$ such that $x_1 x_i$ is inserted.  After considering the case $x_1 x_4$, another choice of $i \neq 4$ enforces a structure that we can exploit further.  Let us look at the cases more closely.

\begin{figure}[H]
    \centering
    \includegraphics[width=0.7\linewidth]{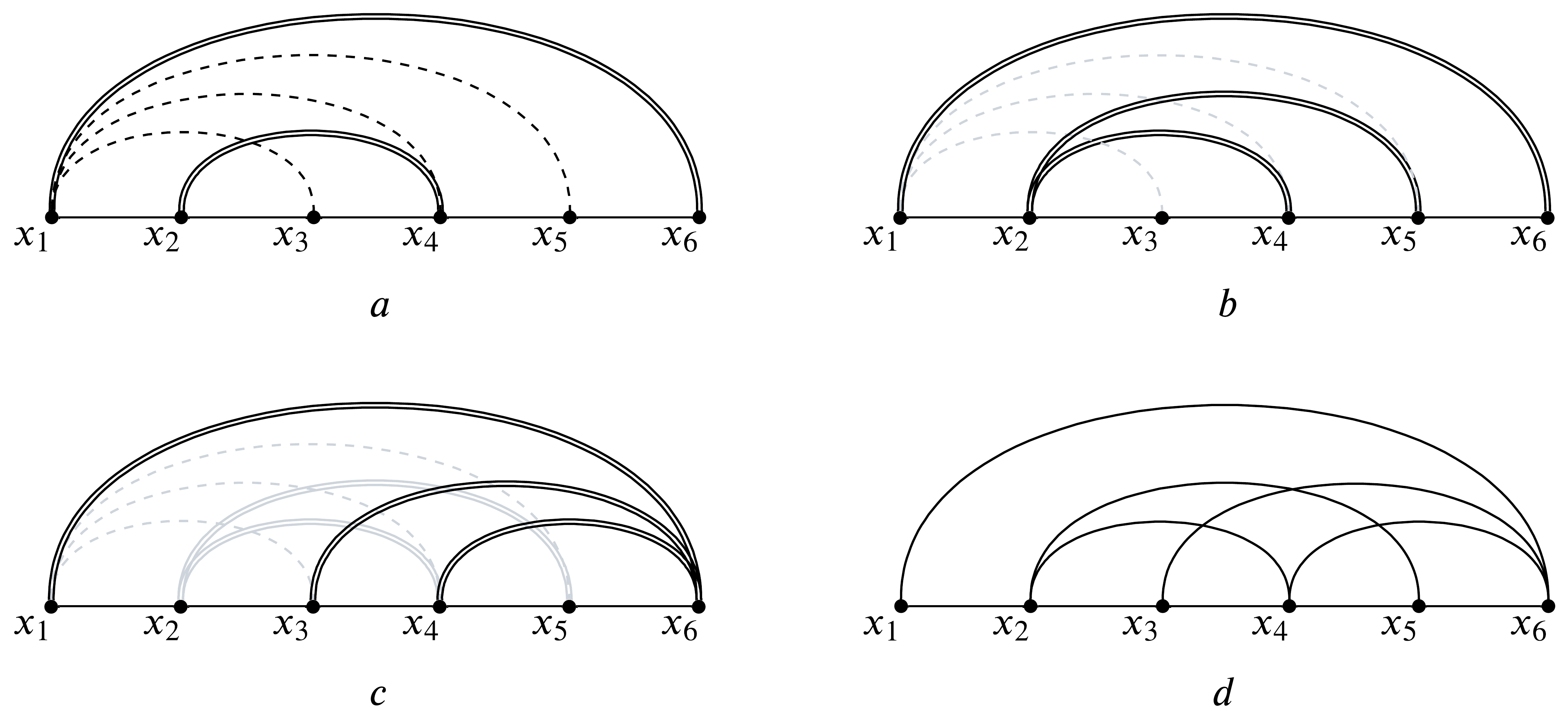}
    \caption{Illustration of Case 3 when $x_i = 6$ (here, dotted edges indicate edges that are assumed to not be inserted, and double lines edges that are assumed to be inserted, or forced).  In (a), if assuming that the only edge inserted incident to $x_1$ is $x_6$, then we must insert $x_2 x_4$.  This creates the $P_4$ $x_1 - x_2 - x_4 - x_5$, and as shown in (b) the insertion of $x_2x_5$ is forced.  There is still the $P_4$ $x_1 - x_6 - x_5 - x_4$, and in (c) $x_6 x_4$ is inserted, which in turn enforces $x_6 x_3$.  Subfigure (d) shows the final result.}
    \label{fig:chain-complement-insert}
\end{figure}

\begin{enumerate}

\item 
We insert no edge incident to $x_1$.  In that case, inserting $x_2 x_4$ is forced because of the $P_4$ $x_1 - x_2 - x_3 - x_4$.  This insertion creates the $P_4$ $x_1 - x_2 - x_4 - x_5$, and 
$x_2 x_5$ is forced.  This creates the $P_4$ $x_1 - x_2 - x_5 - x_6$, and 
$x_2 x_6$ is forced.  This goes on until all $x_2 x_i$ edges are forced for $i = 4 .. c$, so this case implies $c - 3$ insertions.

\item 
We insert $x_1 x_4$.  That's $1$ insertion.

\item 
In the rest of the cases, we assume that we insert some edge incident to $x_1$, but never $x_1 x_4$.  We branch over the first neighbor acquired by $x_1$ by an edge insertion, i.e.,  we branch over the possible $x_i$'s, where $x_i$ is such that $x_1 x_i$ is inserted but none of $x_1 x_3, x_1 x_4, \ldots, x_1 x_{i-1}$.  Note that $x_i \neq x_4$.
We save the case $x_i = x_3$ for last, and notice the following:
\\
If $x_i = x_5$, then there is the $P_4$ $x_1 - x_2 - x_3 - x_4$ that enforces inserting $x_2 x_4$.
There is also the $P_4$ $x_1 - x_5 - x_4 - x_3$ which enforces inserting $x_5 x_3$.  This case has $3$ insertions.
\\
If $x_i = x_6$, then there is the $P_4$ $x_1 - x_2 - x_3 - x_4$ that enforces inserting $x_2 x_4$.  Then there is $x_1 - x_2 - x_4 - x_5$ that enforces inserting $x_2 x_5$.
There is also the $P_4$ $x_1 - x_6 - x_5 - x_4$ which enforces $x_6 x_4$.  Then there is $x_1 - x_6 - x_4 - x_3$ enforcing $x_6 x_3$.  This case has 5 insertions.
\\
Generalizing for any $i \geq 5$, we see that  $x_2 x_4$ must be inserted, which enforces $x_2 x_5$, which in turn enforces $x_2 x_6$, and so on, so that all of $x_2 x_4, x_2 x_5, \ldots, x_2 x_{i-1}$ are inserted.  Likewise, after inserting $x_1 x_i$, $x_i x_{i-2}$ is forced, followed by $x_i x_{i - 3}$, and so on until all of $x_i x_3, x_i x_4, \ldots, x_i x_{i-2}$ are inserted.  The number of insertions is $2i - 7$.

So far, the insertion sets described correspond to the branching vector $(c - 3) + (1) + (3, 5, 7, \ldots, 2c - 7)$.  It remains to handle the case $x_i = x_3$.  We count $1$ insertion for $x_1 x_3$, but we cannot stop there, since appending a $(1)$ to this branching vector does not yield our desired branching factor.  

However, after inserting $x_1 x_3$ in $\overline{G}$, the sequence
$x_1 - x_3 - x_4 - \ldots - x_c$ forms an induced path with $c - 1$ vertices.  We repeat the same three previous cases as we just did, but over this path.  Denoting by $x_j$ the first neighbor of $x_1$ acquired on this new path, this yields the following sub-branches:
\begin{itemize}
    \item $x_j$ does not exist: this case makes $c - 1 - 3 + 1 = c - 3$ insertions ($c - 1$ is the number of vertices on the path, $-3$ is the same as case 1 from before, and $+1$ to count the insertion $x_1 x_3$).

    \item 
    $x_j = x_4$: this case is excluded: recall that if we get here, we have established that $x_1 x_4$ is not inserted. 

    \item 
    $x_j = x_5$: counting the insertion $x_1 x_3$, that's $2$ insertions.

    \item 
    $j > 5$.  Using the same arguments as in case 3 above, we have $2(j - 1) - 7$ forced insertions.  Counting $x_1 x_3$, that's $2(j - 1) - 6$ insertions.
\end{itemize}

Hence, the insertion sets built when $x_i = x_3$ correspond to branching vector $(c - 3) + (2) + (4, 6, 8, \ldots, 2c - 6)$.
No more cases are needed.

\end{enumerate}

The final branching vector we get is $(c - 3, 1, 3, 5, 7, \ldots, 2c - 7) + (c - 3, 2, 4, 6, 8, \ldots, 2c - 6)$, which is $(c - 3, c - 3) + (1, 2, 3, \ldots, 2c - 6)$.  

Now, this branching vector has the form required by Proposition~\ref{prop:branch-factors}, second statement, by writing $(1, 2, \ldots, 2c - 6) + (1/2 \cdot (2c - 6))^2$, where $\alpha = 1/2, \beta = 0, \gamma = 2$, and where the $c$ from the proposition is replaced with $2c - 6$.  Consequently, for large enough $c$ we can obtain branching factor at most $2 + \epsilon$.
\end{proof}

\subsection{Chains that are path graphs}

We now focus on the last possible unavoidable subgraph that prime graph $G$ may contain, which are chains of $1$s, that is, long induced paths.  This is the most complicated case, by far.  
On the easier side, one possibility is that $G$ is itself a path graph, that is, there are no vertices outside of the chain.  Even with vertex weights, handling this case in polynomial time is a standard exercise which will also be useful later on.

\begin{lemma}\label{lem:path-dp}
    Suppose that $G = (V, E, \omega)$ is a (vertex-weighted) path graph.  Then the \textsc{Cograph Deletion} problem can be solved in time $O(n)$, where $n = |V|$.
\end{lemma}

\begin{proof}
    We use dynamic programming over the path, traversing it from left to right.  Denote the path of $G$ by $x_1 - x_2 - \ldots - x_n$. 
    For $i \in [n]$, we define $\F(i)$ as the minimum cost of edge deletions needed to transform $G[\{x_1, \ldots, x_i\}]$ into a cograph.
    We define $\F(0) = 0$.  

    As a base case, we have $\F(1) = \F(2) = \F(3) = 0$.  
    For $i > 3$, it is easy to see that $\F(i)$ satisfies the recurrence
    \begin{align*}
        \F(i) = \min \begin{cases}
            \F(i - 1) + \omega(x_{i-1} x_i) \\
            \F(i - 2) + \omega(x_{i-2} x_{i-1}) \\
            \F(i - 3) + \omega(x_{i-3} x_{i-2}).
        \end{cases}
    \end{align*}
    To justify, observe that at least one of $x_{i-1} x_i, x_{i-2} x_{i-1}$, or $x_{i-3} x_{i-2}$ must be deleted to eliminate the $P_4$ formed by $x_{i-3} - x_{i-2} - x_{i-1} - x_i$.  Each recurrence entry corresponds to one of these choices, and after the deletion is performed, we must make the remaining subpaths a cograph, hence the recurrences.  Computing $\F(n)$ only takes time $O(n)$.
\end{proof}

Now, let $P$ be a vertex set of $G$ that induces a path with $3c$ vertices, and denote this path by $x_1 - x_2 - \ldots - x_{3c}$.  We use $3c$ instead of $c$ here mostly for convenience: we often look at thirds of $P$ because they are long enough, and it is easier to state that these subpaths have $c$ vertices instead of saying they have $c/3$ vertices all the time.  
We can ensure that $3c$ is large enough by adjusting $C$.  These details will be dealt with in the proof of the main theorem later on.

A concept that will useful several times is branching ``around'' an edge.  
That is, let $x_i$ be a vertex of $P$ such that 
$3 \leq i \leq 3c - 3$.  Then \emph{branching around} the edge $x_i x_{i+1}$ consists of building the four deletion sets corresponding to the following possibilities:\\
(a) delete $x_i x_{i+1}$;\\
(b) conserve $x_i x_{i+1}$ and delete both $x_{i-1} x_i, x_{i+1} x_{i+2}$;\\
(c) conserve both $x_i x_{i+1}$ and $x_{i-1} x_i$, enforcing the deletion of $x_{i-2} x_{i-1}$ and $x_{i+1} x_{i+2}$;\\
(d) conserve both $x_i x_{i+1}$ and $x_{i+1} x_{i+2}$, enforcing the deletion of $x_{i-1} x_i$ and $x_{i+2} x_{i+3}$.

The notion of ``branching around'' is illustrated in Figure~\ref{fig:twolight} (ignore the vertices $u$ and $v$).
One can easily confirm that this produces a safe set of deletion sets, since any $P_4$-free deletion set must perform one of those four cases.  The corresponding branching vector is $(1, 2, 2, 2)$.  This achieves branching factor less than $2.303$.  
We could thus stop here and get time $O^*(2.303^k)$, which is the same as Tsur's algorithm~\cite{Tsur:[3]}, although our algorithm's branching cases can be verified by hand without relying automated branching enumeration.  To improve on that complexity, we must work a bit more.

From now on, denote by $N(P)$ the set of neighbors of $P$ outside of $P$, that is, $N(P) = \left(\bigcup_{i \in [3c]} N(x_i) \right) \setminus P$ (where $N(x_i)$ denotes the set of neighbors of $x_i$ in $G$).  
We shall assume that $N(P)$ is non-empty, as otherwise $G$ is a path graph and we can use Lemma~\ref{lem:path-dp}.
Importantly, notice that under the assumption that $N(P) \neq \emptyset$, then some vertex $v \in N(P)$ must have a non-neighbor in $P$.  Indeed, if every vertex in $N(P)$ is a neighbor to all of $P$, then all possible edges between $P$ and $N(P)$ are present and $P$ is a non-trivial module of $G$, contradicting that it is prime.  
We denote $U = \{u \in N(P) : u$ has at least one non-neighbor in $P\}$, so that $U \neq \emptyset$.  

To finish the setup, for $j, l \in [2c+1]$ we denote $P_{j, l} = \{x_j, x_{j+1}, \ldots, x_{l}\}$, which is the set of vertices in the subpath of $P$ from positions $j$ to $l$ (which is empty if $l < j$).

All the subsequent statements use the notions of $G, P, N(P), U$, and $P_{j,l}$'s as we just defined them, and assume that $U$ is non-empty.

\begin{lemma}\label{lem:multiconditions}
    For any $\epsilon > 0$, there is a large enough $c$ such that if any of the following statement holds, then one can achieve branching factor at most $2 + \epsilon$:
    \begin{enumerate}
        \item 
        there is a subpath $P_{j, j + c - 1}$ with $2 \leq j \leq 2c$ such that each vertex of $P_{j, j + c - 1}$ has degree $2$ in $G$;

        \item 
        some vertex of $U$ has at least $c/3$ neighbors in $P$.

        \item 
        some edge $x_i x_{i+1}$ has weight at least two, with $3 \leq i \leq 3c - 3$.
        
    \end{enumerate}
\end{lemma}

\begin{proof}
    Note that if $P$ exists, it can be found in time $n^{O(c)}$.  Moreover, checking whether any of the conditions above holds can be done in polynomial time (in particular, computing $U$ takes time $O(n + m)$).
    
    First, suppose that $G$ has the subpath $P_{j, j + c - 1}$, denoted $x_j - x_{j+1} - \ldots - x_{j+c-1}$, containing only degree $2$ vertices of $G$.  Then, any solution must delete at least one of the three edges of the $P_4$ $x_j - x_{j+1} - x_{j+2} - x_{j+3}$.  Likewise, it must delete one of the edges of $x_{j+c-4} - x_{j+c-3} - x_{j+c-2} - x_{j+c-1}$.  Assume that $c$ is large enough so that $x_{j+3}$ and $x_{j+c-4}$ are distinct.
    We add to our set of deletion sets all $3 \cdot 3$ ways of choosing a deletion in the first $P_4$ and a deletion in the second, which amounts to $9$ deletion sets that each make two deletions.  

    After applying any of those deletion set, because the subpath $P_{j, j + c - 1}$ has only degree $2$ vertices, the graph resulting from the two deletions has a connected component that is a path with at least $c - 6$ vertices.  We can safely solve that connected component optimally, without branching.   For that, we use the linear time algorithm from Lemma~\ref{lem:path-dp}.  Any solution on that path deletes at least $\lfloor (c - 6)/3 \rfloor$ edges, which is a lower bound on the weight of deleted edges.  This number of deletions happens for each of the 9 branches described above, resulting in branching vector $(\lfloor (c - 6)/3 \rfloor)^9$ or better.  One can check that for $c \geq 18$, the corresponding branching factor is less than $2 + \epsilon$ (it actually tends to $1 + \epsilon$).

    Second, suppose that some $u \in U$ has at least $c/3$ neighbors in $P$.  By the definition of $U$, $u$ has a non-neighbor $x_i \in P$.  Assume that $x_{i+1}$ is a neighbor of $u$ (this is without loss of generality: if the first neighbor of $u$ after $x_i$ is some $x_{i+j}$, then redefine $x_i$ as $x_{i+j-1}$, and if such an $x_{i+j}$ does not exist, then reverse the path).  
    Notice that $u$ has at least $c/3 - 3$ neighbor in $P$ that are not neighbors of $x_i$ nor $x_{i+1}$ in $G$ (minus $3$ because we exclude $x_{i-1}, x_{i+1}, x_{i+2}$).  Let $Z$ be the set of neighbors of $u$ on $P$ that are not neighbors of $x_i$ nor $x_{i+1}$.
    Consider the building deletion sets using the following cases: 
    (a) delete $x_i x_{i+1}$; (b) delete $x_{i+1} u$; 
    (c) conserve both edges.  In that case, $x_i - x_{i+1} - u - z$ is a $P_4$ for each $z \in Z$, and we must delete each edge $uz$.  This results in branching vector $(1, 1, c/3 - 3)$, and for large enough $c$ we can achieve branching factor $2 + \epsilon$.   

    Third, suppose that $\omega(x_i x_{i+1})$ has weight $2$ or more, with $3 \leq i \leq n - 3$.  Then, if we branch around $x_i x_{i+1}$ as defined above, the first branch which deletes $x_ix_{i+1}$ reduces the parameter by at least $2$.  The branching vector is therefore $(2, 2, 2, 2)$ or better, whose branching factor is $2$.
\end{proof}

We actually achieve a much better upper bound on the number of neighbors in $P$ that a vertex $u \in U$ could have. 

\begin{lemma}\label{lem:three-nbrs}
    If some vertex $u \in U$ has at least three neighbors in $P_{3, 3c - 2}$, then one can achieve branching factor at most $2 + \epsilon$ with a large enough $c$.
\end{lemma}

\begin{proof}
    Denote $P' = P_{3, 3c - 2}$ for the duration of the proof.
    Let $u \in U$ and suppose that $u$ has at least three neighbors in $P'$.  By the second statement of Lemma~\ref{lem:multiconditions}, we may assume that $u$ has at most $c/3$ neighbors in $P'$, as otherwise we can achieve branching factor at most $2 + \epsilon$. 
    This implies that there are at least two consecutive vertices $x_{i-2}, x_{i-1}$ in $P'$ that are non-neighbors of $u$.  We assume that $x_i$ is a neighbor of $u$, and that $i \leq 3c - 2$ (this is without loss of generality: if $u$ has no neighbor in $P'$ after $x_{i-1}$, reverse the path so that it does, and then redefine $x_{i-2}, x_{i-1}$ to be the immediate predecessors of $x_i$ if needed).
    Note that the vertices $x_{i-3}$ through $x_{i+2}$ all exist since $x_{i-2}$ and $x_i$ are in $P'$.
    
    There two cases depending on whether $u x_{i+1}$ is an edge or not.  If it is not an edge, let $y, z$ be two neighbors of $u$ on $P$ other than $x_{i}$.  Figure~\ref{fig:three-nbrs} illustrates the possible deletion sets.  These lead to branching vector $(1, 3, 4, 3, 3, 4)$ and branching factor $2$.

    \begin{figure}[H]
        \centering
        \includegraphics[width=0.7\linewidth]{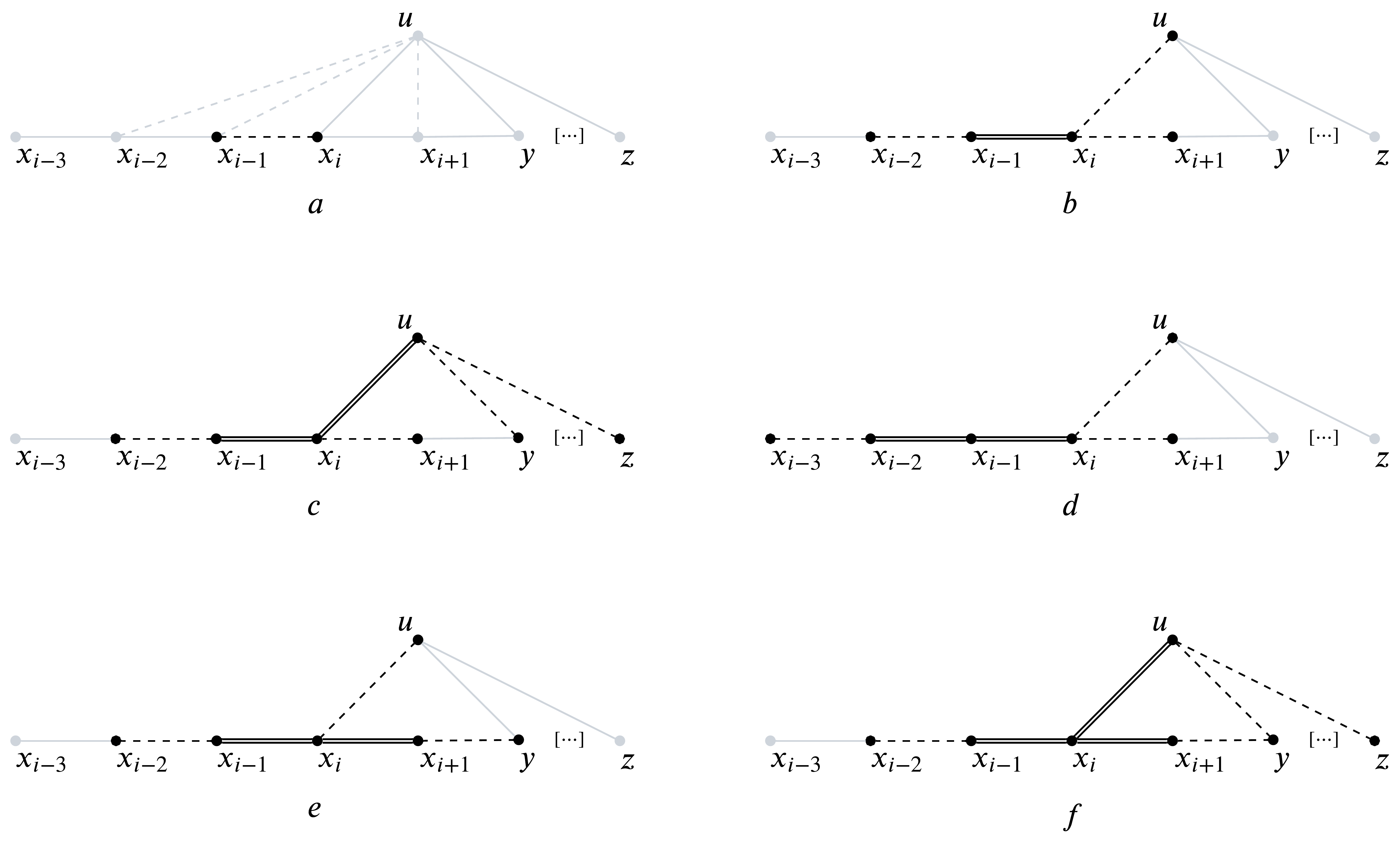}
        \caption{Cases when $u x_{i+1}$ is not an edge.  Here we illustrate $y = x_{i+2}$, but this does not need to be the case --- what matters is that $u$ does not have any of $\{x_{i-2}, x_{i-1}, x_{i+1}\}$ as neighbors. 
        We branch around $x_{i-1} x_i$, but also investigate subbranches to get better branching factors.  That is, either we: (a) delete $x_{i-1} x_i$; or conserve it for the remaining cases. Then either we (b) delete both $x_{i-2} x_{i-1}$ and $x_i x_{i+1}$.  In that case, either $x_i u$ is deleted or not, and the subfigure shows the case where it is deleted; (c) if $x_i u$ is conserved instead, then we must delete both $uy, uz$; (d) we conserve $x_{i-2} x_{i-1}$, in which case we must delete $x_{i-3} x_{i-2}, x_i x_{i+1}, x_i u$; (e) we conserve $x_{i} x_{i+1}$.  In that case, we must delete $x_{i-2} x_{i-1}$ and $x_{i+1} x_{i+2}$.  Then, $x_i u$ could be deleted or not, and the subfigure shows the case where it is deleted; (f) if $x_i u$ is conserved, we also reach a situation where we must delete $uy, uz$.}
        %\fs{[NOTE: ensure the lemma text and indices correspond to the graphic. Note that unnecessary non-existing edges are not shown beyond (a). Also, to unify the document, graphics references use (a,b,c,...) notation, so all numeric references need to be adjusted.]}}
        \label{fig:three-nbrs}
    \end{figure}

    If $u x_{i+1}$ is an edge, let $y$ be a neighbor of $u$ on $P'$ other than $x_i, x_{i+1}$ (note, $y = x_{i+2}$ is possible).  Figure~\ref{fig:three-nbrs-2} illustrates the cases.  These lead to branching vector $(1, 2, 4, 3, 4)$ and branching factor $2$.

    \begin{figure}[H]
        \centering
        \includegraphics[width=0.6\linewidth]{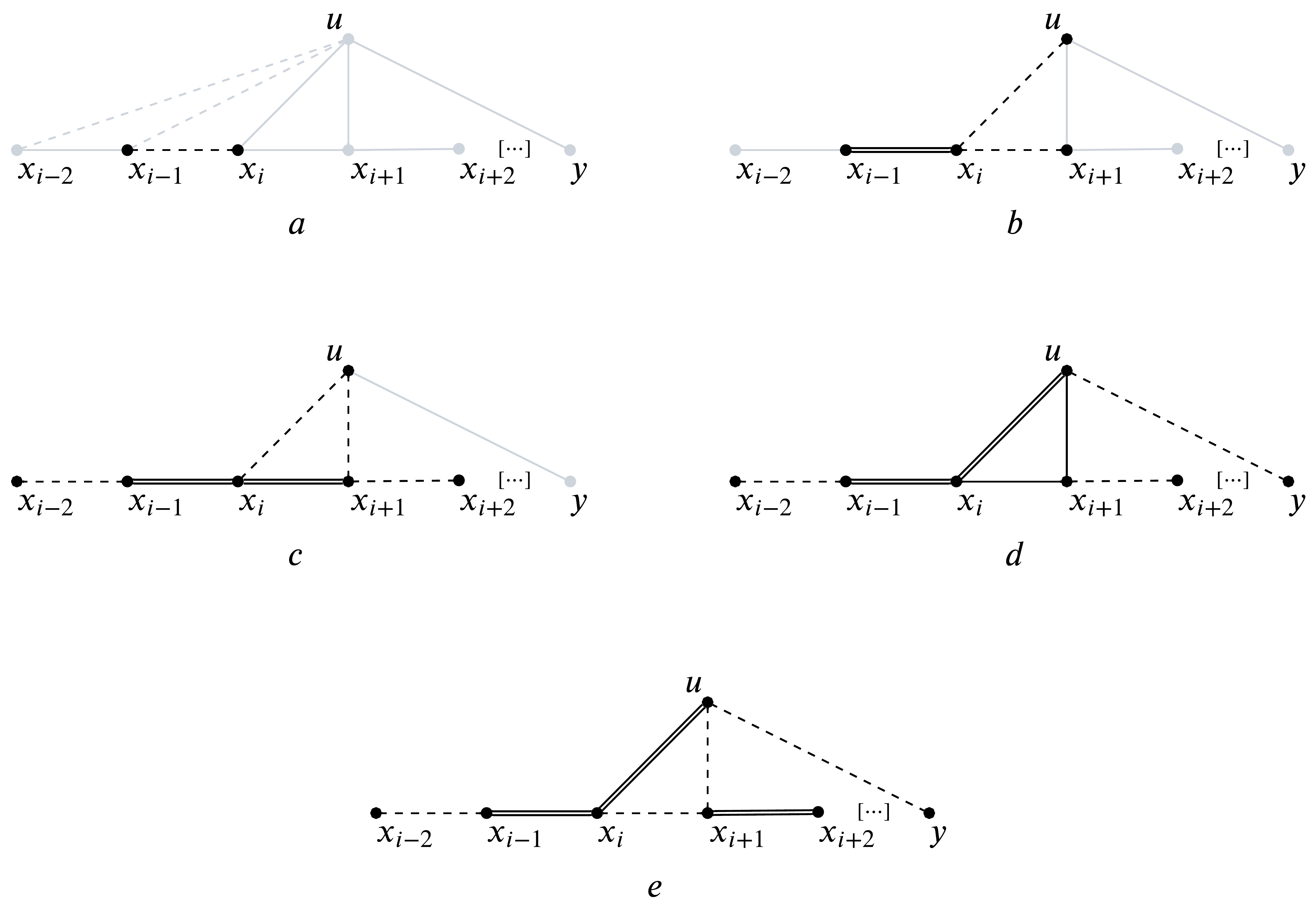}
        \caption{Branching cases when $u x_{i+1}$ is an edge.  We either: (a) delete $x_{i-1} x_i$; or conserve it in the remainder of the cases. Then either (b) delete $x_i u$ and then delete $x_i x_{i+1}$; (c) delete $x_i u$ but conserve $x_i x_{i+1}$.  In that case, $x_{i-2} x_{i-1}$, $x_{i+1} x_{i+2}$ must be deleted, as well as $x_{i+1} u$; (d) conserve $x_i u$ and delete $x_{i+1} x_{i+2}$.  Then $x_{i-2} x_{i-1}$ and $uy$ must be deleted (note, $uy$ must be deleted even if $y = x_{i+2}$); (e) conserve $x_i u$ and conserve $x_{i+1} x_{i+2}$ as well.  In that case, we cannot keep $x_i x_{i+1}$ and it is deleted.  This in turn enforces the deletion of $u x_{i+1}$ (because of $x_{i-1} - x_i - u - x_{i+1}$ created by the forced deletion of $x_i x_{i+1}$), plus the same two deletions as the previous case.  }
        \label{fig:three-nbrs-2}
    \end{figure}

    All possibilities lead to branching factor $2 + \epsilon$ or better.  One can easily check that the described modifications are safe, since any deletion set must perform at least one of the stated deletions, which concludes the proof.
\end{proof}

We next need to distinguish vertices of $P$ that have neighbors in $N(P)$ from those that do not.  We say that a vertex $x_i$ of $P$ is \emph{light} if it has no neighbor in $N(P)$, and otherwise $x_i$ is \emph{heavy}.  
Note, light vertices have degree two in $G$, except possibly $x_1$ or $x_{3c}$ which could be light and have degree one.
We look at possible sequences of light and heavy vertices, and show that they all lead to branching factor $2 + \epsilon$.

\begin{lemma}\label{lem:no-two-light}
    Suppose that $P$ has two consecutive light vertices $x_i, x_{i+1}$ where $4 \leq i \leq 2c$.  Then one may achieve branching factor at most $2 + \epsilon$ with a large enough $c$. 
\end{lemma}

\begin{proof}
    Assume that $x_i, x_{i+1}$ are light and, without loss of generality, that $x_{i+2}$ is an heavy vertex (there must exist a heavy vertex after $x_{i+1}$ as otherwise $P$ would have a long degree two subpath starting at $x_i$, and the bounds on $i$ allow us to invoke  Lemma~\ref{lem:multiconditions}).
    Likewise, we assume that every edge of $P$ has weight $1$ as otherwise we invoke~Lemma~\ref{lem:multiconditions}.  Also note that the vertices $x_{i-3}$ and $x_{i+3}$ exist.

    Suppose first that $x_{i-1}$ is heavy.  
    In that case, we branch around $x_{i} x_{i+1}$, with the definition of ``branching around'' described earlier.  When we consider the case that conserves the edges of the subpath $x_{i-1} - x_i - x_{i+1}$, the edges between $x_{i-1}$ and its neighbors in $U$ must be deleted, and that gives a deletion set of cost at least $3$ (see Figure~\ref{fig:twolight}).  The same occurs when we conserve $x_i - x_{i+1} - x_{i+2}$, and we get branching factor $(1, 2, 3, 3)$ or better, with branching factor $2$.

    \begin{figure}[H]
        \centering
        \includegraphics[width=0.65\linewidth]{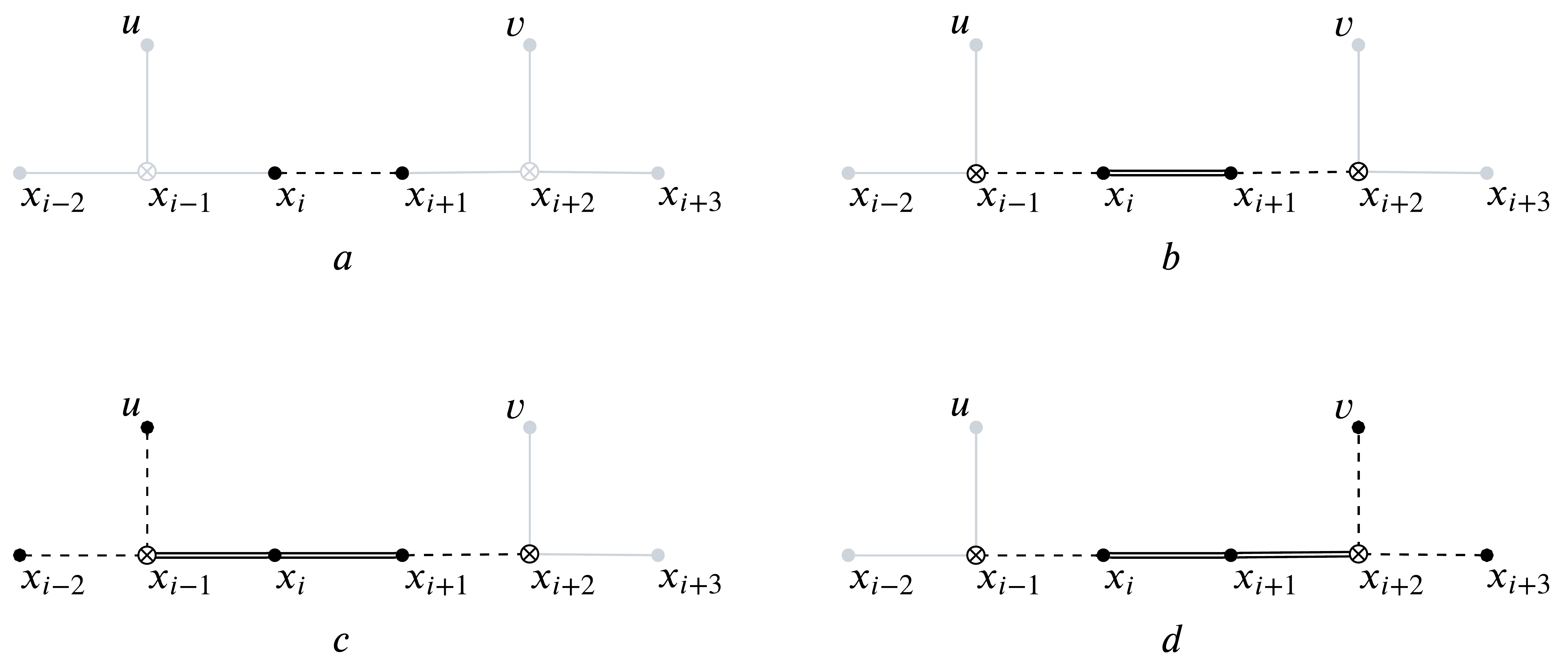}
        \caption{Branching when there are two consecutive light vertices $x_i, x_{i+1}$ and $x_{i-1}$ is heavy.  Light vertices are small filled circles and heavy vertices are large circles with a cross.}
        \label{fig:twolight}
    \end{figure}

    \begin{figure}[H]
        \centering
        \includegraphics[width=0.85\linewidth]{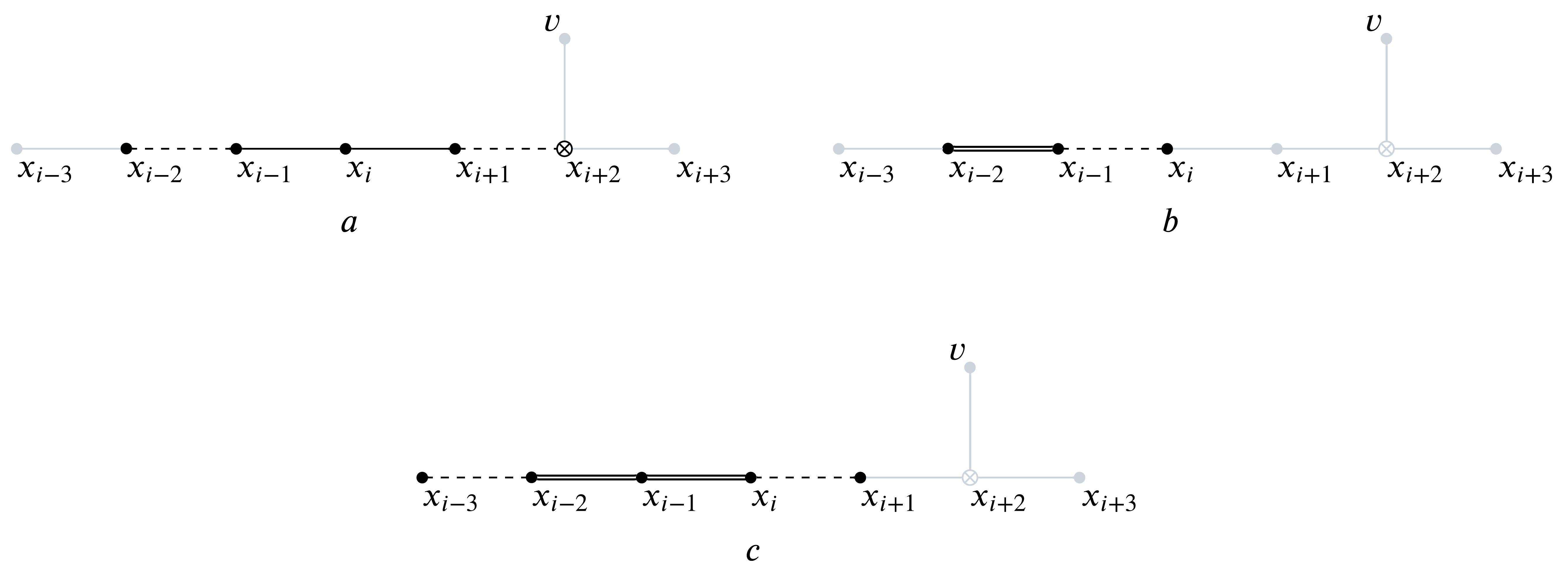}
        \caption{Three or more consecutive light vertices.}
        \label{fig:threelight}
    \end{figure}

    So suppose that $x_{i-1}$ is also a light vertex, as shown in Figure~\ref{fig:threelight} (it doesn't matter whether $x_{i-2}$ is light or heavy). 
 So either (a) we delete $x_{i-2} x_{i-1}$.  
    In the resulting graph $G'$, $x_{i-1} - x_i - x_{i+1} - x_{i+2}$ is a $P_4$ and only $x_{i+2}$ has neighbors outside of $P$.  
    In that case, we claim that there exists an optimal deletion set of $G'$ that deletes $x_{i+1} x_{i+2}$.  Indeed, if an optimal deletion set removes $x_{i-1} x_i$ or $x_i x_{i+1}$, we can reinsert that edge and delete $x_{i+1} x_{i+2}$ instead.  This cannot create new $P_4$s because $x_{i-1}, x_i, x_{i+1}$ are light and form their own connected component in the alternate graph.  Moreover, this alternate deletion set is also optimal because we assume that all edges of $P$ have weight 1.
    So in that branch, we can safely delete $x_{i+1} x_{i+2}$, which yields a deletion set of cost at least $2$.

    Next (b) we conserve $x_{i-2} x_{i-1}$ and delete $x_{i-1} x_i$; or (c) we conserve both $x_{i-2} x_{i-1}$ and $x_{i-1} x_{i}$, which enforces the deletion of $x_{i-3} x_{i-2}$ and $x_i x_{i+1}$.  
    This case results in branching vector $(2, 1, 2)$ and branching factor $2$.
\end{proof}

\begin{lemma}\label{lem:no-alt}
    Suppose that $P$ has vertices $x_{i-1}, x_i, x_{i+1}$ such that $x_{i-1}$ and $x_{i+1}$ are light and $x_i$ is heavy, where $4 \leq i \leq 3c - 2$. Then one can achieve branching factor at most $2 + \epsilon$ with a large enough $c$.  
\end{lemma}

\begin{proof}
    We assume that all edges of $P$ have weight $1$, or else invoke Lemma~\ref{lem:multiconditions}.  Note that vertices $x_{i-3}$ and $x_{i+2}$ exist.  We branch around $x_{i-1} x_i$, but we argue that we can skip the branch where we delete the two edges surrounding it, and only consider the cases a, b, c shown in Figure~\ref{fig:light-heavy-light-1}.  To be specific, we either delete $x_{i-1} x_i$, or conserve it. 
    Under the assumption that the edge is conserved, we argue that some optimal solution does not delete both $x_{i-2} x_{i-1}$ and $x_i x_{i+1}$.  This is shown in Figure~\ref{fig:light-heavy-light-2}.
    To see this, assume that $G'$ is a cograph obtained from $G$ after applying an optimal deletion set, such that both edges are deleted (but in which $x_{i-1} x_i$ is present).

    \begin{figure}[H]
        \centering
        \includegraphics[width=0.85\linewidth]{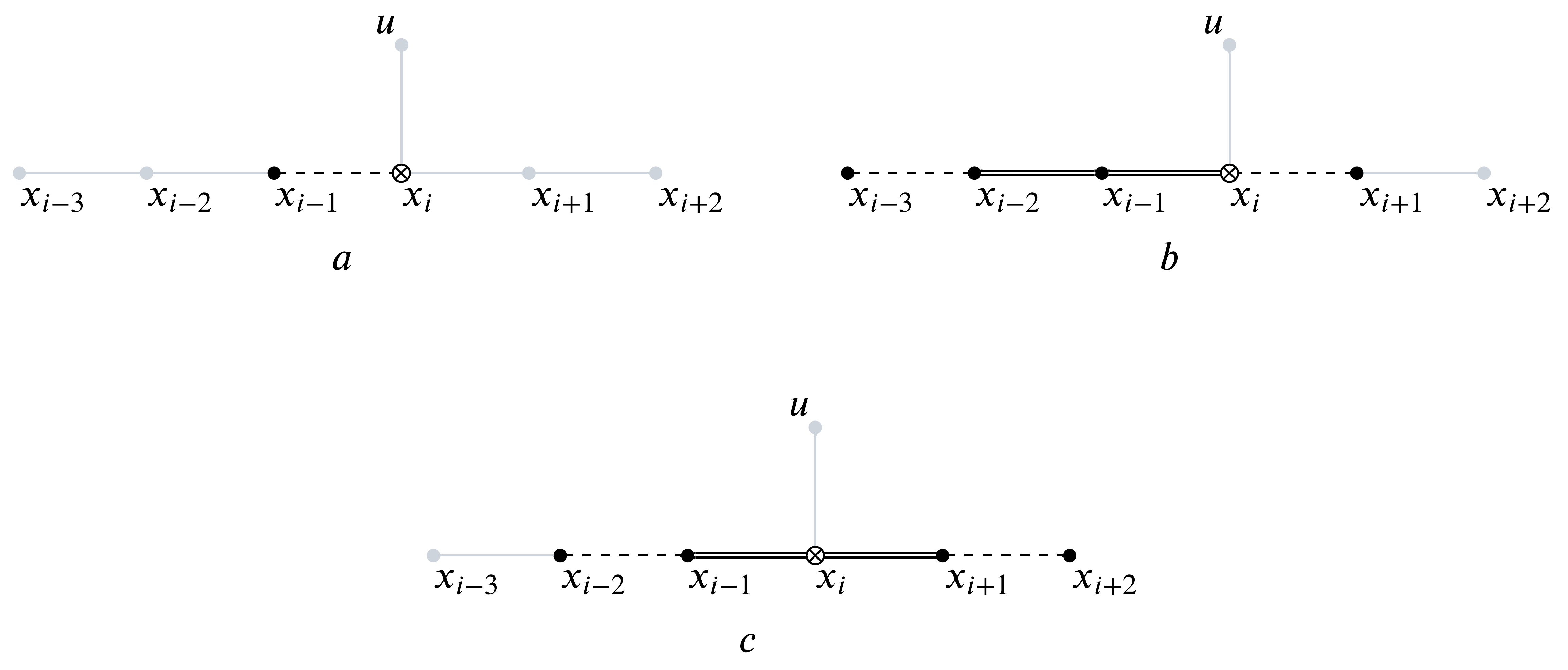}
        \caption{Cases for the light - heavy - light pattern.}
        \label{fig:light-heavy-light-1}
    \end{figure}

    \begin{figure}[H]
        \centering
        \includegraphics[width=0.85\linewidth]{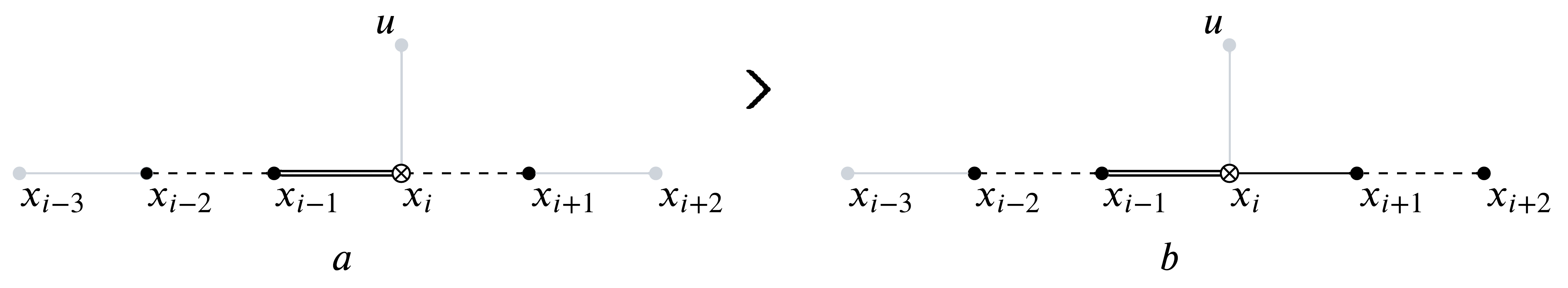}
        \caption{The deletion set with both edges around $x_{i-1} x_i$ being deleted is not necessary, since some alternate optimal solution as shown in (b) exists.}
        \label{fig:light-heavy-light-2}
    \end{figure}

    If $x_{i+1} x_{i+2}$ is deleted as well, consider the graph $G''$ obtained from $G'$ by adding $x_i x_{i+1}$ back in.  If $G''$ is not a cograph, then some $P_4$ contains that edge.  Then, because $x_{i+1}$ only has $x_i$ as a neighbor in that graph, the $P_4$ starts with the edge $x_{i+1} x_i$.  So that $P_4$ has the form $x_{i+1} - x_i - u - v$.  Moreover, $x_{i-1}$ also has only $x_i$ as a neighbor in $G''$, and we can obtain the $P_4$ $x_{i-1} - x_i - u - v$. 
    Since $x_{i+1}$ is not involved, that $P_4$ is also in $G'$, a contradiction.  The fact that we can add $x_i x_{i+1}$ back in contradicts the optimality of $G'$. 

    If $x_{i+1} x_{i+2}$ is not deleted in the optimal deletion set, then consider the alternate solution $G''$ obtained by adding $x_i x_{i+1}$ back in and deleting $x_{i+1} x_{i+2}$ (see Figure~\ref{fig:light-heavy-light-2}).  Again, $x_{i+1}$ and $x_{i-1}$ both have only $x_i$ as a neighbor. If $x_{i+1}$ is in some $P_4$ $x_{i+1} - x_i - u - v$, then $G''$ also has $x_{i-1} - x_i - u - v$.  In $G'$, adding back $x_{i+1} x_{i+2}$ cannot create a shortcut in this $P_4$, so that $P_4$ would be in $G'$, a contradiction.  Moreover, this alternate solution is also optimal since we assume that every edge of $P$ has weight $1$.  This proves our claim.

    So, we do not need to consider the case where $x_{i-1} x_i$ is conserved and its two surrounding edges are deleted. 
    Thus only the cases in Figure~\ref{fig:light-heavy-light-1} and needed, which results in branching vector $(1, 2, 2)$ and branching factor $2$.
\end{proof}

We reach the last possibility.

\begin{lemma}\label{lem:no-two-heavy}
    Suppose that $P$ has two consecutive heavy vertices $x_i, x_{i+1}$, where $c/2 \leq i \leq 3c/2$.
    Then one can achieve branching factor at most $2 + \epsilon$ with a large enough $c$.
\end{lemma}

\begin{proof}
    As a first case, assume  that $x_i, x_{i+1}$ each have distinct neighbors in $U$.  That is, $x_i$ has neighbor $u \in U$ and $x_{i+1}$ has neighbor $v \in U$, and $x_i v$ and $x_{i+1} u$ are non-edges.
    If neither $x_{i-1} v$ nor $x_{i+2} u$ is an edge, then we can do the branching around $x_i x_{i+1}$ as follows: (1) delete $x_i x_{i+1}$; (2) delete its two surrounding edges; (3) conserve $x_{i-1} x_i$ and $x_i x_{i+1}$, which enforces deleting two edges of $P$, plus the edge $x_{i+1} v$; (4) conserve $x_i x_{i+1}$ and $x_{i+1} x_{i+2}$, which enforces deleting two edges of $P$, plus the edge $x_{i} u$.  This has branching vector $(1, 2, 3, 3)$ and branching factor $2$.

    So we may assume that at least one of $x_{i-1} v$ or $x_{i+2} u$ is an edge.  Suppose that $x_{i-1} v \in E(G)$, bearing in mind that the case where $x_{i+2} u \in E(G)$ can be handled in an identical manner by a symmetric argument.
    Note that by Lemma~\ref{lem:three-nbrs}, we may assume that $v$ has at most two neighbors in $P_{3, 3c-2}$, and so $x_{i-1}, x_{i+1}$ are the only neighbors of $v$ on that subpath.  Since we have $c/2 \leq i \leq 3c/2$, we may assume that $c$ is large enough so that $v$ has no neighbors among $x_{i+2}, x_{i+3}, x_{i+4}$. 
    
    Given that, we have two possibilities depending on whether $u x_{i+2}$ is an edge or not.  Figure~\ref{fig:twoheavy-1} enumerates the deletions sets to return when $u x_{i+2}$ is an edge, and Figure~\ref{fig:twoheavy-2} when it is not an edge.  Both cases lead to branching factor $2$.

    \begin{figure}[H]
        \centering
        \includegraphics[width=0.7\linewidth]{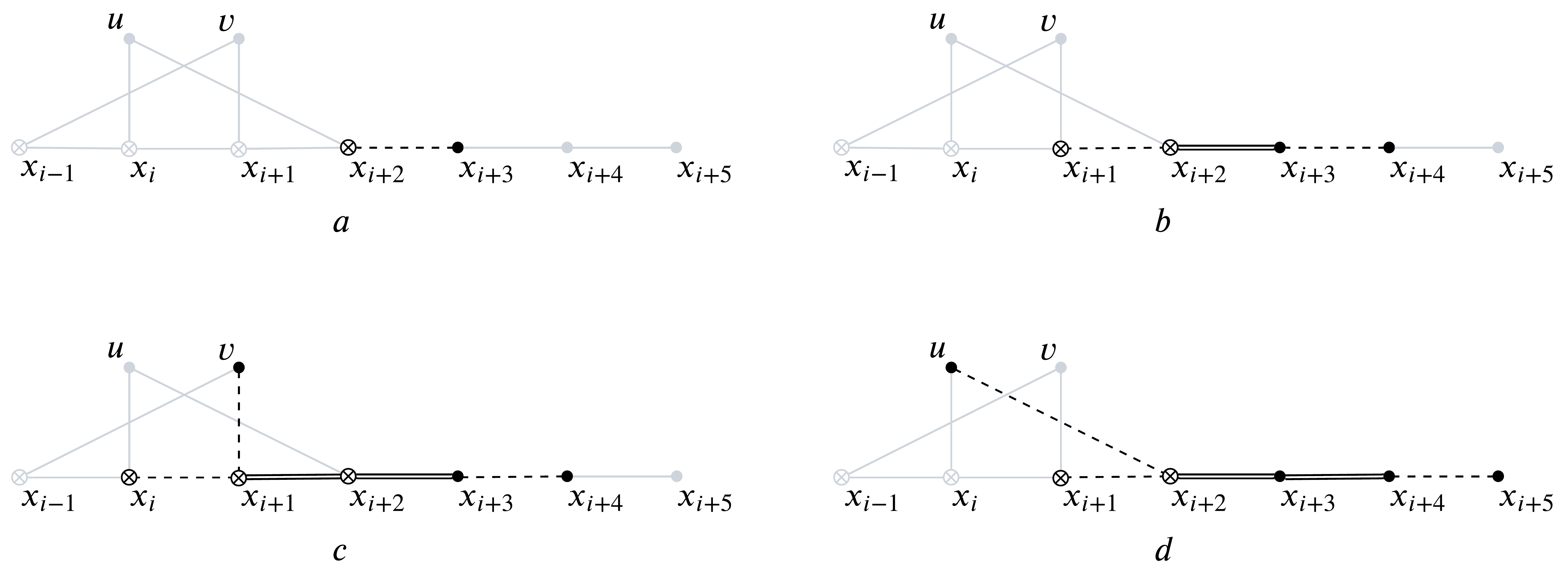}
        \caption{if $u x_{i+2}$ is an edge, we branch around $x_{i+2} x_{i+3}$.  Here, we assume that $u$ and $v$ have at most two neighbors in $P_{3, 3c - 2}$, so that they have no additional neighbors on the subpath illustrated.  So in the branching around, when we make the subpath $x_{i+1} - x_{i+2} - x_{i+3}$ conserved in (c), the edge $x_{i+1} v$ must be deleted.  Likewise in the last case when making $x_{i+2} - x_{i+3} - x_{i+4}$ conserved, $x_{i+2} u$ must be deleted.  This yields branching vector $(1, 2, 3, 3)$ and branching factor $2$. }
        \label{fig:twoheavy-1}
    \end{figure}

    \begin{figure}[H]
        \centering
        \includegraphics[width=0.7\linewidth]{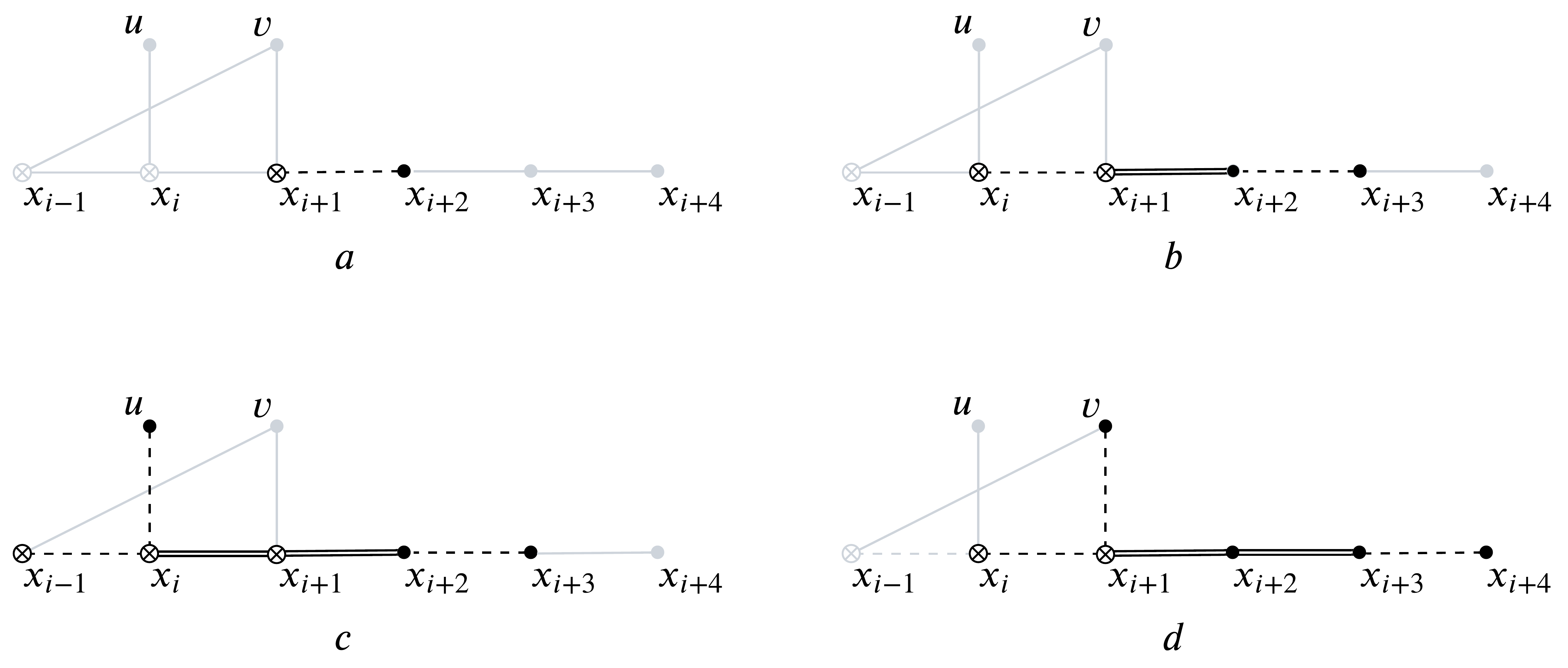}
        \caption{if $u x_{i+2}$ is not an edge, we branch around $x_{i+1} x_{i+2}$.  In this case, the main difference being that when $x_i - x_{i+1} - x_{i+2}$ is conserved we must delete $x_i u$.  This also yields branching vector $(1, 2, 3, 3)$ and branching factor $2$.
         }
        \label{fig:twoheavy-2}
    \end{figure}

    We may now assume that, in the range of subscripts from $c/2$ to $3c/2$, two consecutive heavy vertices do not each have a neighbor in $U$ of their own.  This implies that for two such vertices, one vertex has a neighborhood in $U$ that is a subset of the other. 
    In particular, two such consecutive vertices have at least one common neighbor in $U$.
    We next handle the case where there are three consecutive heavy vertices.
    So suppose that $x_i, x_{i+1}, x_{i+2}$ are heavy, where $c/2 \leq i \leq 3c/2 - 2$.  
    Then, $x_i, x_{i+1}$ have a common neighbor $u$ in $U$.  Also, $x_{i+1} x_{i+2}$ have a common neighbor $v$ in $U$.  We may assume that $u \neq v$, as otherwise $u$ would have three neighbors in $P_{3, 3c-2}$ and we can invoke Lemma~\ref{lem:three-nbrs}.  We can then branch around $x_{i+2} x_{i+3}$ as in Figure~\ref{fig:double-triangle}.

    \begin{figure}[H]
        \centering
        \includegraphics[width=0.65\linewidth]{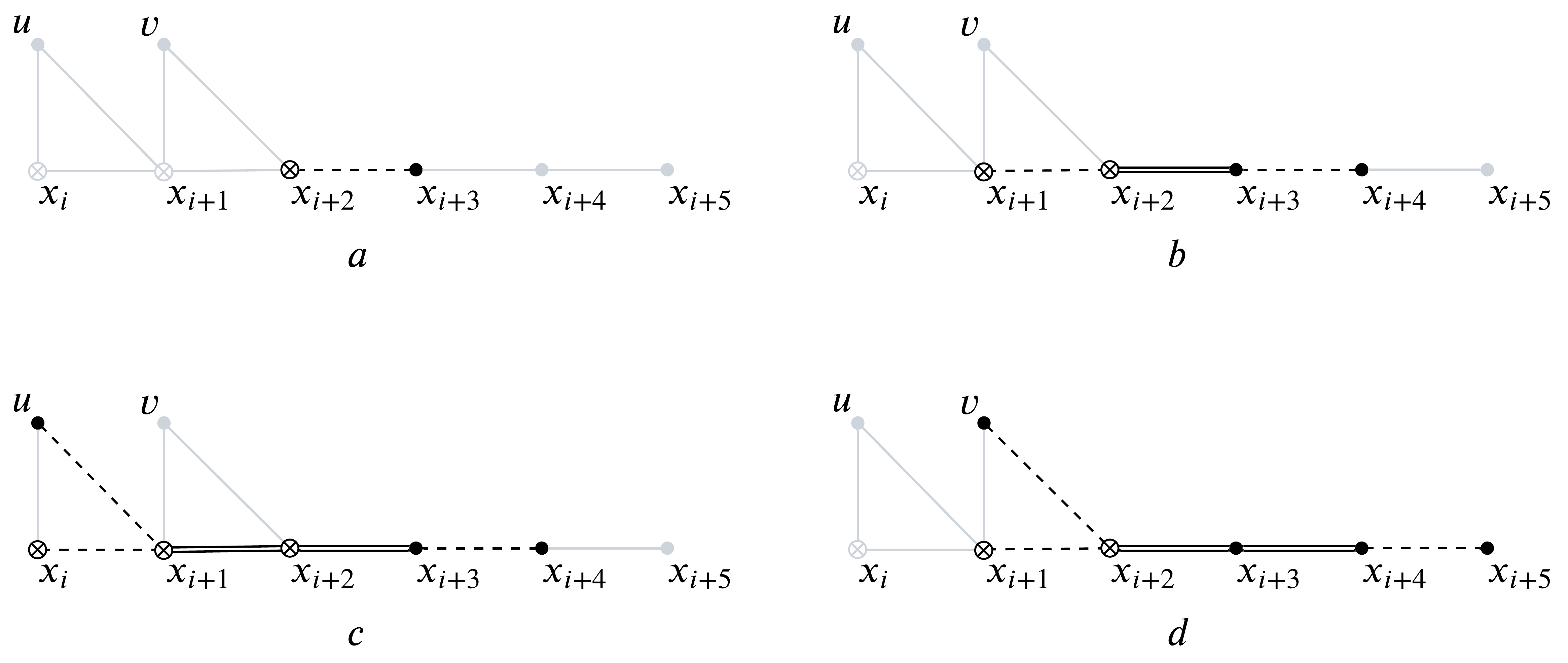}
        \caption{Building deletion sets when there are three consecutive heavy vertices, where we assume that $u$ and $v$ have only two neighbors among the vertices of $P$ illustrated (which will be the case if $c$ is large enough).  In (c), when we make the subpath $x_{i+1} - x_{i+2} - x_{i+3}$ conserved, $x_{i+1} u$ must be deleted.  In (d), when we make $x_{i+2} x_{i+3}$ conserved we must delete $x_{i+2} v$.  This yields branching vector $(1, 2, 3, 3)$ and branching factor $2$. 
        }
        \label{fig:double-triangle}
    \end{figure}

    So, we may assume that there are never more than two consecutive heavy vertices in the range of subscripts from $c/2$ to $3c/2$. 
    In particular, this means that there is $i \in \{c, c+1, c+2\}$ such that $x_i$ is a light vertex.  We can deduce that the structure illustrated in Figure~\ref{fig:last-effing-case} must occur (where the predecessor of $a$ in the figure plays the role of $x_{i+1}$) using the following observations:
    \begin{itemize}
        \item 
        The possible range of $i$ allows applying Lemma~\ref{lem:no-two-light}, so we may assume that there are not two consecutive light vertices, so $x_{i+1}$ is heavy.  
        \item 
        Then by Lemma~\ref{lem:no-alt}, we may assume that the pattern light - heavy - light does not occur, so $x_{i+2}$ is heavy as well.  By our assumption, $x_{i+1}$ and $x_{i+2}$ have a common neighbor in $U$.

        \item 
        There are no three consecutive heavy's, so $x_{i+3}$ is light.

        \item 
        Repeating the same logic, $x_{i+4}, x_{i+5}$ are heavy and have a common neighbor in $U$, and $x_{i+6}$ is light.
    \end{itemize}

    Since $i \in \{c, c+1, c+2\}$, with a large enough $c$ this chain of ``light $-$ heavy $-$ heavy $-$ light''  can be extended as long as desired, and the vertices $a, b, c, d, e, f$ illustrated in Figure~\ref{fig:last-effing-case} can be found (in particular, vertex $u$ neighbor of $a$ may be assumed to have no neighbor among $c, d, f$).  Since we may assume that the common neighbor of consecutive heavy vertices have two neighbors on $P_{3, 3c-2}$ and not more, we may perform the branching as illustrated.  This results in branching vector $(1,2,3,4,5)$ and branching factor $1.97$. 

    \begin{figure}[H]
        \centering
        \includegraphics[width=0.75\linewidth]{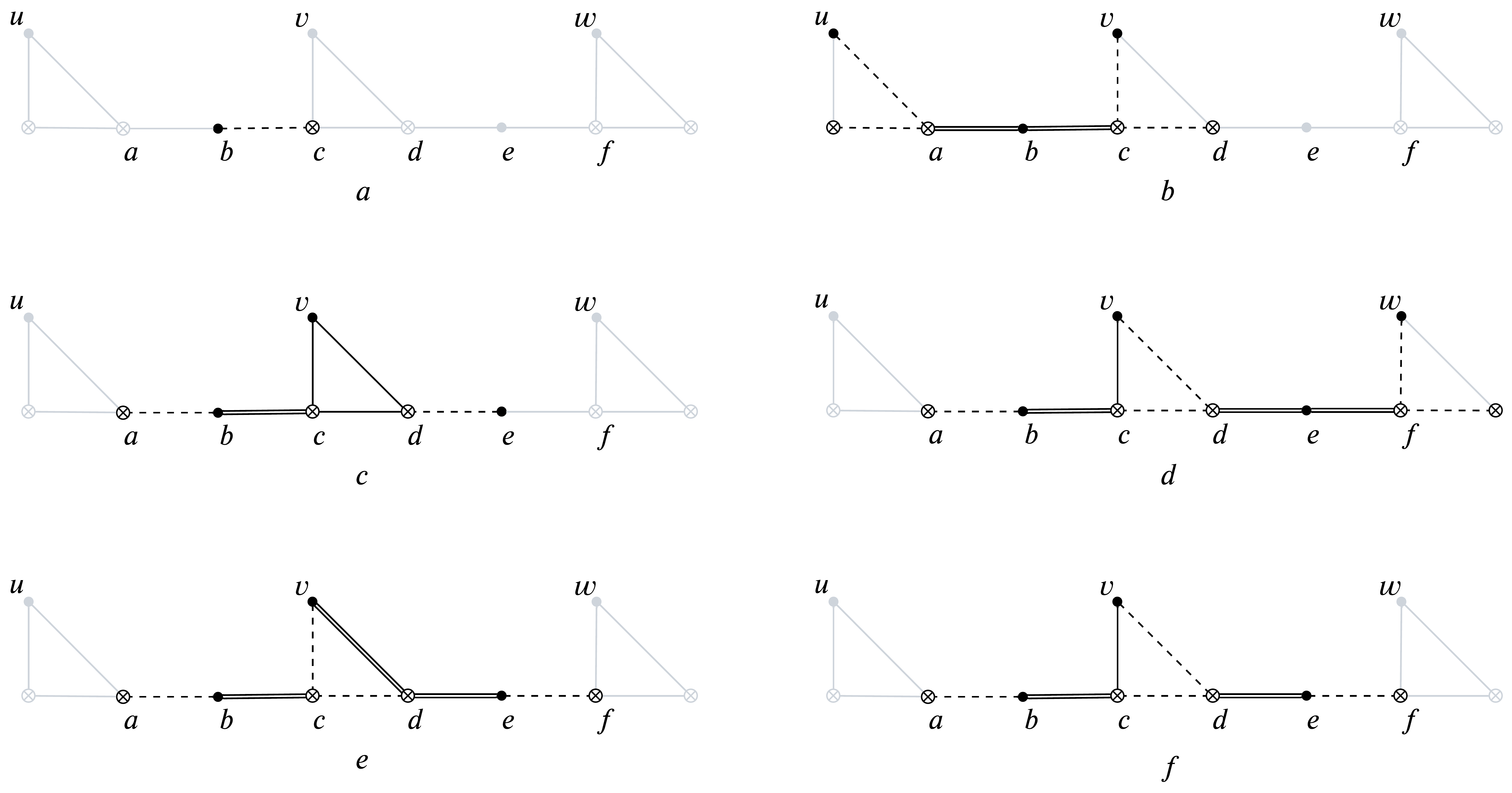}
        \caption{The ultimate case, where to ease reading we now use vertex names $a, b, c, d, e, f$, where $a, c, d, f$ are heavy and $b, e$ are light.   We assume that $u, v, w$ are vertices of $U$ that have exactly two neighbors among the vertices of $P$ illustrated.  
        We consider the following cases: (a) delete $bc$.  Then in what follows, we assume that $bc$ is conserved;  (b) conserve $ab$. This enforces the $4$ deletions shown. 
        In what follows, we assume that $bc$ is conserved and $ab$ is deleted;  
        (c) delete $de$, resulting in $2$ deletions.
        In what follows, we further assume that $de$ is conserved, and thus $cd$ must be deleted in all subsequent cases because of $b - c - d - e$; 
        (d) conserve $ef$.  This enforces a total of $5$ deletions as shown.  In what follows, we further assume that $ef$ is deleted;
        (e) conserve $dv$, where $v$ is a common neighbor of $c, d$ in $U$.  Given that $cd$ must be deleted, $cv$ must then be deleted because of $e - d - v - c$.  With $ab$ and $ef$ this amounts to $4$ deletions; 
        (f) delete $dv$.  With $ab, ef$ assumed to be deleted, plus $cd$ forced to be deleted, we have $4$ deletions.
        \\
        The result is branching vector $(1, 4, 2, 5, 4, 4)$ whose branching factor is less than $1.969$.}
        \label{fig:last-effing-case}
    \end{figure}

    This concludes the proof, since every possible situation leads to a safe deletion set and a branching factor of $2 + \epsilon$ or less.   
\end{proof}

The penultimate deduction we can make from all this cases analysis follows.

\begin{corollary}\label{cor:paths}
    For any constant $\epsilon > 0$, there is a large enough constant $c$ such that if $G$ contains an induced path with $3c$ vertices as an induced subgraph, then one can achieve branching factor at most $2 + \epsilon$. 
\end{corollary}

\begin{proof}
    If $G$ is itself a path graph, we use Lemma~\ref{lem:path-dp}. Otherwise, we focus on $x_c$ and its surrounding vertices (we choose that vertex because it far enough from the path's ends for all the lemmata to be applicable).  If both vertices $x_c$ and $x_{c+1}$ are heavy, we invoke Lemma~\ref{lem:no-two-heavy}.  Suppose that one of these two vertices, call it $x_i$, is light. If $x_{i+1}$ is light, we invoke Lemma~\ref{lem:no-two-light}.  So assume $x_{i+1}$ is heavy.  If $x_{i+2}$ is heavy, then we invoke Lemma~\ref{lem:no-two-heavy}.  
    Finally, if $x_{i+2}$ is light, we have the light - heavy - light pattern and we invoke Lemma~\ref{lem:no-alt}.
\end{proof}

We have gathered all the necessary ingredients to obtain our main result.

\begin{theorem}
    For any real number $\epsilon > 0$, the \textsc{Cograph Deletion} problem can be solved in time $O^*((2 + \epsilon)^k)$.
\end{theorem}

\begin{proof}
    Since \textsc{Cograph Deletion} is an instance of the $\HH$-free Modification problem with $\H$ a set of prime graphs, we may use the algorithm from Theorem~\ref{thm:algocorrect}.
    Recall that the running time of the algorithm is $O(n \cdot b^k (g(n) + h(n) + n + m))$, where $g(n)$ is the time to decide whether a graph $G$ is a cograph, $h(n)$ is the running time of the \textbf{branch} function, and $b$ is its branching factor.
    For cographs, $g(n) \in O(n^2)$ is possible (actually, $O(n + m)$, but we have not included $m$ in the arguments of $g$~\cite{habib2005simple}).
    Thus, we only need to analyze the running time of our \textbf{branch} function and its branching factor.

    Also recall that the algorithm needs a constant $C$ to know when to solve using a brute-force algorithm.
    We assume that $C$ is a large enough constant, so that any prime graph with at least $C$ vertices contains one of the unavoidable subgraphs from Theorem~\ref{thm:chudnov} with large enough value of $12c$.  Here, we use $12c$ because Lemma~\ref{lem:forced-subchain} then guarantees that, if $G$ does not contain one of the specific graphs, there is a chain with $3c$ vertices with one of the unavoidable patterns in the binary representation, and if this pattern is a path it needs to have $3c$ vertices for Corollary~\ref{cor:paths} to be applicable. 
    The other unavoidable subgraphs all apply if we replace $c$ with $12c$.  
    Moreover, large enough means that all the previous lemmata are applicable with the desired value of $\epsilon$.
    Theorem~\ref{thm:chudnov} guarantees that a suitable constant $C$ always exists, since the corresponding $12c$ only needs to be large enough with respect to the constant $\epsilon$.

    When \textbf{branch} receives a prime graph $G$ with $C$ vertices or more, we must find which unavoidable subgraph occurs.  We spend time $n^{O(c)}$ to brute-force search for one of the specific subgraphs from Lemma~\ref{lem:fixed-graphs}. 
 If one is found, we invoke the lemma, and otherwise we spend time $n^{O(c)}$ to find a chain with $12c$ vertices, which must exist by Theorem~\ref{thm:chudnov}.   If the chain's binary representation contains one of $[0101], [001], [011]$, then we use Lemma~\ref{lem:easy-chains}.  
    Otherwise, by Lemma~\ref{lem:forced-subchain} the binary representation has a long substring with only $0$s, in which case we use Lemma~\ref{lem:chain-complement}, or only $1$s, in which case we use Corollary~\ref{cor:paths}.  

    This shows that we can devise a function \textsc{branch} running in time $n^{O(c)}$ and returning a safe set of deletion sets for \textsc{Cograph Deletion}, with branching factor at most $2 + \epsilon$.  This proves that Algorithm~\ref{alg:alg} can be implemented to run in time $O^*((2 + \epsilon)^k)$.
\end{proof}

\section{Conclusion}

We conclude with some remarks on the complexity of the algorithm and possible future uses.  
The reader may notice that the polynomial factors hidden in the $O^*$ notation are significant.  
Recall that the graphs with less than $C$ vertices delegate to an unspecified \textbf{solve} routine.  Such small graphs may still be solved using past FPT algorithms, but this indicates that leveraging our new techniques requires extremely large graphs.  Moreover, the value of $C$ depends on the desired $\epsilon$ in an unspecified manner. 
 The dependency could be established clearly by a deeper analysis on the $C$ required for our desired $c$ and the exact branching vectors that we achieve, but such an analysis is probably not worth it at the moment.  The other probably more important bottleneck is the $n^{O(c)}$ complexity to find one of the unavoidable subgraphs from Theorem~\ref{thm:chudnov}.  
 This may be hard to avoid as this relates to the induced subgraph isomorphism problem, but on the other hand, some optimism may be justified by the facts that one of those structures is guaranteed to exist, and that we are dealing with prime graphs.  We leave this open, but the possibility should be investigated before the algorithm can become practical.  
 
 As for other uses of the general algorithm, let us recall that it is plausible that it can improve the complexity of the \textsc{Cograph Editing} problem.  On the other hand, the analysis for the simpler \textsc{Cograph Deletion} is already difficult, with many possible branching cases, so more clever views will be needed to make this extension interesting.  Let us also mention that improving the current kernel of \textsc{Cograph Editing} from $O(k^2 \log k)$ to $O(k^2)$ may be doable using our approach, since if the quotient graph is a large prime graph, we know it requires a large number of edge modifications, which in turn can be used to bound the size of the modules.  
 Let us mention in passing the \textsc{Cluster Deletion} problem, also known as \textsc{$P_3$-free Deletion}.  The $P_3$ is not prime and cannot make use of our algorithm.  However, the unavoidable structures all contain many $P_3$s, so in the case of prime graphs our approach may still lead to improved branchings, and perhaps non-prime graphs can be handled efficiently using other techniques.
 % The cluster editing problems may also be meaningful to look at in the near future.  The \textsc{Cluster Deletion} is also known as \textsc{$P_3$-free Deletion} and can make use of Algorithm~\ref{alg:alg}.  Although we believe that the reductions in~\cite{misra2020subexponential} exclude the possibility of a $O^*((1 + \epsilon)^k)$ time algorithm for this problem\footnote{This was pointed out by an anonymous reviewer from a previous version of the paper, who we thank.}, we note that the unavoidable subgraphs all contain many $P_3$'s, and most of them lead to an easy $(1+\epsilon)$ branching factor.  It will be interesting to see where the bottleneck is and whether the current FPT complexity can be improved.  To our knowledge, the best is currently $O^*(1.415^k)$ and has not been improved since 2011~\cite{bocker2011even}.  
 As for other $\HF$ editing problems, the unavoidable subgraphs from Theorem~\ref{thm:chudnov} may or may not contain a large number of graphs from $\H$.  For example, they do not all enforce modifications for the $P_6$-free editing problem, so it remains to find other properties of prime graphs that can be useful for a large number of $\HF$ graph classes.

\bibliographystyle{plain}
\bibliography{biblio}

\end{document}